%% file: bare_jrnl_onecolumn.tex
\documentclass[journal,onecolumn]{IEEEtran}

\input{preamble}
\input{header}

\begin{document}

\title{Multivariate Analytic Combinatorics for\\ Cost Constrained Channels}

\author{
	Andreas Lenz, \thanks{Andreas Lenz  was with the  Institute for Communications Engineering, Technical University of Munich, Germany.  He  was supported in part by  the European Research Council (ERC) under the European Union's Horizon 2020 research and innovation programme (grant agreement No 801434) (email: andreas\_lenz@gmx.de).}
	Stephen Melczer,\thanks{Stephen Melczer is with the Department of Combinatorics \& Optimization, University of Waterloo, 200 University Avenue West, Waterloo, ON N2L 3G1, Canada (email: smelczer@uwaterloo.ca).  He was supported by NSERC Discovery Grant RGPIN-2021-02382.}
	Cyrus Rashtchian,\thanks{Cyrus Rashtchian is with Google Research (email: cyroid@google.com). This work was partially completed while he was a postdoctoral scholar with the Department of Computer Science and Engineering and the Qualcomm Institute at the University of California, San Diego.}
	Paul H. Siegel\thanks{Paul H. Siegel is with the Department of Electrical and Computer Engineering, Center for Memory and Recording Research, University of California, San Diego, La Jolla, CA 92093 (email: psiegel@ucsd.edu).  He was supported in part by NSF Grant CCF-2212437.} 
	\thanks{This work, under the title ``Exact Asymptotics for Discrete Noiseless Channels,'' was presented in part at the IEEE International Symposium on Information Theory, Taipei, Taiwan, June 25-30, 2023, pp. 2494-2498 [DOI:10.1109/ISIT54713.2023.10206623].}
}

%\author{IEEE Publication Technology,~\IEEEmembership{Staff,~IEEE,}
        % <-this % stops a space
%\thanks{This paper was produced by the IEEE Publication Technology Group. They are in Piscataway, NJ.}% <-this % stops a space
%\thanks{Manuscript received April 19, 2021; revised August 16, 2021.}}

% The paper headers
%\markboth{Submitted to the IEEE Transactions in Information Theory}%
{}

%\IEEEpubid{0000--0000/00\$00.00~\copyright~2021 IEEE}
% Remember, if you use this you must call \IEEEpubidadjcol in the second
% column for its text to clear the IEEEpubid mark.

\maketitle

\begin{abstract}
Analytic combinatorics in several variables is a branch of mathematics that deals with deriving the asymptotic behavior of combinatorial quantities by analyzing multivariate generating functions. We study information-theoretic questions about sequences in a discrete noiseless channel under cost constraints. Our main contributions involve the relationship between the graph structure of the channel and the singularities of the bivariate generating function whose coefficients are the number of sequences satisfying the constraints. We use these new results to invoke theorems from multivariate analytic combinatorics to obtain the asymptotic behavior of the number of cost-limited strings that are admissible by the channel. This builds a new bridge between analytic combinatorics in several variables and labeled weighted graphs, bringing a new perspective and a set of powerful results to the literature of cost-constrained channels. Along the way, we show that the cost-constrained channel capacity is determined by a cost-dependent singularity of the bivariate generating function, generalizing Shannon's classical result for unconstrained capacity, and provide a new proof of the equivalence of the combinatorial and probabilistic definitions of the cost-constrained capacity. 
\end{abstract}


\input{introduction}
%
%
\input{preliminaries} 
%
%
\input{results}

%
%
\input{proof-outline}
%
%
\input{perron-frobenius}
%
%
\input{spectral-cost-diverse}
%
%
\input{multivariate-singularity-analysis}

%
%
\input{conclusion}
\appendices

\input{appendix}

\bibliography{refs}
\bibliographystyle{IEEEtran}

\end{document}

%% file: preamble.tex
\usepackage{amsmath,amsfonts,amsthm,amssymb,mathrsfs}
\usepackage{algorithmic}
\usepackage{algorithm}
\usepackage{array}
\usepackage{textcomp}
\usepackage{stfloats}
\usepackage{url}
\usepackage{verbatim}
\usepackage{graphicx}
\usepackage{cite}
 
\usepackage{tikz,pgfplots}
\usetikzlibrary{positioning,patterns,calc,backgrounds,matrix,decorations.markings,shapes,arrows,arrows.meta}
\usepgfplotslibrary{fillbetween}
 \pgfplotsset{compat=1.18}
\colorlet{lenzGray}{black!5!}
\tikzstyle{block} = [rectangle, draw, text centered, rounded corners,fill=lenzGray, minimum height=2em]
\tikzstyle{bullet} = [circle, fill=black, minimum width=0.05cm, inner sep=0.05cm]
\tikzset{>=latex}

\newtheorem{thm}{Theorem}[section]
\newtheorem{defn}[thm]{Definition}
\newtheorem{lemma}[thm]{Lemma}
\newtheorem{corollary}[thm]{Corollary}

\newtheorem{example}[thm]{Example}
\newtheorem*{example*}{Example}
\newtheorem{remark}[thm]{Remark}
\newtheorem*{remark*}{Remark}
\newtheorem{proposition}[thm]{Proposition}

\newcommand{\imu}{\mathrm{i}}

\newcommand{\V}{\mathcal{V}}
\newcommand{\E}{\mathcal{E}}
\newcommand{\lang}{\mathcal{L}}
\newcommand{\init}{\mathsf{init}}
\newcommand{\final}{\mathsf{term}}

\newcommand{\ve}[1]{\boldsymbol{#1}}
\newcommand{\adj}{\mathsf{adj}}
\newcommand{\rank}{\mathsf{rank}}
\newcommand{\tr}{\mathsf{tr}}

\newcommand{\N}{\mathbb{N}}
\newcommand{\R}{\mathbb{R}}

\newcommand{\DNAA}{\mathsf{A}}
\newcommand{\DNAC}{\mathsf{C}}
\newcommand{\DNAG}{\mathsf{G}}
\newcommand{\DNAT}{\mathsf{T}}

%% file: header.tex
\usepackage{comment}
\usepackage{amsmath,amssymb,amsthm,subfig}
\usepackage{pgfplots}
\usepackage{cite}
\usepackage{xcolor}
\usepackage{hyperref}
\definecolor{blueish}{HTML}{007F99}
\definecolor{purpleish}{HTML}{72177A}
\definecolor{blue}{HTML}{000099}
\hypersetup{
	colorlinks,
	citecolor=purpleish,
	filecolor=red,
	linkcolor=blueish,
	urlcolor=blue
}
\usepackage[noabbrev,nameinlink]{cleveref}
\usepackage[utf8x]{inputenc}
\usepackage[shortlabels]{enumitem}

\usepackage{tikz}
\usetikzlibrary{shapes,arrows,positioning,calc,patterns}
\tikzstyle{block} = [rectangle, draw, text centered, rounded corners, minimum height=2em,fill=black!5!]
\usetikzlibrary{arrows.meta}
\usepgfplotslibrary{fillbetween}

\tikzset{>=latex}

\newcommand{\br}{\boldsymbol{r}}

%% file: introduction.tex
%!TEX root = bare_jrnl.tex
\section{Introduction}
Since their introduction in Part I of Shannon's landmark 1948 paper, \textit{A Mathematical Theory of Communication}~\cite{shannon_mathematical_1948}, discrete noiseless channels have been an important subject of research for information theorists and coding theorists. They  have also found practical use in the design of transmission codes for digital communication systems and recording codes for data storage systems~\cite{marcus_introduction_2001}.

In this paper, we consider discrete noiseless channels under an  \textit{average cost constraint}. Such a constraint can arise from limitations on the transmission power in an optical fiber~\cite{djordjevic_coding_2010}, the programming  voltage in a non-volatile memory~\cite{liu_coding_2020, liu_rate-constrained_2020,liu_rate-constrained_2022}, or the synthesis time per nucleotide in a DNA-based storage system~\cite{lenz_coding_2020}. 

\subsection{Background}
We begin with some background on costly constrained channels and their capacity.

\subsubsection{Constrained channels with cost}
The labeled directed graph $G$ in Fig.~\ref{fig:synthesis} represents an example of a discrete noiseless channel describing the synthesis of DNA strands using the alternating synthesis sequence 
$\DNAA \DNAC \DNAG \DNAT ~\DNAA \DNAC \DNAG \DNAT \ldots$ (see~\cite{lenz_coding_2020}). 
The channel graph generates sequences of symbols over the alphabet $\Sigma = \{\DNAA, \DNAC,\DNAG,\DNAT\}$ by following  paths through the directed graph and reading off the symbols $\sigma(e) \in \Sigma$ labeling the edges $e$ in the path. Each edge $e$ also  has an associated positive weight or cost $\tau(e)\in {\mathbb N}$, denoting  the synthesis time of the edge label $\sigma(e)$. %(here, meaning its equivalent number of bit transmissions)
The edge labels and costs are shown in the figure as $\sigma(e)|\tau(e)$. The cost is assumed to be additive, so the cost of a sequence generated by a path in the graph is the sum of its edge costs.

Discrete noiseless channels in which all edges have unit cost are well studied \cite{marcus_introduction_2001}; here we are interested in the more general setting of varying edge costs, as in Fig.~\ref{fig:synthesis}.
	
	\begin{figure}[t]%{.5\columnwidth}
		\begin{center}
			\scalebox{1}{
				\begin{tikzpicture}
					%%		\draw[-{>[scale=1.5,
						%%          length=2,
						%%          width=3]},line width=0.4pt] (0,0) to (1,0);
					%%
					%%
					%%\draw[-{>[scale=1.5,
						%%          length=2,
						%%          width=6]},line width=0.4pt] (0,-1) to (1,-1);
					%%
					%%
					%%\draw[-{>[scale=1.5,
						%%          length=6,
						%%          width=3]},line width=0.4pt] (2,0) to (3,0);
					\node  (v0){};
					\node[block,right of=v0, text width =0.5cm,node distance=2cm] (v1) {$\DNAA$};
					\node[block, text width =0.5cm,right of=v1,node distance=4cm] (v2) {$\DNAC$};
					\node[block, text width =0.5cm,below of=v2,node distance=2.5cm] (v3) {$\DNAG$};
					\node[block, text width =0.5cm,left of=v3,node distance=4cm] (v4) {$\DNAT$};
					
					\draw[-{Triangle[scale=1.5,length=2, width=3]}] (v1) to [bend left=15] node[above] {$\DNAC|1$} (v2);
					\draw[-{Triangle[scale=1.5,length=2, width=3]}] ($(v1.east)+(0,-0.2)$) to [bend left=15] node[xshift=-.05cm,yshift=-.2cm,pos=.45,below] {2} (v3);
					\draw[-{Triangle[scale=1.5,length=2, width=3]}] (v1) to [bend left=15] node[right] {$\DNAT|3$} (v4);
					\draw[-{Triangle[scale=1.5,length=2, width=3]}] ($(v1.west)+(0,0.2)$) to [out=150,in = 210, looseness=5] node[left] {$\DNAA|4$} ($(v1.west)+(0,-0.1)$);
					
					\draw[-{Triangle[scale=1.5,length=2, width=3]}] (v2) to [bend left=15] node[right] {$\DNAG|1$} (v3);
					\draw[-{Triangle[scale=1.5,length=2, width=3]}] (v2) to [bend left=15] node[left] {} ($(v4.east)+(0,0.2)$);
					\draw[-{Triangle[scale=1.5,length=2, width=3]}] (v2) to [bend left=15] node[above] {$\DNAA|3$} (v1);
					\draw[-{Triangle[scale=1.5,length=2, width=3]}] ($(v2.east)+(0,0.2)$) to [out=30,in = -30, looseness=5] node[right] {$\DNAC|4$} ($(v2.east)+(0,-0.1)$);
					
					\draw[-{Triangle[scale=1.5,length=2, width=3]}] (v3) to [bend left=15] node[below] {$\DNAT|1$} (v4);
					\draw[-{Triangle[scale=1.5,length=2, width=3]}] ($(v3.west)+(0,0.2)$) to [bend left=15] node[above] {} (v1);
					\draw[-{Triangle[scale=1.5,length=2, width=3]}] (v3) to [bend left=15] node[left] {$\DNAC|3$} (v2);
					\draw[-{Triangle[scale=1.5,length=2, width=3]}] ($(v3.east)+(0,0.2)$) to [out=30,in = -30, looseness=5] node[right] {$\DNAG|4$} ($(v3.east)+(0,-0.1)$);
					
					\draw[-{Triangle[scale=1.5,length=2, width=3]}] (v4) to [bend left=15] node[left] {$\DNAA|1$} (v1);
					\draw[-{Triangle[scale=1.5,length=2, width=3]}] (v4) to [bend left=15] node[left] {} ($(v2.west)+(0,-0.2)$);
					\draw[-{Triangle[scale=1.5,length=2, width=3]}] (v4) to [bend left=15] node[below] {$\DNAG|3$} (v3);
					\draw[-{Triangle[scale=1.5,length=2, width=3]}] ($(v4.west)+(0,0.2)$) to [out=150,in = 210, looseness=5] node[left] {$\DNAT|4$} ($(v4.west)+(0,-0.1)$);	
				\end{tikzpicture}
			}
		\end{center}
		\caption{Channel graph for DNA synthesis using the alternating sequence $\DNAA \DNAC \DNAG \DNAT \; \DNAA \DNAC \DNAG \DNAT \; \ldots$.}
		\label{fig:synthesis}   
	\end{figure} 

\subsubsection{Cost-constrained capacity}
Shannon introduced the concept of \emph{(combinatorial) capacity} of a discrete noiseless channel as the asymptotic growth rate of the number of sequences (of \textit{variable} length) as a function of the sequence cost. In the case of Fig.~\ref{fig:synthesis}, this represents the maximum rate at which information can be encoded into the synthesized DNA strands per unit of synthesis time. Under the assumption of integer edge costs, Shannon analyzed a system of difference equations and derived the now classical result that the capacity is equal to logarithm of the largest root of a determinantal equation associated with the channel. For a channel represented by a graph $G$, we denote this capacity as $C_G$.

Khandekar et al.~\cite{khandekar_discrete_2000} extended Shannon's result to non-integer symbol costs under a mild assumption about the density of sequence costs, and expressed the capacity $C_G$ in terms of the radius of convergence of a series that can be interpreted as a generating function for the  sequence $N(t)$ representing the number of sequences with cost equal to $t$.  Their results were extended to a more general class of channels by B\"{o}cherer et al.~\cite{bocherer_capacity_2010-1}, who expressed  the capacity in terms of a singularity of a complex generating function $F(x)$ for $N(t)$. They interpreted this as a generalization of results from  analytic combinatorics in a single variable~\cite{flajolet_analytic_2009}, a connection that was   established by  B\"{o}cherer~\cite{bocherer_analytic_2007}, who used it to analyze the sub-exponential asymptotics of $N(t)$. Khandekar et al.~\cite{khandekar_discrete_2000} also clarified and extended Shannon's proof of the equivalence of the combinatorial capacity and the  \emph{probabilistic capacity} defined as the maximum entropy rate of a Markov process generating the sequences of the channel. This relationship was further addressed in the setting of more general channels in \cite{bocherer_mximum_entropy_rate_2008,bocherer_capacity_2010}.

In this paper, we consider generalizations of these results to discrete noiseless channels {\bf subject to an  average cost constraint}. In the context of Fig.~\ref{fig:synthesis}, this corresponds to a constraint on the average synthesis time per nucleotide. 

\noindent
\subsection{Contributions}

Our results stem from an integration of contributions within and across three disparate areas, including (a) new results in the spectral theory pertaining to eigenvalues and eigenvectors of graphs and matrices associated with discrete noiseless channels; (b) a geometric and functional analysis of singularities of the complex bivariate generating function that encodes cost and length properties of the channel sequences; and (c) combining these results with a novel application of methods from analytic combinatorics in several variables (ACSV)~\cite{melczer_invitation_2021,PemantleWilsonMelczer2024} to precisely evaluate the asymptotic behavior of the diagonal coefficients of the bivariate generating function. Part of this work was presented at the 2023 IEEE International Symposium on Information Theory~\cite{lenz_isit_2023}.

As a by-product of our main results, we obtain two interesting contributions to information theory. First, we show that the cost-constrained capacity (Definition~\ref{def:capacity:costly}) of the discrete noiseless channel can be expressed explicitly as a  simple function of a specific two-dimensional singularity of the bivariate generating function, thereby generalizing Shannon's classical formula for the channel capacity without cost constraint. This determines the exact asymptotics of the number of fixed-length sequences with limited cost, and fully characterizes the capacity-cost function. Second, the expression for the cost-constrained capacity provides a direct proof of the equivalence of the combinatorial definition of cost-constrained capacity and the probabilistic definition. This equivalence was first established in 2006 via converse inequalities~\cite{soriaga_design_2006}.  

\subsection{Technical Overview}

We next give a brief sketch of the technical underpinnings of our results. Formal definitions of some of the terms used are provided in~Section~\ref{sec:preliminaries}.

\paragraph{Spectral analysis of channel graphs} 
Central to our analysis is a set of new results about a strongly-connected graph $G$. These results  illuminate properties of the spectral radius $\rho_G(x)$ of the cost-enumerator matrix $\ve{P}_G(x)$, which reflects the edge connections and edge costs in $G$. We highlight two graph properties that play a key role in our analysis, \emph{cost diversity} and \emph{cost periodicity}. A cost-diverse graph has at least one pair of equal-length paths with different costs that connect the same pair of vertices. In a cost-periodic graph, for each pair of vertices the costs of all connecting paths of the same length are congruent modulo a fixed integer. We show that cost periodicity is equivalent to a useful formulation of the edge cost function called a \emph{periodic coboundary condition}. Then, we use these definitions to prove structural properties of the eigenvalues of $\ve{P}_G(x)$ on the complex unit circle and log-convexity properties of the spectral radius. We also show that the complement of cost-diverse graphs, namely \emph{cost-uniform} graphs, plays a role in our characterization of costly constrained channels.

\paragraph{Generating functions and singularity analysis for cost-diverse graphs}
We define a generating function $\ve{F}_G(x,y)$ whose coefficients encode information about the number of paths of given length $n$ and cost $t$ emanating from vertices of $G$, denoted $N(t,n)$. For a cost constraint $W$, we let $\alpha=W^{-1}$. To study the asymptotics of $N(t, \lfloor \alpha t \rfloor)$ with the methods of ACSV, we use the previously derived properties of the spectral radius of $\ve{P}_G(x)$ to characterize the  singularities of  $\ve{F}_G(x,y)$. We first determine the set of \emph{minimal} singularities.  For a cost-diverse graph, we further identify those singular points that are \emph{smooth},  satisfy the \emph{critical equation} for $\alpha$, and are \emph{nondegenerate}.

\paragraph{Asymptotic expansions via ACSV}
Our singularity analysis allows the application of a fundamental result in ACSV regarding asymptotic properties of coefficients of multivariate generating functions~\cite[Cor. 5.2]{melczer_invitation_2021}. This leads to our main contribution (Theorem~\ref{thm:fixed-length}), which gives a complete characterization of the asymptotic behavior of $N_G(t,\lfloor \alpha t \rfloor)$ for cost-diverse graphs in two regions of $\alpha$. These two regions provide a characterization of the capacity-cost function $C_G(\alpha)$ in a concise form that elegantly generalizes Shannon's formula (Theorem~\ref{thm:fixed-length:capacity}): one corresponds to a linear scaling of $C_G(\alpha)$, and the other to a non-linear concave behavior. We identify the threshold value of $\alpha$ separating the two regions, as well as the value $\alpha^*$ in the concave region corresponding to the classical combinatorial capacity, $C_G(\alpha^*)=C_G$. The expression for $C_G(\alpha)$ maps directly to the known formula for the cost-constrained capacity $C_G(W)$ in its probabilistic interpretation~\cite{mceliece_maximum_1983,justesen_maxentropic_1984,khandekar_discrete_2000,soriaga_design_2006}. From these results, we extract an asymptotic expansion for the number of paths generated by an arbitrary strongly connected graph and an exact representation of the number of the cost-constrained paths (Theorem~\ref{thm:variable-length:exact}).

\paragraph{Applications}
Subsequence enumeration is a problem that arises in many areas (bioinformatics, information theory, and coding theory). However, developing tight and explicit formulas is an open question in general. Current results are either unwieldy~\cite{mercier_number_2008} or only apply to special cases, such as the alternating sequence~\cite{hirschberg_tight_2000}. Our motivation comes from theoretical models of DNA synthesis in DNA-based storage systems, specifically for a parallel array-based synthesis process.\footnote{A description of the biochemical synthesis process can be found in~\cite{makarychev_batch_2021, kosuri_large-scale_2014,yazdi_dna-based_2015}.} In~\cite{lenz_coding_2020}, the authors provide a connection between costly constrained channels, subsequence enumeration, and efficient DNA synthesis. They show that the capacity of a suitably defined channel characterizes the information rate of synthesized sequences.
\begin{remark}
The tools of ACSV have found other applications in coding theory. They have been used to 
 study asymptotic properties of  runlength-limited sequences with bit-shift correcting properties~\cite{kovacevic_runlength-limited_2019} and those with constraints on their weight and/or number of runs~\cite{KovVuk2022}. They have also been used to determine  Gilbert-Varshamov (GV) bounds for the sticky insertion channel and for the DNA synthesis channel in~\cite{Goyal2023}, as well as for  optimal codes in $L_1$ (or Manhattan) metric in~\cite{Goyal2025}.
\end{remark}

%\begin{remark}
%The tools of ACSV have also been applied to study asymptotic properties of  runlength-limited sequences with bit-shift correcting properties~\cite{kovacevic_runlength-limited_2019} and those with constraints on their weight and/or number of runs~\cite{KovVuk2022}. They have also been used to determine the Gilbert-Varshamov (GV) bound for the sticky insertion channel and for the DNA synthesis channel in~\cite{Keshav2023}
%\end{remark}

%% file: preliminaries.tex
\section{Preliminaries} \label{sec:preliminaries}
We start by setting up basic notation on labeled and weighted graphs, followed by an introduction to generating functions of general integer sequences and a presentation of the generating function of the number of paths with limited cost.
\subsection{Labeled and Weighted Graphs}
Consider a labeled directed graph $G = (\V,\E,\sigma,\tau)$ with vertices $\V$ and edges $\E$. Each edge $e \in \E$ has an initial vertex $\init(e) \in \V$ and a terminal vertex $\final(e) \in \V$. Furthermore, the edges are labeled via a symbol mapping $\sigma : \E \mapsto \Sigma$, where $\Sigma$ is a finite symbol alphabet, and have positive integer weights or costs defined by a cost mapping $\tau: \E \mapsto \N$. A path $\ve{p} = (e_{1},\dots,e_{n})$ of length $n$ is a sequence of edges $e_{1},\dots,e_{n} \in \E$ such that, for all $i \in\{1,\dots,n-1\}$, the final vertex $\final(e_i)$ of the $i$-th edge is the same as the initial vertex $\init(e_{i+1})$ of the next edge. The path starts in $\init(e_1)$ and ends in $\final(e_n)$. A path generates a word $\sigma(\ve{p}) = (\sigma(e_1),\dots,\sigma(e_n)) \in \Sigma^n$ and has cost $\tau(\ve{p}) = \tau(e_1)+\dots+\tau(e_n)$.  For convenience, we sometimes refer to $G = (\V,\E,\sigma,\tau)$ simply as a graph, when the context is clear.

\begin{defn} \label{def:strongly:connected}
	A graph $G = (\V,\E,\sigma,\tau)$ is \emph{strongly connected} if for any two vertices $v_i,v_j \in \V$ there exists a directed path that connects $v_i$ with $v_j$.
\end{defn}

A desirable graph property is that all distinct paths emerging from a vertex generate distinct words. This is guaranteed by the following notion of a graph being \emph{deterministic}~\cite{marcus_introduction_2001}, also known as \emph{right-resolving}~\cite{lind_introduction_1995}.
\begin{defn} \label{def:deterministic}
	A graph $G = (\V,\E,\sigma,\tau)$ is \emph{deterministic} if for all vertices $v \in \V$ the symbol labels $\sigma(e)$ of all edges $e \in \E$ with the same initial vertex $\init(e) = v$ are distinct.
\end{defn}
We use the terms \emph{constrained channel with cost} or \emph{costly constrained channel} to refer to a graph $G = (\V,\E,\sigma,\tau)$ that is strongly connected and deterministic. 

Note that confining to deterministic graphs, as we will do in the sequel, does not restrict the underlying system of constrained sequences, as any labeled graph $G_{\rm lab}=(V,E,\sigma)$ can be represented by an equivalent deterministic graph $G_{\rm lab}'=(V',E',\sigma '$) \cite[Prop. 2.2]{marcus_introduction_2001}. However, we also note that there may not be a cost function $\tau'$ for the equivalent graph such that $G=(V,E,\sigma,\tau)$ is equivalent to $G'=(V',E',\sigma',\tau')$ as a costly constrained channel (see Example 5 in~\cite{abusini_shortmers_2023}).

Periodicity properties of graphs are essential for the subsequent analysis. We start with the notion of the period of a graph.

\begin{defn} \label{def:period}
	Let $G = (\V,\E,\sigma,\tau)$ be a strongly connected graph. We say that $G$ has \emph{period} $d$ if $d$ is the largest integer with the property that for each pair of vertices $v_i$ and $v_j$ the lengths of all paths $\ve{p}$ connecting $v_i$ and $v_j$ are congruent modulo $d$.
\end{defn}
Note that our definition differs from that in \cite[Section 3.3.2]{marcus_introduction_2001}, where the period is defined as the greatest common divisor of all cycle lengths. However, as proven in Lemma~\ref{lemma:gcd:cycle:lengths} in Appendix \ref{appendix}, any graph that has period $d$ in the sense of Definition~\ref{def:period} also has period $d$ in the sense of \cite{marcus_introduction_2001}.

We next establish the notions of uniformity and periodicity of the path costs in a strongly connected graph.

\begin{defn} \label{def:cost:diverse}
	A strongly connected graph $G = (\V,\E,\sigma,\tau)$ is \emph{cost-uniform} if for each pair of vertices $v_i$ and $v_j$, and each length $m$, the costs of all length-$m$ paths $\ve{p}$ connecting $v_i$ and $v_j$ are the same. If $G$ is not cost-uniform, then we say that $G$ is \emph{cost-diverse}.
\end{defn}

\begin{defn} \label{def:cost:period}
For a cost-diverse graph $G = (\V,\E,\sigma,\tau)$ we define the \emph{cost period} $c \in \N$ to be the largest integer with the property that for each pair of vertices $v_i$ and $v_j$, and each length $m$, the costs $\tau(\ve{p})$ of all length-$m$ paths $\ve{p}$ connecting $v_i$ and $v_j$ are congruent modulo $c$.  For a cost uniform graph the congruence holds for all positive integers and we refer to it as a graph with cost period 0. 
\end{defn}

Fig.~\ref{fig:examples:cost:diversity} illustrates the properties discussed above.
\input{fig/examples-cost-period}
\begin{example}
	All graphs in Fig.~\ref{fig:examples:cost:diversity} are easily verified to be deterministic. Fig.~\ref{fig:single:vertex:graph} is a cost-diverse graph with a single vertex. Fig.~\ref{fig:example:constant:cost} is a graph with constant edge cost equal to 1, thus all paths of length $m$ have cost exactly $m$ and the graph is cost-uniform. Fig.~\ref{fig:example:cost:diverse}, on the other hand, is cost-diverse: there are two paths of length $2$ from the left vertex to itself having costs $2$ and $4$, respectively. Fig.~\ref{fig:example:complete:non:cost:diverse} shows a cost-uniform graph, since any cycle of length $m$ from the left vertex to itself has cost $2m$, any cycle of length $m$ from the right vertex to itself has cost $2m$, any path of length $m$ from the left to the right vertex has cost $2m-1$, and  any path of length $m$ from the right to the left vertex has cost $2m+1$. The graph in Fig.~\ref{fig:example:cost:period:2} has cost period $2$, because the cost of any cycle in the graph is even, i.e., a multiple of $2$, and the costs of all paths connecting the left and right vertex are odd, i.e.,  congruent to $1$ modulo $2$.  Similarly, all length-$m$ cycles in the graph in Fig.~\ref{fig:example:cost:period:3} have costs congruent to $m$ modulo $3$. The same is true for length-$m$ paths connecting the left vertex to the middle vertex, and vice versa.  Length-$m$ paths from these vertices to the rightmost vertex have costs congruent to $m+1$ modulo $3$ and length-$m$ paths from the rightmost vertex to either the leftmost or middle vertices have costs congruent to $m+2$ modulo $3$. Thus, the graph has cost period $3$. In addition, the graph has period 2.  This is because all cycles have even length, and this includes cycles of length 2;  all paths between adjacent vertices have odd length; and all paths from the leftmost vertex to the rightmost vertex, and vice versa, have even length.
\end{example}

We now proceed with a novel property that significantly facilitates our analysis and which we will prove to be equivalent to  cost periodicity in Lemma~\ref{lemma:equiv}.

\begin{defn} \label{def:coboundary:condition}
	A strongly connected graph $G = (\V,\E,\sigma,\tau)$ satisfies the \emph{$c$-periodic coboundary condition} if $c$ is the largest integer such that there exists a function $B: \V \rightarrow \mathbb{Q}$ and a constant $b \in \mathbb{Q}$ such that if $e \in  \E$ is an edge from vertex $v_i$ to vertex $v_j$ then the edge cost satisfies
	$$
	\tau(e) \equiv b + B(v_j)-B(v_i) \pmod c.
	$$
	We say that a graph satisfies the \emph{coboundary condition} if the congruence above holds without the modulo operation.
\end{defn}

For many of our results, we analyze the spectrum of an adjacency matrix associated with the labeled and weighted graph $G$.
In fact, we consider a family of adjacency matrices $\ve{P}_G(x)$, parameterized by a value $x$.

\begin{defn} \label{def:cost:enumerator:spectral:radius}
	Given a strongly connected graph $G = (\V,\E,\sigma,\tau)$ with vertices $\V = \{v_1,\dots,v_{|\V|}\}$, the \emph{cost enumerator} matrix $\ve{P}_G(x)$ of $G$ is the $|\V| \times |\V|$ matrix with entries
	$$[\ve{P}_G(x)]_{ij} = \sum_{e \in \E:~ \begin{subarray}{l}\init(e)=v_i,\\ \final(e)=v_j\end{subarray}} x^{\tau(e)}.$$
	We also denote the spectral radius of $\ve{P}_G(x)$ by
	$$ \rho_G(x) = \max \{ |\lambda(x)|: \lambda(x) \text{ is an eigenvalue of } \ve{P}_G(x) \}. $$
\end{defn}
We treat the parameter $x$ as either real-valued or complex-valued, depending on the context. Later we will see that $\rho_G(x)$ plays a central role in the asymptotic behavior of the number of limited-weight paths through $G$. An important quantity is the number of distinct words that are contained in the language of a system. 
\begin{defn}\label{defn:numberwords}
	Given a graph $G = (\V,\E,\sigma,\tau)$,
	for an arbitrary vertex $v \in \V$ we define $\lang_{G,v}(t)$ to be the cost-$t$ \emph{follower set} of $v$, i.e., the set of all words that are generated by some path of cost at most $t$ that starts at $v$. The size of the cost-$t$ follower set is denoted by $N_{G,v}(t) \triangleq |\lang_{G,v}(t)|$. Accordingly, we define $\lang_{G,v}(t,n) \triangleq \lang_{G,v}(t) \cap \Sigma^n$ to be the length-$n$ follower set with size $N_{G,v}(t,n) \triangleq |\lang_{G,v}(t,n)|$.
\end{defn}
 The central quantity of interest for a costly constrained channel is the exponential growth rate of the size of the follower set. This term is often referred to as its \emph{capacity}. The capacity of a channel is independent of the starting vertex, and we omit this in the definition. 
\begin{defn} \label{def:capacity:costly}
	The \textbf{variable-length capacity} of a costly constrained channel $G$ is
	$$
	C_{G} = \limsup_{t\rightarrow\infty}\frac{\log(N_{G,v}(t)) }{t},
	$$
	while the \textbf{fixed-length capacity} is
	$$
	C_G(\alpha) = \limsup_{t\rightarrow\infty}\frac{\log(N_{G,v}(t,\lfloor \alpha t \rfloor)) }{t}.
	$$
\end{defn}

\noindent

Here we use the terms \emph{variable-length} and \emph{fixed-length} capacity to stress their defining nature. In the literature the two quantities are often referred to as (combinatorial) capacity and cost-constrained capacity.

Shannon also introduced a natural probabilistic definition of capacity in terms of the entropy of stationary Markov chains  on the channel graph. A Markov chain $\mathcal{P}$  on the graph $G = (\V,\E,\sigma,\tau)$  defines an edge probability distribution 
$\mathcal{P}:\mathcal{E}\mapsto [0,1]$  and a vertex  probability distribution $\pi:\mathcal{V}\mapsto [0,1]$ where $\pi(v) = \sum_{e: \init(e)=v} \mathcal{P}(e)$.   %We assume that $\mathcal{P}(e) >0$, for all $e\in\mathcal{E}$.
The Markov chain is stationary if $\sum_{e: \final(e)=v} \mathcal{P}(e) = \pi(v)$.
\linebreak
The conditional edge probabilities $q:\mathcal{E}\mapsto [0,1]$ are defined by 
$$
q(e) = \left\{\begin{array}{cl}
\mathcal{P}(e)/\pi(\init(e)) & {\rm if} \; \pi(\init(e))>0\\
0& {\rm otherwise}.
\end{array}
\right.
$$
The entropy $\mathsf{H}(\mathcal{P})$ of the Markov chain $\mathsf{H}(\mathcal{P})$ is defined by 
$$
\mathsf{H}(\mathcal{P}) = -\sum_{v\in\mathcal{V}} \pi(v) \sum_{e: \init(e)=v} q(e).
$$
The Markov chain $\mathcal{P}$ also has a associated average  edge cost $\mathsf{T}(\mathcal{P})$ defined by
$$
\mathsf{T}(\mathcal{P}) = \sum_{e\in\mathcal{E}} \mathcal{P}(e)\tau(e).
$$

%Shannon also introduced a natural probabilistic definition of capacity in terms of the entropy $H(\mathcal{P})$ of stationary Markov chains $\mathcal{P}$ on the channel graph. The Markov chain defines a probability distribution on the edges of the graph,  and a corresponding average  edge cost $T(\mathcal{P})$.

\begin{defn}\label{defn:cprob}
The \textbf{probabilistic capacity} of a costly constrained channel $G$   is
\begin{equation*}
    \label{eq:cprob}
    C_{\rm prob} = \sup_{\mathcal{P}} \frac{\mathsf{H}(\mathcal{P})}{\mathsf{T}(\mathcal{P})},
\end{equation*}
where the supremum is taken over all stationary Markov chains on $G$.
\end{defn}

Shannon remarked that $ C_{\rm prob} \leq C_{G}$, and then
explicitly identified a Markov chain that achieves a normalized entropy equal to $C_G$, thus proving the fundamental equivalence
$
    C_{G} = C_{\rm prob}.
$

As with combinatorial capacity, there is a natural generalization of the probabilistic capacity to the costly constrained setting in which the supremum is taken over Markov chains with average cost at most $W$.

\begin{defn}\label{defn:dprobW}
The  \textbf{cost-constrained probabilistic capacity} of a costly constrained channel $G$ with average symbol cost constraint $W$ is 
\begin{equation*}
    \label{eq:cprobW}
    C_{\rm prob}(W) = \sup_{\mathcal{P}: \mathsf{T}(\mathcal{P})\leq W} \frac{\mathsf{H}(\mathcal{P})}{\mathsf{T}(\mathcal{P})},
\end{equation*}
where the supremum is over stationary Markov chains on $G$ with average cost no more than $W$.
\end{defn}

A concise parametric characterization of $C_{\rm prob}(W)$, found by  constrained optimization methods, is stated in~\cite{justesen_maxentropic_1984,khandekar_discrete_2000, mceliece_maximum_1983}.

\subsection{Generating Functions} \label{sec:generating:functions}
The methods of \emph{analytic combinatorics} derive asymptotic properties of a sequence from analytic properties of its generating function~\cite{flajolet_analytic_2009,melczer_invitation_2021}. Throughout this section, the sequences of interest are the bivariate\footnote{The term bivariate refers to the fact that the integer sequences $N_{G,v}(t,n)$ depend on two variables $t$ and $n$.} sequences $N_{G,v}(t,n)$, whose generating functions we denote
$$ F_{G,v}(x,y) = \sum_{n\geq0} \sum_{t\geq 0} N_{G,v}(t,n) x^ty^n $$
for complex variables $x$ and $y$. As the sequence $N_{G,v}(t,n)$ admits a linear recursion in the variables $t$ and $n$, which we will elaborate on in Section~\ref{sec:derivation:gen}, the generating function $F_{G,v}(x,y)$ is a rational function, and we write
$$ F_{G,v}(x,y) = \frac{Q_{G,v}(x,y)}{H_G(x,y)} $$
for some polynomials $Q_{G,v}(x,y)$ and $H_G(x,y)$. Since $N_{G,v}(t) = \sum_{n\geq 0} N_{G,v}(t,n)$, for the variable-length case, we will regularly abbreviate  the generating function of the integer series $N_{G,v}(t)$ as $F_{G,v}(x) \triangleq F_{G,v}(x,1)$, with numerator $Q_{G,v}(x) \triangleq Q_{G,v}(x,1)$ and denominator $H_G(x) \triangleq H_G(x,1)$.

\begin{lemma} \label{lemma:generating:function}
	Let $G$ be a deterministic graph and let $v$ be a vertex of $G$. The generating function $F_{G,v}(x,y)$ of $N_{G,v}(t,n)$ is given by the entry of
	$$ \ve{F}_G(x,y) = \frac{1}{1-x} \cdot (\ve{I}-y\ve{P}_G(x))^{-1} \ve{1}^\mathrm{T}$$
	corresponding to the vertex $v$, where $\ve{1} = (1,\dots,1) \in \mathbb{R}^{|\V|}$ and $\ve{I}$ is the $|\V| \times |\V|$ identity matrix.
\end{lemma}
We will prove Lemma~\ref{lemma:generating:function} in Section \ref{sec:derivation:gen}. Note that $\ve{I}-y\ve{P}_G(x)$ is not always invertible. However, the values of $x$ and $y$ for which $\ve{I}-y\ve{P}_G(x)$ is singular are singularities of the rational function $F$, which are precisely the objects of interest that determine the asymptotic behavior of the integer sequence $N_{G,v}(t,n)$.

We also remark that the deterministic assumption is crucial to establishing the correspondence between channel sequences and paths in the channel graph that underlies the derivation of the generating function in Lemma~\ref{lemma:generating:function}, see proof of Lemma~\ref{lemma:recursion}. The multivariate singularity analysis of this generating function, in turn, leads to our main results on channel capacity and precise asymptotics of the number of channel sequences in Section~\ref{sec:multivariate:main:results}. 
\begin{example}
	Consider the graph in Fig.~\ref{fig:single:vertex:graph}. In this case $\ve{P}_G(x) = x+x^2$, and thus the generating function of the single vertex is given by%
	$$F_G(x,y) = \frac{1}{(1-x)(1-y(x+x^2))}.$$
\end{example}
%

%% file: fig/examples-cost-period.tex
\begin{figure*}
	\subfloat[Cost-diverse graph with a single vertex.\label{fig:single:vertex:graph}]{%
		\begin{tikzpicture}
			\node[block, text width =0.5cm,node distance=2cm] (v1) {};
			\draw[->] ($(v1.west)+(0,0.1)$) to [out=150,in = 210, looseness=20] node[left] {$a|1$} ($(v1.west)+(0,-0.1)$);
			\draw[->] ($(v1.east)+(0,0.1)$) to [out=30,in = 330, looseness=20] node[right] {$b|2$} ($(v1.east)+(0,-0.1)$);
		\end{tikzpicture}
	}\hfill%
	\subfloat[Cost-uniform graph.\label{fig:example:constant:cost}]{%
		\begin{tikzpicture}		
			\node[block, text width =0.5cm,node distance=2cm] (v1) {$\vphantom{a}$};
			\node[block, right of=v1, text width =0.5cm,node distance=2cm] (v2) {$\vphantom{a}$};
			\draw[-{Triangle[scale=1.5,length=2, width=3]}] ($(v1.west)+(0,0.1)$) to [out=150,in = 210, looseness=15] node[left] {$a|1$} ($(v1.west)+(0,-0.1)$);
			\draw[-{Triangle[scale=1.5,length=2, width=3]}] (v1) to [bend left] node[above] {$b|1$} (v2);
			\draw[-{Triangle[scale=1.5,length=2, width=3]}] (v2) to [bend left] node[above] {$a|1$} (v1);
			%\draw[-{Triangle[scale=1.5,length=2, width=3]}] ($(v2.east)+(0,0.1)$) to [out=30,in = 330, looseness=20] node[right] {$2$} ($(v2.east)+(0,-0.1)$);	
		\end{tikzpicture}
	}\hfill%
	\subfloat[Cost-diverse graph.\label{fig:example:cost:diverse}]{%
		\begin{tikzpicture}		
			\node[block, text width =0.5cm,node distance=2cm] (v1) {$\vphantom{a}$};
			\node[block, right of=v1, text width =0.5cm,node distance=2cm] (v2) {$\vphantom{a}$};
			\draw[-{Triangle[scale=1.5,length=2, width=3]}] ($(v1.west)+(0,0.1)$) to [out=150,in = 210, looseness=15] node[left] {$a|2$} ($(v1.west)+(0,-0.1)$);
			\draw[-{Triangle[scale=1.5,length=2, width=3]}] (v1) to [bend left] node[above] {$b|1$} (v2);
			\draw[-{Triangle[scale=1.5,length=2, width=3]}] (v2) to [bend left] node[above] {$a|1$} (v1);
			%\draw[-{Triangle[scale=1.5,length=2, width=3]}] ($(v2.east)+(0,0.1)$) to [out=30,in = 330, looseness=20] node[right] {$2$} ($(v2.east)+(0,-0.1)$);	
		\end{tikzpicture}
	}\hfill%
	\subfloat[Cost-uniform graph.\label{fig:example:complete:non:cost:diverse}]{%
		\begin{tikzpicture}
			
			\node[block, text width =0.5cm,node distance=2cm] (v1) {$\vphantom{a}$};
			\node[block, right of=v1, text width =0.5cm,node distance=2cm] (v2) {$\vphantom{a}$};
			\draw[-{Triangle[scale=1.5,length=2, width=3]}] ($(v1.west)+(0,0.1)$) to [out=150,in = 210, looseness=15] node[left] {$a|2$} ($(v1.west)+(0,-0.1)$);
			\draw[-{Triangle[scale=1.5,length=2, width=3]}] (v1) to [bend left] node[above] {$b|1$} (v2);
			\draw[-{Triangle[scale=1.5,length=2, width=3]}] (v2) to [bend left] node[above] {$a|3$} (v1);
			\draw[-{Triangle[scale=1.5,length=2, width=3]}] ($(v2.east)+(0,0.1)$) to [out=30,in = 330, looseness=15] node[right] {$b|2$} ($(v2.east)+(0,-0.1)$);	
		\end{tikzpicture}
	}\hfill%
	\subfloat[Cost-diverse graph with cost period $2$.\label{fig:example:cost:period:2}]{%
		\begin{tikzpicture}
			\node[block, text width =0.5cm,node distance=2cm] (v1) {$\vphantom{a}$};
			\node[block, right of=v1, text width =0.5cm,node distance=2cm] (v2) {$\vphantom{a}$};
			\draw[-{Triangle[scale=1.5,length=2, width=3]}] ($(v1.west)+(0,0.1)$) to [out=150,in = 210, looseness=15] node[left] {$a|2$} ($(v1.west)+(0,-0.1)$);
			\draw[-{Triangle[scale=1.5,length=2, width=3]}] (v1) to [bend left] node[above] {$b|1$} (v2);
			\draw[-{Triangle[scale=1.5,length=2, width=3]}] (v2) to [bend left] node[above] {$a|1$} (v1);
			\draw[-{Triangle[scale=1.5,length=2, width=3]}] ($(v2.east)+(0,0.1)$) to [out=30,in = 330, looseness=15] node[right] {$b|2$} ($(v2.east)+(0,-0.1)$);	
		\end{tikzpicture}
	}\hfill%
	\subfloat[Graph with period $2$ and cost period $3$.\label{fig:example:cost:period:3}]{%
		\begin{tikzpicture}
			\node[block, text width =0.5cm,node distance=2cm] (v1) {$\vphantom{a}$};
			\node[block, right of=v1, text width =0.5cm,node distance=2cm] (v2) {$\vphantom{a}$};
			\node[block, right of=v2, text width =0.5cm,node distance=2cm] (v3) {$\vphantom{a}$};
			\draw[-{Triangle[scale=1.5,length=2, width=3]}] (v1) to [bend left] node[above] {$a|1$} (v2);
			\draw[-{Triangle[scale=1.5,length=2, width=3]}] (v2) to [bend left] node[above] {$b|1$} (v1);
			\draw[-{Triangle[scale=1.5,length=2, width=3]}] (v2) to [bend left] node[above] {$a|2$} (v3);
			\draw[-{Triangle[scale=1.5,length=2, width=3]}] (v3) to [bend left] node[above] {$b|3$} (v2);
		\end{tikzpicture}
	}\hfill%
	
%	\begin{subfigure}{.5\textwidth}
%		\centering
%		\begin{tikzpicture}
%			
%			\node[block, text width =0.5cm,node distance=2cm] (v1) {$\vphantom{a}$};
%			\node[block, right of=v1, text width =0.5cm,node distance=2cm] (v2) {$\vphantom{a}$};
%			\node[block, right of=v2, text width =0.5cm,node distance=2cm] (v3) {$\vphantom{a}$};
%			\draw[-{Triangle[scale=1.5,length=2, width=3]}] (v1) to [bend left] node[above] {$c|1$} (v2);
%			\draw[-{Triangle[scale=1.5,length=2, width=3]}] (v2) to [bend left] node[above] {$b|1$} (v1);
%			\draw[-{Triangle[scale=1.5,length=2, width=3]}] (v2) to [bend left] node[above] {$b|2$} (v3);
%			\draw[-{Triangle[scale=1.5,length=2, width=3]}] (v3) to [bend left] node[above] {$a|3$} (v2);
%		\end{tikzpicture}
%		\caption{Another graph with period $2$ and cost period $3$}
%		\label{fig:example:cost:period:3}
%	\end{subfigure}
	
	\caption{Examples illustrating graph properties}
	\label{fig:examples:cost:diversity}
	
\end{figure*}

%% file: results.tex
%!TEX root = bare_jrnl.tex
\section{Main Results} \label{sec:multivariate:main:results}
We are now able to state our main results.
We characterize the fixed-length (i.e., cost-constrained) combinatorial capacity $C_G(\alpha)$  of a discrete noiseless channel described by a strongly connected, deterministic, cost-diverse channel graph (Theorem~\ref{thm:fixed-length:capacity}). As an immediate corollary, we recover the equivalence of the combinatorial and probabilistic capacities in the cost-constrained setting. We also recover Shannon's result on the variable-length capacity $C_G$ for a strongly connected, deterministic channel graph (Theorem~\ref{thm:variable-length:capacity}) .These results are a consequence of a precise  characterization  of the  asymptotic behavior of the number of fixed-length followers $N_{G,v}(t,\alpha t)$ (Theorem~\ref{thm:fixed-length})   and an approximation for  the number of variable-length followers $N_{G,v}(t)$ (Theorem~\ref{thm:variable-length:exact})  in the channel graph. We illustrate the capacity results by deriving the fixed-length and variable-length capacities of $q$-ary alternating sequences (Proposition~\ref{prop:alternating:sequence}).
\subsection{Combinatorial and Probabilistic Capacity}
%
%Our core results describe the number of limited-cost followers in a given costly graph, and the corresponding capacity. Theorem~\ref{thm:fixed-length:capacity}, which characterizes the capacity $C_G(\alpha)$, is derived from an explicit expression for the first order approximation of the number of fixed-length followers $N_{G,v}(t,\alpha t)$ in Theorem~\ref{thm:fixed-length}. Similarly, Theorem~\ref{thm:variable-length:capacity} describes the capacity $C_G$ using an approximation for the number of variable-length followers $N_{G,v}(t)$ given in Theorem~\ref{thm:variable-length:exact}. We begin with the characterization of the fixed-length capacity for arbitrary strongly connected and cost-diverse graphs.
%
\begin{thm}\label{thm:fixed-length:capacity} 
Let $G$ be a strongly connected, deterministic, and cost-diverse graph. Let $\alpha_{G}^{\mathsf{lo}} \triangleq {\rho_G(1)}/{\rho_G'(1)}$ and $\alpha_{G}^{\mathsf{up}} \triangleq \lim\limits_{x\to 0^+} {\rho_G(x)}/{(x\rho_G'(x))}$. For all $\alpha$ with $0 \leq \alpha \leq \alpha_{G}^{\mathsf{lo}}$, 
	$$ C_G(\alpha) = \alpha \log \rho_G(1). $$
	For all $\alpha$ with $\alpha_{G}^{\mathsf{lo}}<\alpha<\alpha_{G}^{\mathsf{up}}$, 
	$$ C_G(\alpha) = -\log x_0 + \alpha\log \rho_G(x_0), $$
	where $x_0$ is the unique real solution to $\alpha x \rho_G'(x) = \rho_G(x)$ in the interval $0<x<1$. For all $\alpha>\alpha_{G}^{\mathsf{up}}$, $C_G(\alpha)=0$.
\end{thm}
\noindent

\begin{remark} \label{rem:alpha:comb:interp}
When $G$ is primitive, meaning strongly connected with period $1$, the delimiting values $\alpha_{G}^{\mathsf{lo}}$ and $\alpha_{G}^{\mathsf{up}}$ have natural combinatorial interpretations. The inverse of $\alpha_{G}^{\mathsf{lo}}$ is the average cost per edge, asymptotically in $n$, over all paths of length $n$ in $G$.  The inverse of $\alpha_{G}^{\mathsf{up}}$ is the minimum average cost per edge among the cycles in $G$. For details, see Proposition~\ref{prop:alpha:comb:interp}.
\linebreak
The evaluation of the fixed-length capacity at $\alpha_G^{\mathsf{up}}$ is  complicated, and this is left for  future work. Thus far, we have not found knowledge of the exact value to be relevant in practice. 
\end{remark}

Theorem~\ref{thm:fixed-length:capacity} improves over previous work \cite{mceliece_maximum_1983, justesen_maxentropic_1984,khandekar_discrete_2000,soriaga_design_2006} in several ways. First, the results of \cite{mceliece_maximum_1983, justesen_maxentropic_1984,khandekar_discrete_2000} only apply  to the cost-constrained probabilistic capacity. Next, none of them explicitly recognizes the role of cost diversity. Moreover, they do not address the full domain of the cost-constrained capacity. In contrast, our results explicitly determine the fixed-length capacity, they can be readily evaluated for cost-diverse graphs, and we consider the entire domain of the capacity function. Specifically, we identify a region for small $\alpha$ in which the capacity exhibits a linear scaling; we determine the exact slope in that region; and we explicitly find the threshold between the linear and non-linear regions. For examples illustrating Theorem~\ref{thm:fixed-length:capacity}, we refer to Proposition~\ref{prop:alternating:sequence} in Section~\ref{sec:alternating}.

\begin{remark}
Our results extend to the case of counting the number of followers of cost \emph{exactly} $t$ instead of \emph{at most} $t$. In that case, the factor $(1-x)$ in the numerator of the generating function $\ve{F}_G(x,y)$ is not present anymore, which has several effects on the results. First, the lower threshold $\alpha_{G}^{\mathsf{lo}}$ decreases to $\alpha_{G}^{\mathsf{lo}} = \lim_{x\to \infty} {\rho_G(x)}/{(x\rho_G'(x))}$. Next, for all $\alpha$ outside the two thresholds, $C_G(\alpha)=0$ and thus the linear region in $\alpha$ disappears.
\end{remark}

Interestingly, our formula for the fixed-length capacity is identical to the formula for the cost-constrained probabilistic capacity in \cite{mceliece_maximum_1983,justesen_maxentropic_1984} (up to  differences in notation and a simple argument to address the linear scaling region). Thus, an immediate corollary of Theorem~\ref{thm:fixed-length:capacity} is the equivalence between  fixed-length capacity and cost-constrained probabilistic capacity. This fundamental result was first proved in \cite{soriaga_design_2006}  using~\cite{justesen_maxentropic_1984} by establishing converse inequalities based on typical sequence arguments, optimization techniques, and a variant of Lemma~\ref{lemma:diff:perron}.
\begin{corollary}
\label{cor:comb_prob_equivalence}
For any strongly connected and cost-diverse graph $G$, the cost-constrained probabilistic capacity $ C_\mathsf{prob}(\alpha^{-1})$ is equal to the fixed-length capacity $ C_G(\alpha)$. 
\end{corollary}

\begin{remark}
In the context of Corollary~\ref{cor:comb_prob_equivalence}, we have a probabilistic counterpart to Remark~\ref{rem:alpha:comb:interp}.  If $G$ is primitive,  $(\alpha_G^{\mathsf{lo}})^{-1}$ can be viewed as the minimum average cost for which maximum entropy can be attained by a Markov chain on $G$. This results in the same linear regime for $\alpha < \alpha_G^{\mathsf{lo}}$. Furthermore, $(\alpha_G^{\mathsf{up}})^{-1}$ can be viewed as the minimum attainable average cost of any Markov chain on $G$.
\end{remark}

We now present the results for the case of variable-length sequences. The following theorem is part of Shannon's famous results on discrete noiseless channels~\cite{shannon_mathematical_1948}.
\begin{thm}\label{thm:variable-length:capacity}
Let $G$ be a strongly connected and deterministic graph and denote by $x_0$ the unique positive solution to $\rho_G(x)=1$. Then the combinatorial capacity of $G$ satisfies
$$ C_G = - \log x_0. $$
\end{thm}
%
%Here it is not required that the graph be cost-diverse, as we are counting limited-cost paths of arbitrary lengths. 
\noindent 
Note that, in contrast to Theorem~\ref{thm:fixed-length:capacity} on fixed-length capacity,  Theorem~\ref{thm:variable-length:capacity} on variable-length capacity 
does not impose the condition that the graph be cost-diverse because here we are counting limited-cost paths of arbitrary lengths.
\begin{remark}\label{rem:maximality:concavity}
Under the conditions of Theorem~\ref{thm:fixed-length:capacity} the function $C_G(\alpha)$ is concave as a function of $\alpha$, and its maximum is equal to  $C_G(\alpha^*) = C_G$ where $\alpha^*=2^{C_G}/\rho_G'(2^{-C_G})$. For details see Proposition~\ref{prop:maximality:concavity}.
\end{remark}

\subsection{Precise Asymptotics}
\label{sec:precise_asymptotics}
Theorem~\ref{thm:fixed-length:capacity} is a direct consequence of the following stronger result, which gives the precise asymptotic behavior of $N_{G,v}(t,\alpha t)$. 

\begin{thm}\label{thm:fixed-length} 
Let $G$ be a strongly connected, deterministic, and cost-diverse graph with period $d$ and  cost period $c$. Denote by $b$ and $B(v_j)$ the quantities from the $c$-periodic coboundary condition in Definition~\ref{def:coboundary:condition}. For all $\alpha$ with $0<\alpha<\alpha_{G}^{\mathsf{lo}}$ and for any $v \in \V$, there is an asymptotic expansion
$$ N_{G,v}(t,\alpha t) = \sum_{j=0}^{d-1}( \lambda_j(1))^{\alpha t} {[\ve{u}_j^\mathrm{T}(1) \ve{v}_j(1) \mathbf{1}^\mathrm{T}]_v} + O\left(\delta^t\right), 
$$
where $0<\delta<(\rho_G(1))^\alpha$ and $\ve{u}_j(x)$ and $\ve{v}_j(x)$ are the right and left eigenvectors of $\ve{P}_G(x)$, with $\ve{v}_j(x)\ve{u}_j^\mathrm{T}(x)=1$, corresponding to the eigenvalues $\lambda_j(x)=\rho_G(x)\mathrm{e}^{2\pi\imu j/d}$. For all $\alpha$ satisfying $\alpha_{G}^{\mathsf{lo}}<\alpha<\alpha_{G}^{\mathsf{up}}$ and $t$ with $\alpha t \in \mathbb{N}$,
\begin{align*}
	N_{G,v}(t,\alpha t) =&\! \sum_{k=0}^{c-1}\sum_{j=0}^{d-1} \left( \frac{(\mathrm{e}^{2\pi\imu bk/c}\lambda_j(x_0))^\alpha}{x_0\mathrm{e}^{2\pi\imu k/c}}\right)^t  \frac{t^{-1/2}}{\sqrt{2\pi \alpha\mathcal{H}(x_0)}}\\
	&\cdot
	\left( \frac{[\ve{D}_k^{-1}\ve{u}_j^\mathrm{T}(x_0) \ve{v}_j(x_0) \ve{D}_k \mathbf{1}^\mathrm{T}]_v}{(1-x_0 \mathrm{e}^{2\pi\imu k/c})} + O\left(\frac{1}{t}\right)\!\right),
\end{align*}
where $ \mathcal{H}(\mathrm{e}^s)  =  \frac{\partial^2}{\partial s^2}\ln \rho_G(\mathrm{e^s})$, $x_0$ is the unique positive solution to $ \alpha x \rho_G'(x) = \rho_G(x),$ and the $\ve{D}_k$ are diagonal matrices with $[\ve{D}_k]_{jj} = \mathrm{e}^{2\pi\imu kB(v_j)/c }$. For all $\alpha>\alpha_{G}^{\mathsf{up}}$, $N_{G,v}(t,\alpha t)$ is eventually $0$.
\end{thm}
In Theorem~\ref{thm:fixed-length}, for the case where $\alpha = \alpha_{G}^{\mathsf{up}}$, comments similar to those in the last part of   Remark~\ref{rem:alpha:comb:interp} apply.

To the best of our knowledge, such a first-order approximation of the number of limited-cost and fixed-length paths through arbitrary strongly connected graphs is new.  Notably, disregarding the $O(1/t)$ term, the term following $t^{-1/2}$ is independent of $t$. 

%We now present the results for the case of variable-length sequences. The following theorem is part of Shannon's famous results on discrete noiseless channels~\cite{shannon_mathematical_1948}.
%%
%\begin{thm}\label{thm:variable-length:capacity}
%Let $G$ be a strongly connected and deterministic graph and denote by $x_0$ the unique positive solution to $\rho_G(x)=1$. Then the combinatorial capacity of $G$ satisfies
%$$ C_G = - \log x_0. $$
%\end{thm}
%%
%Here it is not required that the graph be cost-diverse, as we are counting limited-cost paths of arbitrary lengths. 

We also obtain an exact expression for $N_{G,v}(t)$, the size of the cost-$t$ follower set. This uses a univariate singularity analysis of the generating function $F_{G,v}(x)$.
\begin{thm}\label{thm:variable-length:exact}
Let $G$ be a strongly connected and deterministic graph and denote by $x_1,\dots,x_m$ the solutions to $(1-x)\det(\ve{I}-\ve{P}_G(x)) = 0$. Then, for any vertex $v \in \V$, there exist polynomials $\Pi_{G,v,i}(t)$ calculable from the generating function $F_{G,v}(x)$ such that
$$ N_{G,v}(t) = \sum_{i=1}^m \Pi_{G,v,i}(t) x_i^{-t}. $$
The degree of the polynomial $\Pi_{G,v,i}(t)$ is equal to the multiplicity of the root $x_i$ minus one, in its defining equation.
\end{thm}
The polynomials $\Pi_{G,v,i}(t)$ can easily be computed with a partial fraction decomposition of the generating function \cite{flajolet_analytic_2009}.

%\begin{remark}\label{rem:maximality:concavity}
%Under the conditions of Theorem~\ref{thm:fixed-length:capacity} the function $C_G(\alpha)$ is concave as a function of $\alpha$, and its maximum is equal to  $C_G(\alpha^*) = C_G$ where $\alpha^*=2^{C_G}/\rho_G'(2^{-C_G})$. For details see Proposition~\ref{prop:maximality:concavity}.
%\end{remark}

\subsection{Alternating Synthesis Sequences}
\label{sec:alternating}

To illustrate these results, we consider the discrete noiseless channel that describes the synthesis of sequences using the $q$-ary alternating sequence, i.e., the periodic sequence obtained by repeating the length-$q$ sequence $(0,1,\dots ,q-1)$, for $q\geq 2$.  This is a generalization of the case $q=4$ with period $(\DNAA, \DNAC,\DNAG,\DNAT)$ represented by the graph in Fig.~\ref{fig:synthesis}. For the $q$-ary case, denote the corresponding channel graph  by $G_q$. It is known that, over a $q$-ary alphabet,  the alternating sequences  maximize the number of distinct subsequences \cite{hirschberg_tight_2000}. This implies that  the variable-length capacity $C_{G_q}$ is maximum among channels  corresponding to periodic synthesis sequences over the same alphabet~\cite{lenz_coding_2020}. The fixed-length capacity $C_{G_q}(\alpha)$ describes the exponential growth rate of the number of length-$\alpha t$ subsequences, and we now derive the variable-length and fixed-length capacities of the channel graphs for $q$-ary alternating synthesis sequences.

\begin{proposition} \label{prop:alternating:sequence}
	Consider the $q$-ary alternating sequence. The variable-length capacity  associated with this synthesis sequence is given by
	$$C_{G_q} = - \log_2 x_q,$$
	where $x_q$ is the unique positive solution to $\sum_{i=1}^{q} x^i = 1$. The fixed-length capacity is
	$$C_{G_q}(\alpha) = \begin{cases}
		\alpha\log_2 q & \alpha<\frac{2}{q+1}  \\
		\alpha \log_2 \left( \alpha\sum_{i=1}^q  i x_q(\alpha)^{i-1/\alpha} \right)  &
		\frac{2}{q+1}<\alpha<1 \end{cases}
	$$
	where $x_q(\alpha)$ is the unique solution to $\sum_{i=1}^{q} (1-\alpha i) x^i = 0$, on the interval $0<x<1$.
\end{proposition}

%	For $q=2$ and $q=3$, we obtain $C_{G_2} \approx 0.694$, $C_{G_3} \approx0.879$ and
%	%
%	\begin{align*}
%		 C_{G_2}(\alpha) &= \alpha h\left(\frac{1-\alpha}{\alpha}\right), \\
%		 C_{G_3}(\alpha) &= \alpha h\left(\frac{\gamma}{\alpha}\right) + \gamma h\left(\frac{1-\alpha-\gamma}{\gamma}\right),
%	\end{align*}
%	%
%	where $\gamma = -\frac23\alpha + \frac16 \sqrt{-8\alpha^2+12\alpha-3}+ \frac{1}{2}$ \textcolor{red}{and $h(p$ is the binary entropy function, defined as $h(p)=-p \log p - (1 - p) \log(1 - p)$, $0<p <1$, and $h(0)=h(1)=0$}.
%	%

%
The proof of Proposition~\ref{prop:alternating:sequence} makes use of the following lemma, which simplifies analytical derivations through an explicit expression of the generating function without matrix inversion. Denote by $S(v) = \{ \tau(e): e \in E, \init(e) = v \}$ the multiset of costs of all outgoing edges from $v \in V$.
\begin{lemma} \label{lemma:uniform:cost:spectrum}
	Let $G$ be a strongly connected graph with $S(v_1) = S(v_2)$ for all $v_1,v_2 \in V$. Then,
	$$F_{G,v}(x,y) = \frac{1}{(1-x)\left(1-y\lambda_G(x)\right)}$$
	for all $v \in V$, where $\lambda_G(x) = \sum_{\tau \in S}x^\tau$.
\end{lemma}

\begin{proof}
	Let $\mathbf{1} = (1,\dots,1)$ denote the all-ones vector of length $|V|$. If $S(v_1) = S(v_2)$ for all $v_1,v_2 \in V$, it follows that 
	\begin{align*} P_G(x) \mathbf{1}^\mathrm{T} &= \left(\sum_{\tau \in S(v_1)}x^\tau,\dots,\sum_{\tau \in S(v_{|V|})}x^\tau\right)^\mathrm{T}\\
	&=  \left(\sum_{\tau \in S}x^\tau\right)  \mathbf{1}^\mathrm{T}, 
	\end{align*}
	and thus $\mathbf{1}^\mathrm{T}$ is a right eigenvector of $P_G(x)$ with eigenvalue $\lambda_G(x) = \sum_{\tau \in S}x^\tau$. It follows that
	$$ (I-yP_G(x)) \mathbf{1}^\mathrm{T}  = (1-y \lambda_G(x)) \mathbf{1}^\mathrm{T},$$ 
	and therefore
	\begin{align*}F_{G}(x,y) &= \frac{1}{1-x} \cdot (I-yP_G(x))^{-1} \mathbf{1}^\mathrm{T}\\
	 &= \frac{1}{(1-x)(1-y\lambda_G(x))} \mathbf{1}^\mathrm{T}.
	 \end{align*}
\end{proof}

\begin{proof}[Proof of Proposition~\ref{prop:alternating:sequence}]
	For the special case of the $q$-ary alternating sequence, the synthesis graph $G_q$ is a complete graph, where each vertex has $q$ outgoing edges. The cost multiset of the outgoing edges is $S(v) = \{ 1,2,\dots,q \}$ for all vertices $v \in V_{G_{q}}$. Thus, we can apply Lemma~\ref{lemma:uniform:cost:spectrum} and obtain
	$$\lambda_{G_{q}}(x) = \sum_{i=1}^{q} x^i. $$
	The results on the fixed-length capacity then directly follow from applying Theorem~\ref{thm:fixed-length:capacity}. Similarly, the variable-length capacity follows from Theorem~\ref{thm:variable-length:capacity}. 
	\end{proof}
For $q=2$ and $q=3$,  an explicit computation of the determining equations followed by some algebraic reformulations yields
 $C_{G_2} \approx 0.694$, $C_{G_3} \approx0.879$, and
	\begin{align*}
		 C_{G_2}(\alpha) &= \alpha h\left(\frac{1-\alpha}{\alpha}\right), \\
		 C_{G_3}(\alpha) &= \alpha h\left(\frac{\gamma}{\alpha}\right) + \gamma h\left(\frac{1-\alpha-\gamma}{\gamma}\right),
	\end{align*}
	where $\gamma = -\frac23\alpha + \frac16 \sqrt{-8\alpha^2+12\alpha-3}+ \frac{1}{2}$ and $h(p)$ is the binary entropy function, defined as 
 $h(p)= -p \log p - (1 - p) \log(1 - p) $ for $0\leq p \leq 1.$

Fig.~\ref{fig:alternating:sequence} provides a visualization of the capacity-cost curves for $q$-ary alternating sequences, for a selection of alphabet sizes $q$. The $q=4$ case relevant to DNA synthesis is highlighted with a solid line. 

\input{alternating-sequence}

%% file: alternating-sequence.tex
% This file was created by matlab2tikz.
%
%The latest updates can be retrieved from
%  http://www.mathworks.com/matlabcentral/fileexchange/22022-matlab2tikz-matlab2tikz
%where you can also make suggestions and rate matlab2tikz.
%
\begin{figure*} 
	\centering
	\begin{tikzpicture}
		\begin{axis}[%
			%width=0.85\linewidth,
			width=\linewidth,
			height=8cm,
			xticklabel style={
				/pgf/number format/fixed,
				/pgf/number format/precision=5
			},
			scaled x ticks=false,
			grid=both,
			grid style={line width=.1pt, draw=gray!10},
			major grid style={line width=.2pt,draw=gray!50},
			minor tick num=5,
			xmin=0,
			xmax=1,
			ymin=0,
			ymax=1,
			xlabel = {$\alpha$},
			ylabel = {$C_{\br_q}(\alpha)$},
			legend pos = north west,
			legend style={legend cell align=left,align=left,draw=white!15!black},
			%transpose legend,
			extra x ticks = {0.5663,0.6184,0.723},
			extra x tick labels = {\tiny$0.57$,\tiny$0.62$,\tiny$0.72$},
			extra x tick style = {ticklabel pos=top, grid style={opacity=0.7,densely dashed}},
			extra y ticks = {0.694,0.879,0.947},
			extra y tick labels = {\tiny$C_{\br_2}$,\tiny$C_{\br_3}$,\tiny$C_{\br_4}$},
			extra y tick style = {ticklabel pos=right, grid style={opacity=0.7,densely dashed}},
			]
			\addplot[loosely dashed]
			table[]{%
         		0         0
				0.6600    0.6600
				0.6800    0.6783
				0.7000    0.6897
				0.7200    0.6941
				0.7400    0.6921
				0.7600    0.6838
				0.7800    0.6694
				0.8000    0.6490
				0.8200    0.6226
				0.8400    0.5901
				0.8600    0.5512
				0.8800    0.5057
				0.9000    0.4529
				0.9200    0.3921
				0.9400    0.3219
				0.9600    0.2399
				0.9800    0.1409
				1.0000         0
			};
			\addplot [dashed]
			table[]{%
				         0         0
				0.5000    0.7925
				0.5200    0.8208
				0.5400    0.8430
				0.5600    0.8595
				0.5800    0.8708
				0.6000    0.8773
				0.6200    0.8791
				0.6400    0.8766
				0.6600    0.8699
				0.6800    0.8592
				0.7000    0.8444
				0.7200    0.8257
				0.7400    0.8031
				0.7600    0.7766
				0.7800    0.7460
				0.8000    0.7113
				0.8200    0.6723
				0.8400    0.6289
				0.8600    0.5806
				0.8800    0.5270
				0.9000    0.4676
				0.9200    0.4015
				0.9400    0.3271
				0.9600    0.2422
				0.9800    0.1414
				1.0000         0
			};
			\addplot [solid]
			table[]{%
				        0         0
				0.4000    0.8000
				0.4200    0.8366
				0.4400    0.8668
				0.4600    0.8915
				0.4800    0.9111
				0.5000    0.9261
				0.5200    0.9369
				0.5400    0.9436
				0.5600    0.9466
				0.5800    0.9460
				0.6000    0.9418
				0.6200    0.9344
				0.6400    0.9236
				0.6600    0.9096
				0.6800    0.8924
				0.7000    0.8720
				0.7200    0.8483
				0.7400    0.8214
				0.7600    0.7911
				0.7800    0.7574
				0.8000    0.7200
				0.8200    0.6788
				0.8400    0.6335
				0.8600    0.5838
				0.8800    0.5291
				0.9000    0.4689
				0.9200    0.4021
				0.9400    0.3274
				0.9600    0.2423
				0.9800    0.1414
				1.0000         0
			};
			\addplot [dashdotted]
			table[]{%
				         0         0
				0.0200    0.0600
				0.0400    0.1200
				0.0600    0.1800
				0.0800    0.2400
				0.1000    0.3000
				0.1200    0.3600
				0.1400    0.4200
				0.1600    0.4800
				0.1800    0.5400
				0.2000    0.6000
				0.2200    0.6600
				0.2400    0.7163
				0.2600    0.7646
				0.2800    0.8062
				0.3000    0.8423
				0.3200    0.8735
				0.3400    0.9004
				0.3600    0.9235
				0.3800    0.9430
				0.4000    0.9592
				0.4200    0.9724
				0.4400    0.9826
				0.4600    0.9900
				0.4800    0.9948
				0.5000    0.9970
				0.5200    0.9966
				0.5400    0.9937
				0.5600    0.9884
				0.5800    0.9806
				0.6000    0.9704
				0.6200    0.9576
				0.6400    0.9424
				0.6600    0.9246
				0.6800    0.9043
				0.7000    0.8812
				0.7200    0.8554
				0.7400    0.8267
				0.7600    0.7950
				0.7800    0.7602
				0.8000    0.7219
				0.8200    0.6801
				0.8400    0.6343
				0.8600    0.5842
				0.8800    0.5294
				0.9000    0.4690
				0.9200    0.4022
				0.9400    0.3274
				0.9600    0.2423
				0.9800    0.1414
				1.0000         0
			};
			\addplot [dotted]
			table[]{%
				         0       0
				0.0200    0.1414
				0.0400    0.2423
				0.0600    0.3274
				0.0800    0.4022
				0.1000    0.4690
				0.1200    0.5294
				0.1400    0.5842
				0.1600    0.6343
				0.1800    0.6801
				0.2000    0.7219
				0.2200    0.7602
				0.2400    0.7950
				0.2600    0.8267
				0.2800    0.8555
				0.3000    0.8813
				0.3200    0.9044
				0.3400    0.9248
				0.3600    0.9427
				0.3800    0.9580
				0.4000    0.9710
				0.4200    0.9815
				0.4400    0.9896
				0.4600    0.9954
				0.4800    0.9988
				0.5000    1.0000
				0.5200    0.9988
				0.5400    0.9954
				0.5600    0.9896
				0.5800    0.9815
				0.6000    0.9710
				0.6200    0.9580
				0.6400    0.9427
				0.6600    0.9248
				0.6800    0.9044
				0.7000    0.8813
				0.7200    0.8555
				0.7400    0.8267
				0.7600    0.7950
				0.7800    0.7602
				0.8000    0.7219
				0.8200    0.6801
				0.8400    0.6343
				0.8600    0.5842
				0.8800    0.5294
				0.9000    0.4690
				0.9200    0.4022
				0.9400    0.3274
				0.9600    0.2423
				0.9800    0.1414
				1.0000       0
			};
			\addlegendentry{$q=2$};
			\addlegendentry{$q=3$};
			\addlegendentry{$q=4$};
			\addlegendentry{$q=8$};
			\addlegendentry{$q\rightarrow\infty$};
			
%			\addplot[color=black!20!]
%			table[row sep=crcr] {
%				0 0.694\\
%				0.723 0.694\\
%			};
%			\addplot[color=black!20!]
%			table[row sep=crcr] {
%				0.723 0\\
%				0.723 0.694\\
%			};
		\end{axis}
	\end{tikzpicture}%
	\caption{Synthesis capacity of the alternating sequences over different alphabet sizes. The maxima are highlighted for $q\in\{2,3,4\}$ together with their maximizing $\alpha$. Notice that these plots confirm the concavity of the fixed-length capacity in $\alpha$ and its maximum at the variable-length capacity, which is derived in Proposition~\ref{prop:maximality:concavity}.}%
	\label{fig:alternating:sequence}%
\end{figure*}

%% file: proof-outline.tex
%!TEX root = bare_jrnl.tex
\section{Proof Outline} \label{sec:mul:technical:overview}
Before going further into details we provide an overview of the ingredients required to prove Theorems~\ref{thm:fixed-length:capacity}, \ref{thm:fixed-length}, \ref{thm:variable-length:capacity}, and \ref{thm:variable-length:exact}. To begin with, we concisely highlight the main steps of a multivariate singularity analysis that connects properties of specific singularities of $F_G(x,y)$ to the asymptotic expansion of the \emph{diagonal} coefficients $N_G(t,\alpha t)$. Afterwards, we discuss how we use the theory on irreducible matrices \cite{horn_matrix_2012} to show the implications of strong connectivity and cost-diversity on the spectral properties of cost-enumerator matrices and thus on the singularities of the generating functions.

\subsection{Analytic Combinatorics in Several Variables}
\label{sec:overview-acsv}
Analytic combinatorics~\cite{flajolet_analytic_2009} is a branch of mathematics that uses complex analysis to deduce the asymptotics of an integer sequence $N(t)$ from its generating function $F(x)$. Similarly, \emph{analytic combinatorics in several variables} (ACSV)~\cite{melczer_invitation_2021,PemantleWilsonMelczer2024} treats multivariate integer sequences $N(t_1,t_2)$ (this discussion is specialized to the bivariate rational case we consider) and their generating functions $F(x,y)$. The multivariate analysis resembles the univariate case, translating properties of the generating function near singularities to an asymptotic expansion of the integer series. 

Due to the multivariate nature of the series, there are several ways in which the coefficients $(t_1,t_2)$ can grow to infinity. Typically, one sets $(t_1,t_2) = (t\alpha_1,t\alpha_2)$, for a fixed \emph{diagonal direction} $(\alpha_1,\alpha_2)$, and lets $t\to\infty$. Similar to the univariate case, the singularities closest to the origin determine the asymptotic behavior on the diagonal. In the multivariate case, however, not all of those singularities are relevant for asymptotics. Two properties of singularities -- minimality and criticality -- thus come into play.

\emph{Minimal points} are those singularities for which $H(x,y)$ has no other root with strictly smaller coordinate-wise modulus\footnote{We use the terms \emph{modulus}, \emph{absolute value}, and \emph{magnitude} of a complex variable interchangeably.}. A minimal point is \emph{strictly minimal} if no other singularity has the same coordinate-wise modulus, and \emph{finitely minimal} if only a finite number of other singularities have the same coordinate-wise modulus. For the rational generating functions $F(x,y) = Q(x,y)/H(x,y)$ that we treat in our analysis, there are two types of critical points. First, the \emph{smooth critical points} are the solutions to the polynomial equations $H(x,y) = \alpha_2 x H_{x}(x,y) - \alpha_1 y H_{y}(x,y) = 0$ where at least one of the partial derivatives of $H$ does not vanish. 
Among the smooth critical points, we will be further interested in the \emph{nondegenerate} critical points for which a certain function characterizing the local singularity structure has a nonzero second derivative.
Second, the \emph{non-smooth critical points}, which for us lie in the family of \emph{multiple points}, are any points where the singularity set is the union of two smooth surfaces that intersect, meaning that both partial derivatives of $H$ vanish (in general the characterization of critical points is more complicated). A more detailed  discussion of critical points can be found in Chapters 5 and 9 of \cite{melczer_invitation_2021}.

Under a few additional conditions, the existence of minimal critical points means asymptotics of $N(t\alpha_1,t\alpha_2)$ can be determined from local properties of the generating function $F(x,y)$ near these points.  For more details, see Section~\ref{sec:acsv}.

\subsection{From Cost-Diverse Graphs to Multivariate Analytical Combinatorics via Spectral Analysis}

The starting point of our ACSV analysis is the generating function derived in Lemma \ref{lemma:generating:function}. Before we can invoke the general results of ACSV, however, we need to establish a comprehensive theory about costly constrained channels and their associated cost-enumerator matrices to gather the necessary understanding of the associated singularities. To start with, through the restriction to cost-diverse graphs (Definition \ref{def:cost:diverse}), we avoid certain degenerate cases. Previous work \cite{khandekar_discrete_2000} observed that graphs with constant edge cost have the  property that the cost of any path is a linear function of its length, meaning that the capacity is simply determined by the number of paths through the graph of a given length. Generalizing this observation, we introduced the notion of cost-uniform and cost-diverse graphs in Definition~\ref{def:cost:diverse}. We show that if a graph is not cost-diverse, i.e., it is cost-uniform, then the cost of any path is an affine linear function of the path length and thus the average cost of any path approaches a constant. Specifically, all paths of length $n$ have a cost of $\gamma n + \beta$, for some $\gamma$ and $\beta$. This means that $N(t,\alpha t)$ counts all paths if $\alpha<\frac{1}{\gamma}$ and no paths if $\alpha>\frac{1}{\gamma}$ . Therefore,  $C_G(\alpha)$ wil be a linear ramp from zero to the combinatorial capacity on the interval $0 \leq \alpha < \frac{1}{\gamma}$  beyond which it drops to zero.

Focusing on cost-diverse and strongly connected graphs, we derive a variety of interesting properties. Referring to the fact that the cost-enumerator matrix $\ve{P}_G(x)$ of a strongly-connected graph is irreducible (see Definition \ref{def:irreducible}) for positive $x \in \mathbb{R}^+$, we start in Section \ref{sec:perron:frobenius} by deriving general properties of irreducible matrices. To this end, we use the famous Perron-Frobenius Theorem (Theorem \ref{thm:perron:frobenius}) and a refinement \cite[Thm. 3.18]{marcus_introduction_2001} (see Theorem \ref{thm:irreducible:spectral:circle}) to deduce properties of the parametrized cost-enumerator matrix. These results will serve us in Section \ref{sec:spectral:radius:diverse} where we derive spectral properties of the cost-enumerator matrix of a cost-diverse graph. A key milestone for our results is Lemma \ref{lemma:equiv}, which provides an equivalence between cost diversity and the coboundary condition (Definition~\ref{def:coboundary:condition}) and establishes the  implication of these propeties for the cost-enumerator matrix, namely a nice behavior of the spectral radius under rotations, and the log-log-linearity or log-log-convexity of the spectral radius along the real axis.

Our equivalence result in Lemma~\ref{lemma:equiv} establishes key properties of the cost-enumerator matrix $\ve{P}_G(x)$ and is the basis for a derivation of the attributes of the generating function. This appears in Section~\ref{sec:multivariate:proof}. At a high level, we need to find the minimal singularities of our generating functions $\ve{F}_G(x,y)$ and characterize their critical points in order to apply the ACSV theorems in Section~\ref{sec:acsv}. More concretely, in Lemma~\ref{lemma:minimal:singularities} we identify the minimal singularities of $\ve{F}_G(x,y)$ and express them as a function of the graph period $d$, the cost period $c$, and the spectral radius $\rho_G(x)$. Due to the Perron-Frobenius Theorem, $\rho_G(x)$ is the single real eigenvalue of maximum modulus of $\ve{P}_G(x)$, which we use to show that the points $(x,1/\rho_G(x))$ are minimal singularities for $0<x<1$.
Next, we prove in Lemma~\ref{lemma:singularities:critical:smooth} that the minimal points that we have found in Lemma~\ref{lemma:minimal:singularities} are smooth points. We further derive a condition based on $\alpha$, the spectral radius $\rho_G(x)$, and its derivative $\rho_G'(x)$ that determines criticality of the minimal singularities. A key component of the proof is Lemma~\ref{lemma:cost-enumerator:complex:circle}, which shows  that the rotation of $x$ by multiples of $2\pi/c$ along the complex circle results in similar cost-enumerator matrices. Diving deeper into the critical point condition, Lemma~\ref{lemma:alpha:boundary:unique:solution} guarantees a unique smooth critical point when $\alpha$ is in a certain interval. The proof uses the strict log-log convexity of $\rho_G(x)$ proven in Lemma~\ref{lemma:log:log:convex}. The final component of our multivariate singularity analysis is Lemma~\ref{lemma:singularity:nondegenerate}, which proves that the singular set near the smooth critical points has nondegenerate geometry. For an overview of this roadmap, see Fig.~\ref{fig:relationship:properties}.

To establish Theorem~\ref{thm:fixed-length}, we then apply results from \cite{melczer_invitation_2021} and use the spectral properties of $\ve{P}_G$ that we have derived from the graph properties. When $(x_0, 1/\rho_G(x_0))$ is a smooth point of the singular set of the generating function, the  asymptotic behaviour is determined using Theorem~\ref{prop:ACSVsmooth}, while in the non-smooth case it follows from an application of Theorem~\ref{prop:ACSVnonsmooth}. 

%% file: perron-frobenius.tex
%!TEX root = bare_jrnl.tex
\section{Perron-Frobenius Theory} \label{sec:perron:frobenius}
In this section, we briefly revisit the central statements of the powerful Perron-Frobenius theorem and derive associated results on irreducible matrices parametrized by a variable $x$. These results are key ingredients to prove our main statements.
\subsection{Known Results from Perron-Frobenius Theory}
The Perron-Frobenius Theorem is a well-known result about the spectral properties of irreducible matrices. For the following definition of irreducible matrices, recall the notion of strong connectivity of a graph from Definition~\ref{def:strongly:connected}. 
\begin{defn} \label{def:irreducible}
	Let $\ve{P} \in \R^{M \times M}$ be a square real matrix with nonnegative entries. Associate with $\ve{P}$ the directed graph $G$ with $M$ vertices which is constructed by connecting state $i$ to $j$ if and only if $[\ve{P}]_{ij}>0$. We call $\ve{P}$ \emph{irreducible} if $G$ is strongly connected.
\end{defn}
Perron \cite{perron_zur_1907} and Frobenius \cite{frobenius_uber_1912} revealed important   properties of the eigenvalues of irreducible matrices. Among those spectral properties is the existence of a positive real eigenvalue which is equal to the spectral radius of the matrix, i.e., the largest magnitude of any eigenvalue. In the following statements, which are an excerpt of the original Perron-Frobenius theorem, we collect those properties of irreducible matrices that are most relevant for our purposes.
\begin{thm}[\cite{perron_zur_1907,frobenius_uber_1912}] \label{thm:perron:frobenius}
	Let $\ve{P}$ be an irreducible matrix with spectral radius $\rho$. Then,
	\begin{enumerate}
		\item $\rho$ is an eigenvalue with multiplicity one.
		\item There exist positive right and left eigenvectors $\ve{u}>0$ and $\ve{v}>0$ corresponding to the eigenvalue $\rho$ such that $\ve{P}\ve{u}^\mathrm{T} = \rho \ve{u}^\mathrm{T}$ and $\ve{v}\ve{P} = \rho \ve{v}$.
	\end{enumerate}
\end{thm}
By the Perron-Frobenius theorem, for an irreducible matrix with spectral radius $\rho$ there is a unique eigenvalue $\lambda$ which is equal to the spectral radius. We will refer to this eigenvalue as the \emph{Perron root} in the sequel. In fact, the structure of the eigenvalues on the spectral circle are precisely known for irreducible matrices. To characterize these eigenvalues, recall the definition of periodicity of a graph from Definition~\ref{def:period}; we say a matrix $\ve{P}$ has \emph{period} $d$ if the associated directed graph from Definition~\ref{def:irreducible} has period $d$. If the period of an irreducible matrix $\ve{P}$ is $d$ then $\ve{P}$ has precisely $d$ simple eigenvalues of maximum modulus. More precisely, those eigenvalues precisely divide the complex circle into $d$ equally sized segments. The following theorem summarizes this property.
\begin{thm}[{{\cite[Thm. 3.18]{marcus_introduction_2001}}}] \label{thm:irreducible:spectral:circle}
	Let $\ve{P} $ be an irreducible matrix with period $d$. Then $\ve{P}$ has precisely $d$ simple eigenvalues of maximum modulus. Denoting $\rho$ as the spectral radius of $\ve{P}$, those eigenvalues have the form $\rho\mathrm{e}^{2\pi \imu j/d}$, where $j \in \{0,1,\dots,d-1 \}$.
\end{thm}
This theorem holds since, as proven in Lemma~\ref{lemma:gcd:cycle:lengths} in Appendix \ref{appendix}, any graph that is periodic in the sense of Definition~\ref{def:period} is also periodic as defined in \cite{marcus_introduction_2001}. 
Another very useful result for irreducible matrices is Wielandt's theorem \cite{wielandt_unzerlegbare_1950}, which we now describe.
\begin{thm}[{{\cite{wielandt_unzerlegbare_1950}}}]	\label{thm:wielandt}
	Let $\ve{P} \in \mathbb{R}^{M\times M}$ be an irreducible matrix and $\ve{Q} \in \mathbb{C}^{M\times M}$ be a matrix with $|[\ve{Q}]_{ij}| \leq [\ve{P}]_{ij}$ for all $1 \leq i,j \leq M$. Then $\rho(\ve{Q}) \leq \rho(\ve{P})$. Furthermore, equality holds (i.e., $\rho(\ve{P}) \mathrm{e}^{\imu \phi}$ is an eigenvalue of $\ve{Q}$ for some $\phi$) if and only if there exist $\theta_1,\dots,\theta_M$ such that
	$$ \ve{Q} = \mathrm{e}^{\imu \phi} \ve{D}^{-1} \ve{P} \ve{D}, $$
	where $\ve{D}$ is a diagonal matrix with entries $[\ve{D}]_{jj} = \mathrm{e}^{\imu \theta_j}$.
\end{thm}
The power of this theorem lies in the exact characterization of the conditions under which the spectral radii of two  matrices, where one matrix is component-wise smaller than the other, agree. %We can thus use this theorem to analyze the spectral radius of the cost-enumerator matrix $\ve{P}_G(x)$, when $x$ rotates along the complex circle.
For a detailed proof of this theorem and for more details on irreducible matrices, including a comprehensive section on the Perron-Frobenius Theorem, we refer the reader to the textbooks \cite[Section 8.3]{meyer_matrix_2000} and \cite[Section 8.4]{horn_matrix_2012}.
\subsection{Essentials on Irreducible Matrices} \label{sec:irreducible:matrices}
We proceed with establishing basic results on irreducible matrices, which will be used in the derivation of our main statements. Assume that $\ve{P}$ is an irreducible matrix with period $d$. We start with a simple result on the rank of the adjoint matrix $\rho \mathrm{e}^{2\pi \imu j/d} \ve{I}-\ve{P}$, where $\rho$ is the spectral radius of the irreducible matrix $\ve{P}$.
\begin{lemma} \label{lemma:adj:rank:1}
	Let $\ve{P}$ be an irreducible matrix with period $d$ and spectral radius $\rho$. Then the adjoint matrix $\adj(\rho\mathrm{e}^{2\pi \imu j/d} \ve{I}-\ve{P})$ has rank one for all $j \in \{0,1,\dots,d-1\}$.
\end{lemma}
\begin{proof}
	Note that the result can be deduced from, e.g., \cite[Prob. 6.2.11]{meyer_matrix_2000} but we provide a short proof for the reader's convenience. Denote by $M$ the number of rows (and columns) of $\ve{P}$ and abbreviate $\theta_j \triangleq 2 \pi j/d$. We first show that $\rank (\rho \ve{I}-\ve{P}) = M-1$. The eigenvalues of $\rho\mathrm{e}^{\imu \theta_j} \ve{I} -P$ are given by $(\rho\mathrm{e}^{\imu \theta_j} - \lambda_i)$ for $i\in\{1,\dots,M\}$, where the $\lambda_i$ are the (not necessarily distinct) eigenvalues of $\ve{P}$. Since $\ve{P}$ is irreducible and has period $d$, by the Perron-Frobenius Theorem (Theorem~\ref{thm:perron:frobenius}) and Theorem~\ref{thm:irreducible:spectral:circle}, the $\rho\mathrm{e}^{\imu \theta_j}$ for $j\in \{0,1,\dots,d-1\}$ are eigenvalues of multiplicity one and thus exactly one of the eigenvalues $\rho\mathrm{e}^{\imu \theta_j} - \lambda_i$ will be zero and all other nonzero. Therefore $\rank (\rho\mathrm{e}^{\imu \theta_j} \ve{I}-\ve{P}) = M-1$. Next, we observe that $\adj(\rho\mathrm{e}^{\imu \theta_j} \ve{I}-\ve{P}) (\rho\mathrm{e}^{\imu \theta_j} \ve{I}-\ve{P}) = \det (\rho\mathrm{e}^{\imu \theta_j} \ve{I}-\ve{P})\ve{I} = \ve{0}$ and thus $\adj(\rho\mathrm{e}^{\imu \theta_j} \ve{I}-\ve{P})$ spans a subspace of the left nullspace of $(\rho\mathrm{e}^{\imu \theta_j} \ve{I}-\ve{P})$. Since $\rho\mathrm{e}^{\imu \theta_j} \ve{I}-\ve{P}$ has rank $M-1$, it follows that $\rank(\adj(\rho\mathrm{e}^{\imu \theta_j} \ve{I}-\ve{P})) \leq 1$. On the other hand, $\rho\mathrm{e}^{\imu \theta_j} \ve{I}-\ve{P}$ has rank $M-1$ and thus there exists an $(M-1)\times(M-1)$ submatrix of $\rho\mathrm{e}^{\imu \theta_j} \ve{I}-\ve{P}$ which is non-singular \cite[Ch. 4.5]{meyer_matrix_2000}, and it follows that at least one entry of $\adj(\rho\mathrm{e}^{\imu \theta_j} \ve{I}-\ve{P})$ is nonzero. Therefore $\adj(\rho\mathrm{e}^{\imu \theta_j} \ve{I}-\ve{P})$ cannot have rank zero and thus has rank one.
\end{proof}
Next, we establish a useful characterization of the adjoint matrix $\adj(\rho \ve{I}-\ve{P})$. In particular, we will show that we can represent this adjoint matrix as the outer product of the right and left eigenvectors associated with the Perron root $\rho$.
\begin{lemma} \label{lemma:adj:outer:product} 
	Let $\ve{P}$ be an irreducible matrix with period $d$. Then there are $d$ eigenvalues $\lambda_j=\rho\mathrm{e}^{2\pi \imu j/d}$ for $j \in \{0,1,\dots,d-1\}$ of maximum modulus with corresponding right and left eigenvectors  $\ve{u}_j$ and $\ve{v}_j$, normalized such that $\ve{v}_j(x)\ve{u}_j^\mathrm{T}(x)=1$. The adjoint matrix $\adj(\rho\mathrm{e}^{ 2\pi \imu j/d} \ve{I}-\ve{P})$ is given by
	$$ \adj(\rho\mathrm{e}^{ 2\pi \imu j/d} \ve{I}-\ve{P}) = c_j \cdot \ve{u}_j^\mathrm{T}\ve{v}_j, $$
	where $c_j\neq0$ is a linear scaling factor. Thus, $\adj(\rho \ve{I}-\ve{P})$ is either all-positive or all-negative.
\end{lemma}

\begin{proof}
	Again abbreviate for convenience $\theta_j \triangleq 2 \pi j/d$. By Lemma~\ref{lemma:adj:rank:1}, the adjoint matrix has rank one. It follows that $\adj(\rho\mathrm{e}^{\imu \theta_j} \ve{I}-\ve{P})$ can be written as the product $\ve{u}_j^\mathrm{T}\ve{v}_j$ of two vectors $\ve{u}_j$ and $\ve{v}_j$, i.e., $\adj(\rho\mathrm{e}^{\imu \theta_j} \ve{I}-\ve{P}) = \ve{u}_j^\mathrm{T}\ve{v}_j$. The properties of the adjoint matrix~\cite[p.~20]{horn_matrix_2012} imply that
	\begin{align*}
		\mathrm{adj}(\rho\mathrm{e}^{\imu \theta_j} \ve{I} - \ve{P})(\rho\mathrm{e}^{\imu \theta_j} \ve{I} - \ve{P}) &= (\rho\mathrm{e}^{\imu \theta_j} \ve{I} - \ve{P})\mathrm{adj}(\rho\mathrm{e}^{\imu \theta_j} \ve{I} - \ve{P}) \\
		&= \det(\rho\mathrm{e}^{\imu \theta_j} \ve{I} - \ve{P}) \ve{I}.
	\end{align*}
	By Theorem~\ref{thm:irreducible:spectral:circle}, $\rho\mathrm{e}^{\imu \theta_j}$ is an eigenvalue of $\ve{P}$, which implies that $\rho\mathrm{e}^{\imu \theta_j} \ve{I} - \ve{P}$ is singular, so $\det(\rho\mathrm{e}^{\imu \theta_j} \ve{I} - \ve{P})=0$.
	Hence, 
	$$
	\mathrm{adj}(\rho\mathrm{e}^{\imu \theta_j} \ve{I} - \ve{P})(\rho\mathrm{e}^{\imu \theta_j} \ve{I} - \ve{P}) = (\rho\mathrm{e}^{\imu \theta_j} \ve{I} - \ve{P})\mathrm{adj}(\rho\mathrm{e}^{\imu \theta_j} \ve{I} - \ve{P}) =\ve{0}.
	$$
	Therefore, the columns of $\mathrm{adj}(\rho\mathrm{e}^{\imu \theta_j} \ve{I} - \ve{P})$ are right eigenvectors of $\ve{P}$ associated to 
	$\rho\mathrm{e}^{\imu \theta_j}$. Similarly, the rows of  $\mathrm{adj}(\rho\mathrm{e}^{\imu \theta_j} \ve{I} - \ve{P})$ are left eigenvectors of $\ve{P}$ associated to   
	$\rho\mathrm{e}^{\imu \theta_j}$, and therefore \mbox{$\adj(\rho\mathrm{e}^{\imu \theta_j} \ve{I}-\ve{P}) = \ve{u}_j^\mathrm{T}\ve{v}_j$}, where $\ve{u}_j$ and $\ve{v}_j$ are right and left eigenvectors corresponding to $\rho\mathrm{e}^{\imu \theta_j}$. It is not possible that $c_j = 0$, since $\rank(\adj(\rho\mathrm{e}^{\imu \theta_j} \ve{I} - \ve{P}))=1$ by Lemma~\ref{lemma:adj:rank:1}.
	
	We proceed with proving the second statement. By the Perron-Frobenius Theorem, $\ve{u}_0$ is either all-zero, all-positive, or an all-negative vector, and the same applies to $\ve{v}_0$. If we now assume that $\ve{B} \triangleq  \mathrm{adj}(\rho  \ve{I} - \ve{P})$ satisfies $[\ve{B}]_{11}>0$, the observations above imply that every entry of $\ve{B}$ must be positive. Similarly, if $[\ve{B}]_{11}<0$, we can conclude that every entry of $\ve{B}$ must be negative. 
\end{proof}
The cost-enumerator matrix $\ve{P}_G(x)$ is a matrix that is parametrized by  a complex-valued variable $x \in \mathbb{C}$. In our analysis, due to the strong connectivity of the graph $G$, the matrix $\ve{P}_G(x)$ is irreducible for all positive and real-valued $x\in\mathbb{R}^+$. By Definition~\ref{def:cost:enumerator:spectral:radius}, the entries of $\ve{P}_G(x)$ are polynomials in $x$ and thus analytic\footnote{A function is analytic at a point $x$  if it can locally be represented by a power series. A function is analytic in a domain if and only if it is complex differentiable in the same domain, see, e.g. \cite[Thm. IV.1]{flajolet_analytic_2009}} functions in $x$. This analyticity then implies, by the implicit function theorem for algebraic functions, that the eigenvalue $\lambda(x)$ that is equal to $\rho_G(x)$ on the real axis, is analytic in a neighborhood around the positive real axis.
\begin{lemma} \label{lemma:spectral:radius:analytic}
	Let $\ve{P}(x)$ be a matrix with spectral radius $\rho(x)$, whose entries are analytic functions in $x \in \mathbb{C}$. Also assume that $\ve{P}(x)$ is irreducible with period $d$ for all $x \in \mathbb{R}^+$. Then, for each $j\in \{0,1,\dots,d-1\}$ and all real-valued $x \in \mathbb{R}^+$, there exists a unique eigenvalue $\lambda_j(x)$ of $\ve{P}(x)$ with $\lambda_j(x) = \rho(x)\mathrm{e}^{2\pi \imu j/d}$, which is analytic in a complex neighborhood around the positive real axis. Furthermore, the associated right and left eigenvectors $\ve{u}_j(x)$ and $\ve{v}_j(x)$, normalized to $\ve{v}_j(x)\ve{u}_j^\mathrm{T}(x)=1$, are analytic on the same domain.
\end{lemma}
\begin{proof}
	By the Perron-Frobenius Theorem (Theorem~\ref{thm:perron:frobenius}) and the extension in Theorem~\ref{thm:irreducible:spectral:circle}, for every $j \in \{0,1,\dots,d-1\}$ and $x_0\in \mathbb{R}^+$, the value $\lambda_j(x_0)=\rho(x_0)\mathrm{e}^{\imu 2\pi \imu j/d}$ is a simple root of the characteristic polynomial $\phi(\lambda)= \det(\lambda\ve{I}- \ve{P}(x_0))$. The coefficients of this polynomial $\phi(\lambda)$ are polynomials in analytic functions, as the entries of $\ve{P}(x)$ are analytic by assumption. The implicit function theorem for algebraic functions \cite[pp. 66-67]{wilkinson_algebraic_1988} then implies that for each $x_0>0$ there exists an $\epsilon>0$ such that $\lambda_j(x)$ is an eigenvalue of $\ve{P}(x)$ and $\lambda_j(x)$ is an analytic function for all $x\in\mathbb{C}$ with $|x-x_0|<\epsilon$. As proven in \cite[pp. 66-67]{wilkinson_algebraic_1988}, the associated eigenvectors are also analytic functions in $x$ in a neighborhood around the positive real axis.
\end{proof}
Note that a continuous continuation of the Perron root to the whole complex plane does not in general have to be unique. This is because the Perron-Frobenius theorem only guarantees the uniqueness of the Perron root $\ve{P}(x)$  for positive $x$. For all other $x \in \mathbb{C} \setminus \R^+$ the eigenvalues $\lambda_j(x)$ might intersect, meaning that the implicit function theorem does not hold, and thus a unique analytic extension of the root is not possible anymore. The following example illustrates the generic case, showing that the eigenvalues intersect at the origin $x=0$.
\begin{example}
	Consider the graph $G$ with cost-enumerator matrix
	$$\ve{P}_G(x) = \begin{pmatrix}
		x^2&x\\x&x^2
	\end{pmatrix}.$$
	The two eigenvalues of this matrix are given by $\lambda_1(x) = x+x^2$ and $\lambda_2(x) = -x+x^2$. We directly see that $\lambda_1(0) = \lambda_2(0) = 0$, and thus the two eigenvalues intersect at the origin $x=0$. In fact, all eigenvalues of any cost-enumerator matrix intersect at $x=0$, since $\ve{P}_G(0) = \ve{0}$.
\end{example}
Besides the spectral radius $\rho(x)$ of the cost enumerator matrix, we are also interested in its derivative. This is because the derivative appears as a component of the critical point equations (see, e.g., Theorem~\ref{thm:fixed-length:capacity}) and it can be further used to analyze the convexity of $\rho(x)$. 
\begin{lemma} \label{lemma:diff:perron}
	Let $\ve{P}(x)$ be a matrix with spectral radius $\rho(x)$, whose entries are analytic functions in $x \in \mathbb{C}$. Furthermore, let $\ve{P}(x)$ be irreducible with period $d$ for all $x \in \mathbb{R}^+$. Then the eigenvalues $\lambda_j(x)$ of $\ve{P}(x)$ of maximum modulus, and the associated right and left eigenvectors $\ve{u}_j(x)$ and $\ve{v}_j(x)$, normalized to $\ve{v}_j(x)\ve{u}_j^\mathrm{T}(x)=1$, are analytic in a neighborhood around $\mathbb{R}^+$, and
$$ \ve{v}_j(x)\frac{\partial \ve{P}(x)}{\partial x} \ve{u}_j^\mathrm{T}(x) =  \frac{\partial \lambda_j(x)}{\partial x}. $$
\end{lemma}
\begin{proof}
	To start with, denote by $\lambda_j(x)$ the eigenvalues of maximum modulus whose existence is guaranteed by Lemma~\ref{lemma:spectral:radius:analytic}. The differentiability of $\lambda_j(x),\ve{u}_j(x),$ and $\ve{v}_j(x)$ then follows from the analyticity of $\lambda_j(x)$ proven in Lemma~\ref{lemma:spectral:radius:analytic}. Differentiating $\ve{P}(x)\ve{u}_j^\mathrm{T}(x)=\lambda_j(x)\ve{u}_j^\mathrm{T}(x)$ on both sides with respect to $x$ yields
	$$  \ve{P}(x) \frac{\partial \ve{u}_j^\mathrm{T}(x)}{\partial x} + \frac{\partial \ve{P}(x)}{\partial x} \ve{u}_j^\mathrm{T}(x) = \lambda_j(x) \frac{\partial \ve{u}_j^\mathrm{T}(x)}{\partial x} + \frac{\partial \lambda_j(x)}{\partial x} \ve{u}_j^\mathrm{T}(x). $$
	Multiplying with $\ve{v}_j(x)$ from the left, one obtains
	$$ \ve{v}_j(x) \frac{\partial \ve{P}(x)}{\partial x} \ve{u}_j^\mathrm{T}(x) = \frac{\partial \lambda_j(x)}{\partial x} $$
	as desired. We remark that a special case of this result was proved in~\cite{soriaga_design_2006}.
\end{proof}
Note that, although $\lambda_1(x)=\rho(x)$ for all $x \in \mathbb{R}^+$ (where we denote by $\lambda_1(x)$ the Perron root), the spectral radius $\rho(x)$ is not necessarily differentiable with respect to complex-valued $x$, as $\rho(x)$ is equal to the \emph{magnitude} of the largest eigenvalue. Even though the eigenvalues $\lambda_j(x)$ of maximum modulus are analytic in a neighborhood around the real axis, the magnitude function is not an analytic function on the whole complex plane.

%% file: spectral-cost-diverse.tex
%!TEX root = bare_jrnl.tex
\section{Spectral Properties of Cost-Diverse Graphs}
\label{sec:spectral:radius:diverse}
An important requirement of Theorem~\ref{thm:fixed-length} is that the graph $G$ be cost-diverse. Roughly speaking, by Definition~\ref{def:cost:diverse}, cost diversity means that the  average costs assumed by paths connecting two vertices do not approach a constant for large path lengths. This property is important in the derivation of the asymptotics of the bivariate series $N_{G,v}(t,\alpha t)$ as it entails a smooth behavior of the series in the parameter $\alpha$. Conversely, if $G$ is cost-uniform, there is in fact only a single value for $\alpha$ for which the series $N_{G,v}(t,\alpha t)$ does not vanish eventually. Note that \cite{khandekar_discrete_2000} found that graphs for which all edge costs are the same have this discontinuous behavior, however these are not the only graphs that fall into this category. We generalize this observation and show that cost diversity is the precise graph property that distinguishes between a smooth and discontinuous behavior\footnote{Cost diversity has been shown in \cite{liu_coding_2020} to impose desirable properties on the Perron root.}. We further extend the notion of cost diversity to the property of having cost period $c$ (Definition \ref{def:cost:period}) and show that it relates to a very special structure of the cost-enumerator matrix $\ve{P}_G(x)$, when $x$ is rotated in multiples of $2\pi/c$ along the complex circle. 

Connections between cost diversity and the spectral radius will be integral to Theorem~\ref{thm:fixed-length}. As we will see, the coboundary condition defined in Definition~\ref{def:coboundary:condition} arises in a variety of results related to the Perron-Frobenius Theorem and is very useful for proving several of our results. We further need the notion of log-log-convexity, which is defined as follows.
\begin{defn} \label{def:log:log:convexity}
	Let $I \subseteq \R^+$ be an interval and $f(x) : I\mapsto \R^+$ be a function on that interval. We call $f(x)$ \emph{log-log-convex} if $\ln f(\mathrm{e}^s)$ is convex in the variable $s$ on the interval $\ln I \triangleq \{\ln x : x \in I\}$. Analogously, we introduce the notions of \emph{strict log-log-convexity} and \emph{log-log-linearity}.
\end{defn}
With these definitions we arrive at Lemma~\ref{lemma:equiv}, the central statement of this section.
%We will prove the result using Lemmas~\ref{lemma:cost:uniform:coboundary} \ref{lemma:spectral:radius:complex:circle}, \ref{lemma:log:log:convex}, and Corollaries~\ref{cor:cost:uniformity:continuum:solutions} and \ref{cor:log:log:linear} below. Fig.~\ref{fig:relationship:properties} depicts the roadmap for our subsequent derivations that establish the connections between the graph properties, spectral properties of $\ve{P}_G(x)$,  and the singularities of $\ve{F}_G(x,y)$.
%
%
\begin{lemma}\label{lemma:equiv}
	Let $G$ be a strongly connected graph. The following statements are equivalent.
	\begin{enumerate}
		\item[(a)] The graph $G$ has cost period $c$ .
		\item[(b)] The graph $G$ satisfies the $c$-periodic coboundary condition.
		\item[(c)] For any $x \in \mathbb{R}^+$, there are precisely $c$ solutions, $\phi_k = 2\pi k/c$ for $k\in \{0,1,\dots,c-1\}$, to the equation $\rho_G(x\mathrm{e}^{\imu\phi}) = \rho_G(x)$ in the interval $0\leq \phi<2\pi$.
	\end{enumerate}
	Furthermore, if $G$ is cost-uniform then the spectral radius $\rho_G(x)$ is log-log-linear on $x\in \mathbb{R}^+$. If $G$ is cost-diverse then the spectral radius $\rho_G(x)$ is strictly log-log-convex on $x\in \mathbb{R}^+$.
\end{lemma}
\input{fig/relationship-lemmas}
We will prove the result using a sequence of lemmas and their corollaries.  Fig.~\ref{fig:relationship:properties} depicts the roadmap for our subsequent derivations that establish the connections between the graph properties, spectral properties of $\ve{P}_G(x)$,  and, in Section~\ref{sec:multivariate:proof},  the singularities of $\ve{F}_G(x,y)$.

%Lemmas~\ref{lemma:cost:uniform:coboundary} and~\ref{lemma:spectral:radius:complex:circle}, \ref{lemma:log:log:convex}, and Corollaries~\ref{cor:cost:uniformity:continuum:solutions} and \ref{cor:log:log:linear} below. Fig.~\ref{fig:relationship:properties} depicts the roadmap for our subsequent derivations that establish the connections between the graph properties, spectral properties of $\ve{P}_G(x)$,  and the singularities of $\ve{F}_G(x,y)$.
\begin{remark}
	Lemma~\ref{lemma:cost:uniform:coboundary} establishes the equivalence $(a) \Leftrightarrow (b)$ in Lemma~\ref{lemma:equiv}. 
%\linebreak	
	Lemma~\ref{lemma:spectral:radius:complex:circle} proves $(a) \Rightarrow (c)$. Corollary~\ref{cor:cost:uniformity:continuum:solutions} proves $(c) \Rightarrow (b)$ through its contrapositive: If $G$ has precisely $c$ solutions to $\rho_G(x \mathrm{e}^{\imu\phi}) = \rho_G(x)$ on the unit circle, then $\rho_G(x \mathrm{e}^{\imu\phi}) = \rho_G(x)$ does not hold for all $x\in\mathbb{C}$ and $0\leq\phi<2\pi$. Then, by Corollary~\ref{cor:cost:uniformity:continuum:solutions}, it follows that $G$ does not satisfy the coboundary condition, hence it must satisfy the $c$-periodic couboundary condition for some $c$. Specifically, it must be fulfilled with the same $c$ as in statement $(c)$, as any other $c$ would lead to a contradiction. 
	%and Corollary \ref{cor:cost:uniformity:continuum:solutions} combine to prove the equivalences $(a) \Leftrightarrow (c)$ and $(b) \Leftrightarrow (c)$.  They state that for any strongly connected graph $G$ the equation $\rho_G(x \mathrm{e}^{\imu\phi}) = \rho_G(x)$ either has a finite number of solutions, $\phi_k = 2\pi k/c$, when $G$ is cost-diverse with cost period $c$, or holds for all $\phi$ when $G$ is cost-uniform. Because cost period zero refers to cost-uniformity, this means statement \emph{(c)} in Lemma \ref{lemma:equiv} covers the case of invariance of the spectral radius on the complex circle. Furthermore, we see that when $\rho_G(x \mathrm{e}^{\imu\phi}) = \rho_G(x)$ has an infinite number of solutions, then the set of such solutions comprise the full complex circle.
%\linebreak	
	Finally, Corollary~\ref{cor:log:log:linear} and Lemma~\ref{lemma:log:log:convex} establish the last two statements of Lemma~\ref{lemma:equiv}. 
	\linebreak
	The remaining Lemmas are either ingredients to these Lemmas or cover the cost-uniform case.
\end{remark}
%
%We proceed with proving the lemmas required for the derivation of Lemma~\ref{lemma:equiv}.

\subsection{Equivalence of Cost-Diversity and Coboundary Condition}
We start with proving the equivalence of cost-uniformity and the coboundary condition. For convenience, we say that two integers are congruent modulo $0$ if and only if they are equal. The following result is a generalization of the equivalence between the coboundary condition and cost-uniformity observed in \cite{liu_coding_2020} to arbitrary cost periods $c$.
\begin{lemma} \label{lemma:cost:uniform:coboundary}
	Let $G$ be a strongly connected graph. Then $G$ has cost period $c$ if and only if it fulfills the $c$-periodic coboundary condition. 
\end{lemma}
\begin{proof}
	We first show that \textbf{the $c$-periodic coboundary condition implies cost period $c$}. Let $\ve{p} = (e_1, e_2, \ldots, e_m)$ be a path from vertex $v_i$ to vertex $v_j$ 
	with path cost $\tau(\ve{p}) = \sum_{k=1}^m \tau(e_k)$. Suppose $\ve{p}$ is
	represented by the vertex sequence 
	$
	v_i=v_{i_0} \rightarrow v_{i_1} \rightarrow \cdots \rightarrow v_{i_m} = v_j.
	$
	The coboundary condition allows the path cost to be written as
	\begin{align*}
		\tau(\ve{p})&=\sum_{k=1}^m (b + B(v_{i_k})-B(v_{i_{k-1}})) + zc  \\
		&= mb + B(v_j) - B(v_i) + zc,
	\end{align*}
	for some integer $z \in \mathbb{Z}$. Thus, the costs of all paths of length $m$ that connect $v_i$ and $v_j$ are congruent modulo $c$ and, by definition, the graph $G$ has cost period $c$.
	
	We now show that \textbf{cost period $c$ implies the $c$-periodic coboundary condition}. We start by showing that there exists $b \in \mathbb{Q}$ such that the cost of any cycle $\ve{p}$ of length $m$ satisfies $\tau(\ve{p}) \equiv bm \pmod c$. Let $v_1 \in \V$ and let  $\ve{p}_1$ be a cycle at vertex $v_1$ of length $m_1$. Such a cycle exists by strong connectivity of the graph. Suppose $\tau(\ve{p}_1) = t_1$, let $v_2\in \V$, and let $\ve{p}_2$ be a cycle of length~$m_2$ at vertex $v_2$ with cost $\tau(\ve{p}_2)=t_2$. Denote by $g_{12}=\gcd(m_1,m_2)$ the greatest common divisor of $m_1$ and $m_2$. Strong connectivity of $G$ implies there is a path $\ve{p}_{1\to 2}$ from $v_1$ to $v_2$ with length $n \geq 1$ and cost $\tau(\ve{p}_{1\to 2})=t$. Define the path $\ve{p}$ comprising $m_2/g_{12}$ repetitions of cycle $\ve{p}_1$ followed by $\ve{p}_{1\to 2}$, and the path $\ve{p}'$ comprising $\ve{p}_{1\to 2}$ followed by $m_1/g_{12}$ repetitions of the cycle $\ve{p}_2$. The paths $\ve{p}$ and $\ve{p}'$ both have length $m_1m_2/g_{12} + n$. So, as $G$ has cost period $c$, $m_2 t_1/g_{12} +t = \tau(\ve{p})=\tau(\ve{p}') + zc= m_1 t_2/g_{12} +t + zc$ for some $z \in \mathbb{Z}$, implying that
	\begin{equation*}
		m_2t_1-m_1t_2 = g_{12}zc .
	\end{equation*}
	We then employ a variation of the Chinese Remainder Theorem described in Lemma~\ref{lemma:var:crt} below, which implies the existence of $b \in \mathbb{Q}$ such that any cycle $\ve{p}$ of length $m$ in $G$ has a cost \mbox{$\tau(\ve{p}) \equiv mb \pmod c$}, which is congruent to $mb$ modulo $c$. 
	
	Now, define a function $B: \V \rightarrow \mathbb{R}$ as follows. Set $B(v_1) = 0$. For a vertex $v_i\neq v_1$, choose a path $\ve{p}_{1\to i}$ from $v_1$ to $v_i$ of length $n \geq 1$ and define $B(v_i) = \tau(\ve{p}_{1\to i})-n b$. We claim that $B(v_i) \bmod c$ is independent of the chosen path $\ve{p}_{1\to i}$. To see this, suppose $\ve{p}_{1\to i}'$ and $\ve{p}_{1\to i}''$ are two such paths from $v_1$ to $v_i$ of length $n'$ and $n''$, respectively, and let $\ve{p}_{i\to1}$ be a path of length $p$ from $v_i$ to $v_1$.
	The cycle $\ve{p}'=(\ve{p}_{1\to i}', \ve{p}_{i\to1})$ has length $n'+p$, so $\tau(\ve{p}') = (n' + p)b + z'c$ where $z' \in \mathbb{Z}$. Similarly, the cycle $\ve{p}''=(\ve{p}_{1\to i}'', \ve{p}_{i\to1})$ has length $n''+p$ and cost $\tau(\ve{p}'') = (n'' + p)b + z''c$ for some $z''\in \mathbb{Z}$. Then
	$\tau(\ve{p}_{1\to i}') = \tau(\ve{p}') -\tau(\ve{p}_{i\to1}) = (n' + p)b+ z'c -\tau(\ve{p}_{i\to1})$ and \mbox{$\tau(\ve{p}_{1\to i}'') = \tau(\ve{p}'')-\tau(\ve{p}_{i\to1}) = (n'' + p)b+ z''c -\tau(\ve{p}_{i\to1})$}. It follows that 
	$$
	\tau(\ve{p}_{i\to1}) = (n'+p)b+ z'c - \tau(\ve{p}_{1\to i}') = (n''+p)b+ z''c - \tau(\ve{p}_{1\to i}''),
	$$
	from which we conclude
	$$
	\tau(\ve{p}_{1\to i}') - n' b  = \tau(\ve{p}_{1\to i}'')-n'' b + (z'-z'')c .
	$$
	This confirms that, by definition, $B(v_i) \bmod c$ is independent of the choice of path from $v_1$ to $v_i$.
	
	Finally, let $e\in \E$ be an edge from vertex $v_i$ to vertex $v_j$, and let $\ve{p}_{j\to1}$ denote a path from vertex $v_j$ to $v_1$ of length $q$. Consider the cycle $\ve{p}_1=(\ve{p}_{1\to i}, e, \ve{p}_{j\to1})$, with cost $\tau(\ve{p}_1)=(n+1+q)b + z_1c$ for some $z_1 \in \mathbb{Z}$.
	Noting that $\tau(\ve{p}_1)=\tau(\ve{p}_{1\to i}, e)+\tau(\ve{p}_{j\to1})$, and using the fact that \mbox{$\tau(\ve{p}_{1\to i}, e) = B(v_j)+(n+1)b + z_jc$} for some $z_j \in \mathbb{Z}$, we find
	\begin{align*}
		\tau(\ve{p}_{j\to1}) &= (n+1+q)b + z_1c - (B(v_j)+(n+1)b+z_jc)\\
		&= qb - B(v_j) + (z_1-z_j)c.
	\end{align*}
	We can also write $\tau(\ve{p}_1)=\tau(\ve{p}_{1\to i})+\tau(e)+\tau(\ve{p}_{j\to1})$, implying that 
	\begin{align*}
		\tau(e) &= \tau(\ve{p}_1) - (\tau(\ve{p}_{1\to i})+\tau(\ve{p}_{j\to1}))\\
		&= (n+1+q)b - ((B(v_i)+nb) + (qb - B(v_j))) \\
		& \quad\quad+ (z_j-z_i)c\\
		&=b +B(v_j) -B(v_i) + (z_j-z_i)c.
	\end{align*}
	This confirms that the $c$-periodic coboundary condition holds. 
\end{proof}
\subsection{Cost Period and Spectral Properties}
We next show that cost-diversity implies that there can only be a finite number of solutions to $\rho_G(x\mathrm{e}^{\imu\phi}) = \rho_G(x)$ over $0\leq\phi<2\pi$. In fact, we will prove a stronger statement: for all $x\in\mathbb{R}^+$, the solutions are exactly $\phi_k = 2\pi k/c$. This property is vital as it implies that the minimal singularities of the corresponding generating functions are be finitely minimal. We start with an auxiliary result on the structure of the cost-enumerator matrix.
\begin{lemma} \label{lemma:cost-enumerator:complex:circle}
	Let $G$ be a strongly connected graph with cost period $c$. Then, for all $x \in \mathbb{C}$ and $k\in\mathbb{Z}$,
	$$ \ve{P}_G\left(x\mathrm{e}^{2\pi\imu k/c}\right) =  \mathrm{e}^{2\pi\imu kb/c} \ve{D}_k^{-1} \ve{P}_G(x) \ve{D}_k,$$
	where $\ve{D}_k$ is a diagonal matrix with entries $[\ve{D}_k]_{jj} = \mathrm{e}^{2\pi\imu k B(v_j)/c}$, and $b$ and $B(v_j)$ are defined by the coboundary condition. Denoting by $\lambda_1(x),\dots,\lambda_{|\V|}(x)$ the eigenvalues of $\ve{P}_G(x)$, it holds that
	$$ \lambda_j\left(x \mathrm{e}^{2\pi\imu k/c}\right) = \mathrm{e}^{2\pi\imu kb/c } \lambda_j(x). $$
\end{lemma}
\begin{proof}
	The graph $G$ satisfies the $c$-periodic coboundary condition by Lemma~\ref{lemma:cost:uniform:coboundary}. Hence, there exists a constant $b$ and functions $B:\V\mapsto \R$ such that for any two vertices $v_i$ and $v_j$, each edge $e$ from $v_i$ to $v_j$ has cost $\tau(e)$, which can be written 
	$$
	\tau(e)= b + B(v_j) - B(v_i) + z_ec
	$$
	for some integer $z_e \in \mathbb{Z}$. Setting $\phi_k \triangleq 2\pi k/c$, Definition~\ref{def:cost:enumerator:spectral:radius} implies that the entries of the cost-enumerator matrix are given by
	\begin{align*}
		\left[\ve{P}_G\left(x \mathrm{e}^{\imu\phi_k}\right)\right]_{ij} &= \sum_{e \in \E:~ \begin{subarray}{l}\init(e)=v_i,\\ \final(e)=v_j\end{subarray}} x^{\tau(e)} \mathrm{e}^{\imu\phi_k \tau(e)} \\
		&=[\ve{P}_G(x)]_{ij}\mathrm{e}^{\imu\phi_k(b+B(v_j)-B(v_i))}.
	\end{align*}
	Introducing the diagonal matrix $\ve{D}_k$ with entries $[\ve{D}_k]_{ii} = \mathrm{e}^{\imu \phi_k B(v_i)}$, we can decompose the cost-enumerator matrix to
	$$ \ve{P}_G\left(x \mathrm{e}^{\imu\phi_k}\right) = \mathrm{e}^{\imu \phi_k b} \ve{D}_k^{-1} \ve{P}_G(x) \ve{D}_k. $$
	The second part of the lemma directly follows from the similarity\footnote{Two square matrices $\ve{A}$ and $\ve{B}$ are \emph{similar} if there exists an invertible diagonal matrix $\ve{D}$ such that $\ve{A} = \ve{D}^{-1}\ve{B}\ve{D}$. Similar matrices have the same eigenvalues with the same multiplicities \cite[Cor. 1.3.4]{horn_matrix_2012}.} of the matrices $\ve{P}_G(x \mathrm{e}^{\imu\phi_k})$ and $\mathrm{e}^{\imu\phi_kb}\ve{P}_G(x)$ proven in the first part of the lemma.
\end{proof}
We now continue with another auxiliary result that will serve to prove the subsequent result on the spectral structure of $\ve{P}_G(x)$ on the complex circle. Note that the proof of this result is conceptually related to the proof in \cite[Prop. 3.8]{marcus_introduction_2001} that an irreducible graph is aperiodic (i.e., has period 1) if and only if it is primitive.
\begin{lemma} \label{lemma:cost:period:existence:paths}
	Let $G$ be a strongly connected and cost-diverse graph with  cost period $c$. Then there exist two equal length cycles at the same vertex whose cost difference is precisely $c$.
\end{lemma}
\begin{proof}
	For convenience, for a path $\ve{p}$ we write $\init(\ve{p})$ for the initial vertex of its first edge and $\final(\ve{p})$ for the terminal vertex of its last edge. Since $c$ is the cost period of $G$, Definition~\ref{def:cost:period} guarantees the existence of $\eta \in \N$ pairs of paths $\ve{p}_i$ and $\ve{p}'_i$ for $1\leq i\leq\eta$ such that: (1) $\ve{p}_i$ and $\ve{p}'_i$ share the same length $m_i$, (2) $\ve{p}_i$ and $\ve{p}'_i$ start in the same vertex $ v_i\triangleq\init(\ve{p}_i)=\init(\ve{p}'_i) $ and end in the same vertex $ u_i\triangleq\final(\ve{p}_i)=\final(\ve{p}'_i) $, and (3)  the greatest common divisor of their cost differences $\tau(\ve{p}_i)-\tau(\ve{p}_i')$ is $c$. Hence, by Bézout's identity, there exist (possibly negative) integers $z_i \in \mathbb{Z}$ such that
	$$ \sum_{i=1}^{\eta} (\tau(\ve{p}_i)-\tau(\ve{p}_i'))z_i = c $$
	For each $i$, choose an arbitrary path $\ve{p}_{u_i\to v_i}$ that connects $u_i$ and $v_i$ and construct two cycles $\ve{\Gamma}_i = (\ve{p}_i,\ve{p}_{u_i\to v_i})$ and $\ve{\Delta}_i = (\ve{p}_i',\ve{p}_{u_i\to v_i})$ that share the same return path $\ve{p}_{u_i\to v_i}$ from $u_i$ to $v_i$. Further choose arbitrary paths $\ve{q}_{i}$, $1\leq i\leq \eta$ connecting $v_i$ and $v_{i+1}$ and $v_\eta$ and $v_1$. Now, set $\mu_i = \max \{z_i,0\}$ and $\mu_i' = \mu_i-z_i$, denote by $\ve{\Gamma}_i^{\mu}$ for $\mu \in \N_0$ the $\mu-$fold repetition of the cycle $\ve{\Gamma}_i$, and construct two large cycles $\ve{\Gamma}$ and $\ve{\Delta}$ by 
	\begin{align*}
		\ve{\Gamma} &= ( \ve{\Gamma}_1^{\mu_1},\ve{\Delta}_1^{\mu_1'},\ve{q}_{1}, \ve{\Gamma}_2^{\mu_2},\ve{\Delta}_2^{\mu_2'},\ve{q}_{2}, \dots, \ve{q}_{\eta}), \\
		\ve{\Delta} &= ( \ve{\Gamma}_1^{\mu_1'},\ve{\Delta}_1^{\mu_1},\ve{q}_{1}, \ve{\Gamma}_2^{\mu_2'},\ve{\Delta}_2^{\mu_2},\ve{q}_{2}, \dots, \ve{q}_{\eta}).
	\end{align*}
	In other words, $\ve{\Gamma}$ starts at $v_1$, circles $\mu_1$ times along $\ve{\Gamma}_1$, then $\mu_1'$ times along $\ve{\Delta}_1$, then proceeds to move along $\ve{q}_1$ to $v_2$. There it circles $\mu_2$ times along $\ve{\Gamma}_2$ and $\mu_2'$ times along $\ve{\Delta}_2$, and so on, until it moves back from $v_\eta$ to $v_1$ along $\ve{q}_\eta$. The cycle $\ve{\Delta}$ is created similarly. For a visualization of the construction of the cycle $\ve{\Gamma}$, see Fig.~\ref{fig:cycle:paths}. Notice that $\mu_i$ and $\mu_i'$ are guaranteed to be non-negative by their definitions. Computing the cost difference of $\ve{\Gamma}$ and $\ve{\Delta}$, one obtains
	\begin{align*}
		\tau(\ve{\Gamma}) &- \tau(\ve{\Delta}) \\
		&= \sum_{i=1}^{\eta} (\mu_i\tau(\ve{p}_i) +\mu_i'\tau(\ve{p}_i')) 
		- \sum_{i=1}^{\eta} (\mu_i'\tau(\ve{p}_i) +\mu_i\tau(\ve{p}_i')) \\
		&= \sum_{i=1}^{\eta} (\tau(\ve{p}_i)-\tau(\ve{p}_i'))z_i = c.
	\end{align*}
	Hence, there exist two cycles at the vertex $v_1$ of the same length $m$ whose cost is precisely $c$. 
\end{proof}
\input{fig/paths-circles}
Lemmas~\ref{lemma:cost-enumerator:complex:circle} and \ref{lemma:cost:period:existence:paths} can be combined to prove the following result on the structure of the spectral radius on the complex circle.
\begin{lemma} \label{lemma:spectral:radius:complex:circle}
	Let $G$ be a strongly connected and cost-diverse graph with cost period $c$. Then, for any $x \in \mathbb{R}^+$, there are precisely $c$ solutions $\phi_k = 2\pi k/c$ for $k\in \{0,1,\dots,c-1\}$ to the equation $\rho_G(x\mathrm{e}^{\imu\phi}) = \rho_G(x)$ in the interval $0\leq \phi<2\pi$. For all other $\phi$ the inequality $\rho_G(x\mathrm{e}^{\imu\phi}) < \rho_G(x)$ holds.
\end{lemma}
\begin{proof}
	By Lemma~\ref{lemma:cost-enumerator:complex:circle}, for all $k\in \{0,1,\dots,c-1\}$ and $1 \leq j \leq |\V|$ we have $ \lambda_j(x \mathrm{e}^{\imu\phi_k}) = \mathrm{e}^{\imu \phi_k b} \lambda_j(x),$ which implies that $\rho_G(x\mathrm{e}^{\imu\phi_k}) = \rho_G(x)$.
	
	We proceed with proving that for all other values of $\phi$ the spectral radius $\rho_G(x\mathrm{e}^{\imu\phi})$ is strictly less than $\rho_G(x)$. We start with the observation that, for any $\phi \in \R$, we have
	\begin{align*}
		|[\ve{P}_G(x\mathrm{e}^{\imu\phi})]_{ij}| &= \left| \sum_{e \in \E: \init(e)=v_i,\final(e)=v_j} (x\mathrm{e}^{\imu\phi})^{\tau(e)} \right| \\
		&\leq \sum_{e \in \E: \init(e)=v_i,\final(e)=v_j} x^{\tau(e)} = [\ve{P}_G(x)]_{ij}.
	\end{align*}
	By Wielandt's theorem~\cite[Sec.~8.3]{meyer_matrix_2000} (Theorem~\ref{thm:wielandt}),  it follows that the spectral radius satisfies $\rho_G(x\mathrm{e}^{\imu\phi}) \leq \rho_G(x)$, with equality if and only if there exist $\theta,\theta_1,\theta_2,\dots,\theta_{|\V|}$ such that
	$$ \ve{P}_G(x\mathrm{e}^{\imu\phi}) = \mathrm{e}^{\imu\theta} \ve{D}^{-1} \ve{P}_G(x) \ve{D}, $$
	where $\ve{D}$ is a diagonal matrix with entries $[\ve{D}]_{jj} = \mathrm{e}^{\imu \theta_j}$. It therefore suffices to prove that this equality can not be fulfilled for any $0\leq\phi<2\pi$ that is not equal to some $\phi_k$. Raising the above equation to the power $m \in \N$, it follows that
	$$ \ve{P}_G^m(x\mathrm{e}^{\imu\phi}) = \mathrm{e}^{\imu m \theta} \ve{D}^{-1} \ve{P}_G^m(x) \ve{D}. $$
	In particular, the entry $i,j$ of this equation reads as
	$$[\ve{P}_G^m(x\mathrm{e}^{\imu\phi})]_{ij} = \mathrm{e}^{\imu (m \theta+\theta_{j}-\theta_{i})} [\ve{P}_G^m(x)]_{ij}, $$
	and it follows that
	$$ |[\ve{P}_G^m(x\mathrm{e}^{\imu\phi})]_{ij}| = |[\ve{P}_G^m(x)]_{ij}| = [\ve{P}_G^m(x)]_{ij}. $$
	Denote now by $\mathcal{P}_{ij}(m) = \{ \ve{p} = (e_1,\dots,e_m) : \init(e_1) = v_i, \final(e_m) = v_j\}$ the set of paths of length $m$ from $v_i$ to $v_j$. It is well known \cite{khandekar_discrete_2000} that $[\ve{P}_G^m(x)]_{ij} = \sum_{\ve{p} \in \mathcal{P}_{ij}(m)} x^{\tau(\ve{p})}$. By Lemma~\ref{lemma:cost:period:existence:paths}, for a graph with cost period $c$ there exists a length $m$ and a vertex $v_i$ such that there are two cycles of length $m$ at $v_i$ whose cost differs by exactly $c$. Let $m$ and $v_i$ fulfill this property and denote by $\tau$ and $\tau+c$ the costs that are assumed by these two cycles.
	Thus, the polynomial $[\ve{P}_G^m(x)]_{ii}$ contains the sum of at least two monomials $x^{\tau}$ and $x^{\tau+c}$, each with integer-valued coefficients. Now, recall that the triangle inequality of a sum of complex numbers is tight if and only if the complex angles of all summands agree. Therefore, if $2\pi \phi c$ is not an integer multiple of $2\pi$ then $|[\ve{P}_G^m(x\mathrm{e}^{\imu\phi})]_{ii}| < [\ve{P}_G^m(x)]_{ii}$ and the claim follows.
\end{proof}
Conversely, for cost-uniform graphs, the eigenvalues of $\ve{P}_G(x)$ have a special structure that can be derived explicitly. 
\begin{lemma}  \label{lemma:coboundary:perron:root}
	Let $G = (\V,\E,\sigma,\tau)$ be a strongly connected graph that satisfies the coboundary condition. Then, for all $1 \leq j \leq |\V|$ and $x \in \mathbb{C}$, the eigenvalues $\lambda_1(x),\dots,\lambda_{|\V|}(x)$ of $\ve{P}_G(x)$ have the form
	$$ \lambda_j(x) = \lambda_j(1)x^b, $$
	where $b$ is the constant of the coboundary condition.
\end{lemma}
\begin{proof}
	The graph $G$ satisfies the coboundary condition by assumption. Hence, there exists a constant $b$ and functions $B:\V\mapsto \R$ such that, for any two vertices $v_i$ and $v_j$, each edge from $v_i$ to $v_j$ has cost 
	$$
	\tau_{ij}= b + B(v_j) - B(v_i).
	$$
	Furthermore, the number of edges from $v_i$ to $v_j$ is precisely $[\ve{P}_G(1)]_{ij}$. It follows that each entry $[\ve{P}_G(x)]_{ij}$ of $\ve{P}_G(x)$ is equal to
	$$[\ve{P}_G(x)]_{ij} = [\ve{P}_G(1)]_{ij}x^{\tau_{ij}} = [\ve{P}_G(1)]_{ij}x^{b + B(v_j) - B(v_i)}$$
	and, introducing the diagonal matrix $\ve{D}(x)$ with entries $[\ve{D}(x)]_{ii} = x^{B(v_i)}$, we can decompose the cost-enumerator matrix as
	$$ \ve{P}_G(x) = x^b \ve{D}^{-1}(x) \ve{P}_G(1) \ve{D}(x). $$
	Thus, the characteristic polynomial $\phi(\lambda,x)$ of the cost-enumerator matrix $\ve{P}_G(x)$ becomes
	\begin{align*}
		&\det(\lambda \ve{I}- \ve{P}_G(x)) \\
		=& \det(\lambda \ve{I}-x^b \ve{D}^{-1}(x) \ve{P}_G(1) \ve{D}(x)) \\
		=& \det(\ve{D}) \det(\lambda \ve{I}-x^b \ve{D}^{-1}(x) \ve{P}_G(1) \ve{D}(x)) \det{\ve{D}^{-1}} \\
		\overset{(a)}{=}& \det(\lambda \ve{I}-x^b  \ve{P}_G(1)),
	\end{align*}
	where in $(a)$ we used the multiplicativity of the determinant.	Thus, $\phi(\lambda,x) = x^{b|\V|} \phi(x^{-b}\lambda,1)$. Since the eigenvalues of $\ve{P}_G(x)$ are precisely the roots of the characteristic polynomial, we can identify $\lambda_j(x)$ as the roots of $\phi(\lambda,x)$ and $\lambda_j(1)$ as the roots of $\phi(\lambda,1)$. By a variable substitution, it follows that $\lambda_j(x) = \lambda_j(1) x^b$ for all $1 \leq j \leq |\V|$ and $x \in \mathbb{C}$.
\end{proof}
Lemma~\ref{lemma:coboundary:perron:root} illustrates that the eigenvalues of cost-uniform graphs have the very special structure of being monomials in $x$. All eigenvalues share the same exponent $b$ from the coboundary condition and their coefficient is given by the corresponding eigenvalue of the matrix $\ve{P}_G(1)$. The following example illustrates Lemma~\ref{lemma:coboundary:perron:root}.
\begin{example}
	Consider the cost-uniform graph from Fig.~\ref{fig:example:complete:non:cost:diverse} on Page~\pageref{fig:example:complete:non:cost:diverse}. We can verify, by analyzing the cost of the edges which are self-loops, that the constant from the coboundary condition is given by $b=2$. Computing the eigenvalues, we obtain $\lambda_1(x) = 2x^2$ and $\lambda_2(x) = 0$, confirming the statement from Lemma~\ref{lemma:coboundary:perron:root}.
\end{example}
This puts us in the position to prove the converse to Lemma~\ref{lemma:spectral:radius:complex:circle}. That is, we can show that if a graph is cost-uniform, or equivalently satisfies the coboundary condition, then the spectral radius is invariant on the complex circle.
\begin{corollary} \label{cor:cost:uniformity:continuum:solutions}
	Let $G$ be a strongly connected graph that satisfies the coboundary condition. Then $\rho_G(x\mathrm{e}^{\imu\phi}) = \rho_G(x)$ for all $x \in \mathbb{C}$ and $0\leq \phi<2\pi$.
\end{corollary}
\begin{proof}
	The corollary directly follows from Lemma~\ref{lemma:coboundary:perron:root}, using the fact that the coboundary condition implies that for all $x\in \mathbb{C}$,
	$$ \rho_G(x\mathrm{e}^{\imu\phi}) = \rho_G(1) |x\mathrm{e}^{\imu\phi}|^b =  \rho_G(1) |x|^b = \rho_G(x). $$
\end{proof}
\subsection{Cost-Diversity and Strict Log-Log-Convexity}
We conclude this section with a discussion of the log-log-convexity of the spectral radius. This property will help in several places to prove Theorem~\ref{thm:variable-length:exact}. First, we show that for cost-uniform graphs the spectral radius is log-log-linear on the real axis.
\begin{corollary} \label{cor:log:log:linear}
	Let $G$ be a strongly connected graph. If $G$ is cost-uniform then $\rho_G(x)$ is log-log-linear on the interval $x\in \mathbb{R}^+$.
\end{corollary}
\begin{proof}
	By Lemma~\ref{lemma:coboundary:perron:root}, $ \ln \rho_G(\mathrm{e}^s) = \ln (\rho_G(1)) + bs$ for all real-valued $s \in \mathbb{R}$, which is a linear function in $s$.
\end{proof}
We now turn towards proving the converse to the previous corollary, showing that if the graph $G$ is cost-diverse then $\rho_G(x)$ is strictly log-log-convex. Notice that the (non-strict) log-log convexity of the Perron root is known from classical results on irreducible matrices \cite{cohen_convexity_1981,cohen_derivatives_1978,kingman_convexity_1961}.
\begin{lemma} \label{lemma:log:log:convex}
	Let $G$ be a strongly connected, cost-diverse graph. Then $\rho_G(x)$ is strictly log-log-convex for all $x \in \mathbb{R}^+$.
\end{lemma}
Our proof of Lemma~\ref{lemma:log:log:convex} makes use of the following result.
\begin{lemma}[{{\cite[Thm. 1.37]{stanczak_resource_2006}}}] \label{lemma:irreducible:matrix:strict:log:convexity} 
	Let $\ve{P}(s)$ be an irreducible matrix whose nonzero entries are log-convex functions of $s \in \R$. Then the spectral radius $\rho(s)$ of $\ve{P}(s)$ is log-convex. If, additionally, at least one entry of $\ve{P}(s)$ is strictly log-convex, then $\rho(s)$ is strictly log-convex. 
\end{lemma}
\begin{proof}[Proof of Lemma~\ref{lemma:log:log:convex}]
	Consider the $m$-th power $\ve{P}_G^m(x)$ of $\ve{P}_G(x)$. We know from \cite{khandekar_discrete_2000} that, denoting $\mathcal{P}_{ij}(m)$ as the set of paths of length $m$ from $v_i$ to $v_j$, the entry $i,j$ of the matrix $\ve{P}_G^m(x)$ is given by $[\ve{P}_G^m(x)]_{ij} = \sum_{\ve{p} \in \mathcal{P}_{ij}(m)} x^{\tau(\ve{p})}$. We will show that this entry is strictly log-log-convex  if there exist two paths of length $m$ from $v_i$ to $v_j$ with different costs. Taking the second derivative of the log-log expression, we obtain
	\begin{align*}
		&\frac{\partial^2}{\partial s^2}\ln ([\ve{P}_G^m(\mathrm{e}^{s})]_{ij}) 
		\\
		&= \frac{\sum_{\ve{p}} \mathrm{e}^{s\tau(\ve{p})}\sum_{\ve{p}} \tau(\ve{p})^2 \mathrm{e}^{s\tau(\ve{p})} - \left(\sum_{\ve{p}}\tau(\ve{p}) \mathrm{e}^{s\tau(\ve{p})}\right)^2}{\left([\ve{P}_G^m(\mathrm{e}^s)]_{ij}\right)^2}.
	\end{align*}
	Identifying the vectors $\ve{u} = (\mathrm{e}^{s\tau(\ve{p})/2}:\ve{p}\in \mathcal{P}_{ij}(m))$ and $\ve{v} = (\tau(\ve{p})\mathrm{e}^{s\tau(\ve{p})/2}:\ve{p}\in \mathcal{P}_{ij}(m))$,  both of which have length $|\mathcal{P}_{ij}(m)|$, the numerator is equal to $(\ve{u}\cdot\ve{u})(\ve{v}\cdot \ve{v})-(\ve{u}\cdot\ve{v})^2$, where $\ve{u}\cdot\ve{v}$ denotes the inner product of the vectors $\ve{u}$ and $\ve{v}$. The numerator is therefore non-negative by the Cauchy-Schwarz inequality, see, e.g., \cite[Ch. 0.6.3]{horn_matrix_2012}, and thus the entries $[\ve{P}_G^m(x)]_{ij}$ are either $0$ or positive and log-convex. Furthermore, due to the cost-diversity of the graph $G$, there exist $m,i,$ and $j$ such that there exist two paths of length $m$ from $v_i$ to $v_j$ with different costs, and thus $\ve{u}$ and $\ve{v}$ are linearly independent. In this case, the Cauchy-Schwarz inequality holds with strict inequality and thus the numerator is positive, which implies that $[\ve{P}_G^m(\mathrm{e}^{s})]_{ij}$ is strictly convex in $s$.
	
	The spectral radius of $\ve{P}_G^m(\mathrm{e}^s)$ is given by $\rho_G^m(\mathrm{e}^s)$. With Lemma~\ref{lemma:irreducible:matrix:strict:log:convexity} it follows that $\rho_G^m(\mathrm{e}^s)$ is strictly log-convex. Since raising to a positive integer power does not change log-convexity, $\rho_G(\mathrm{e}^s)$ is also strictly log-convex, so $\rho_G(x)$ is strictly log-log-convex.
\end{proof}
%

%% file: fig/relationship-lemmas.tex
\begin{figure*}
	\centering
	\begin{tikzpicture}

		\node[block, text width=3cm] (cp) {Cost period (Definition~\ref{def:cost:period})};
		\node[block, right=5.5cm of cp, text width=3cm] (cc) {Coboundary condition (Definition~\ref{def:coboundary:condition})};
		\node[block, below=3.5cm of cc.center, text width=3cm, anchor=center] (srcc) {Spectral radius complex circle};
		\node[block, below=3.5cm of cp.center, text width=3cm, anchor=center] (llc) {Strict log-log-convexity (Definition~\ref{def:log:log:convexity})};
		
		\node[block, below =3.5cm of llc.center, text width=3cm, anchor=center] (nd) {Non-degeneracy (Definition~\ref{def:nondegenerate})};
		\node[block, right=1.15cm of nd, text width=3cm] (sc) {Smoothness and criticality (Definition~\ref{def:critical:smooth})};
		\node[block, below=3.5cm of srcc.center, text width=3cm, anchor=center] (fm) {Finite minimality (Definition~\ref{def:minimal:singularity})};
		
		\node[left =1cm of cp, text width=2cm, align=center] {Graph $G$};
		\node[left =1cm of llc, text width=2cm, align=center] {Spectral properties of $\ve{P}_G(x)$};
		\node[left =1cm of nd, text width=2cm, align=center] {Singularities of $\ve{F}_G(x,y)$};
		
		\draw[->] (cp) to[bend left=10] node[above] {Lemma~\ref{lemma:cost:uniform:coboundary}} (cc)  ;
		\draw[->] (cc) to[bend left=10] node[below] {Lemma~\ref{lemma:cost:uniform:coboundary}} (cp)  ;
		\draw[->] (srcc) to[bend right=10] node[right,pos=0.45] {Corollary~\ref{cor:cost:uniformity:continuum:solutions}}
		 node[right,pos=0.25]{(contrapositive)}  (cc)  ;
		\draw[->] (cp) to[bend right=10,in=190] node[above right=0cm and 0cm,pos=0.6,text width=2.5cm] {Lemmas~\ref{lemma:spectral:radius:complex:circle}~\&~\ref{lemma:cost-enumerator:complex:circle}} (srcc)  ;
		%\draw[->] (cp) to[bend left=10] node[right,pos=0.75] {Lemma~\ref{lemma:log:log:convex}}  (llc)  ;
		\draw[->] (cp) to[bend right=10] node[left,pos=0.75] {Lemma~\ref{lemma:log:log:convex}} (llc)  ;
		
		\draw[->] (llc) to[bend right=10] node[left,pos=0.6] {Lemma~\ref{lemma:singularity:nondegenerate}} (nd)  ;
		\draw[->] (srcc) to[bend right=10] node[left,pos=0.65] {Lemma~\ref{lemma:minimal:singularities}} (fm)  ;
		\draw[->] (srcc) to[bend right=10] node[left=1cm,pos=0.75, text width = 2.5cm] {Lemmas~\ref{lemma:singularities:critical:smooth}~\&~\ref{lemma:alpha:boundary:unique:solution}} (sc)  ;
		
		\draw[dashed] ($(cp) + (-4.75,-1.35)$) -- ($(cp) + (10.5,-1.35)$);
		\draw[dashed] ($(llc) + (-4.75,-1.35)$) -- ($(llc) + (10.5,-1.35)$);
		
	\end{tikzpicture}
	\caption{Relationships between the properties of a \textcolor{blue}{cost-diverse} strongly connected graph, the nature of the cost-enumerator spectrum, and the singularities of the corresponding generating function.}
	\label{fig:relationship:properties}
\end{figure*}

%% file: fig/paths-circles.tex
\begin{figure*}
	\centering
	\begin{tikzpicture}
		\node[bullet, label=below:$v_1$] (v1) {};
		\node[bullet, right = 2.6cm of v1, label=below:$v_2$] (v2) {};
		\node[bullet, right = 2.6cm of v2, label=below:$v_3$] (v3) {};
		\node[bullet, right = 3cm of v3, label=below:$v_{\eta-1}$] (v4) {};
		\node[bullet, right = 2.6cm of v4, label=below:$v_\eta$] (v5) {};
		\node[bullet, right = 2.6cm of v5, label=below:$v_1$] (v6) {};
		\node[right = 1.15cm of v3] {$\dots$};
		\coordinate[right = .5cm of v3] (c1);
		\coordinate[left = .5cm of v4] (c2);
		
		\draw[fill=black,fill opacity=0.05,postaction={decorate,decoration={markings, mark=at position 0.25 with {\arrow[fill opacity=1]{>}}}}]
		($(v1)+(0,0.75)$) circle (0.75);
		\draw[fill=black,fill opacity=0.05,postaction={decorate,decoration={markings, mark=at position 0.25 with {\arrow[fill opacity=1]{>}}}}]
		($(v2)+(0,0.9)$) circle (0.9);
		\draw[fill=black,fill opacity=0.05,postaction={decorate,decoration={markings, mark=at position 0.25 with {\arrow[fill opacity=1]{>}}}}]
		($(v5)+(0,0.8)$) circle (0.8);
		\draw[fill=black,fill opacity=0.05,postaction={decorate,decoration={markings, mark=at position 0.75 with {\arrow[fill opacity=1]{<}}}}]
		($(v1)-(0,0.75)$) circle (0.75);
		\draw[fill=black,fill opacity=0.05,postaction={decorate,decoration={markings, mark=at position 0.75 with {\arrow[fill opacity=1]{<}}}}]
		($(v2)-(0,0.9)$) circle (0.9);
		\draw[fill=black,fill opacity=0.05,postaction={decorate,decoration={markings, mark=at position 0.75 with {\arrow[fill opacity=1]{<}}}}]
		($(v5)-(0,0.8)$) circle (0.8);
		
		\node[above = .4cm of v1] {$\ve{\Gamma}_1$};
		\node[above = .5cm of v2] {$\ve{\Gamma}_2$};
		\node[above = .4cm of v5] {$\ve{\Gamma}_\eta$};
		\node[below = .4cm of v1] {$\ve{\Delta}_1$};
		\node[below = .6cm of v2] {$\ve{\Delta}_2$};
		\node[below = .45cm of v5] {$\ve{\Delta}_\eta$};

		\node[above = 1.5cm of v1] {$\mu_1$};
		\node[above = 1.8cm of v2] {$\mu_2$};
		\node[above = 1.6cm of v5] {$\mu_\eta$};
		\node[below = 1.5cm of v1] {$\mu_1'$};
		\node[below = 1.8cm of v2] {$\mu_2'$};
		\node[below = 1.6cm of v5] {$\mu_\eta'$};
		
		\draw[decoration={markings, mark=at position 0.5 with {\arrow{>}}}, postaction={decorate}] (v1) -- node[midway, below] {$\ve{q}_{1}^{}$} (v2);
		\draw[decoration={markings, mark=at position 0.5 with {\arrow{>}}}, postaction={decorate}] (v2) -- node[midway, below] {$\ve{q}_{2}^{}$} (v3);
		\draw[decoration={markings, mark=at position 0.5 with {\arrow{>}}}, postaction={decorate}] (v4) -- node[midway, below] {$\ve{q}_{{\eta-1}}^{}$} (v5);
		\draw[decoration={markings, mark=at position 0.5 with {\arrow{>}}}, postaction={decorate}] (v5) -- node[midway, below] {$\ve{q}_{\eta}^{}$} (v6);

		\draw[-] (v3) -- (c1);
		\draw[-] (c2) -- (v4);
	\end{tikzpicture}
	\caption{Construction of the path $\ve{\Gamma}$ in the proof of Lemma~\ref{lemma:cost:period:existence:paths}.}
	\label{fig:cycle:paths}
\end{figure*}

%% file: multivariate-singularity-analysis.tex
%!TEX root = bare_jrnl.tex
\section{Multivariate Singularity Analysis} \label{sec:multivariate:proof}
The main step in proving Theorem~\ref{thm:fixed-length} is showing that the prerequisites of \cite[Thm. 5.1 and 9.1]{melczer_invitation_2021} are fulfilled, which requires exhibiting certain properties of the singularities of $F_{G,v}$. We start by deriving the generating function and reviewing the properties of the generating function required to understand \cite[Thm. 5.1 and 9.1]{melczer_invitation_2021}. Afterwards, we prove that these properties apply to the generating function of $N_{G,v}(t,n)$, the size of the limited-cost follower sets.
\subsection{Derivation of the Generating Function} \label{sec:derivation:gen}
The beginning point of the multivariate singularity analysis is the derivation of the generating function of the sequence $N_{G,v}(t,n)$. Together with the detailed analysis of the singularities of the generating function in the proceeding sections, this will allow us to use the powerful machinery of analytical combinatorics in several variables. The following recursion is the key observation we need to derive the generating function of the series $N_{G,v}(t,n)$.

\begin{lemma} \label{lemma:recursion}
	Let $G = (\V,\E,\tau,\sigma)$ be a strongly connected, deterministic graph. Then the size of the follower set of any vertex $v \in \V$ obeys the recursion
	$$ N_{G,v}(t,n) = \sum_{e\in \E: \init(e)= v } N_{G,\final(e)}(t-\tau(e),n-1), $$
	for all $n>0$ and $t\geq 0$, where $N_{G,v}(t,0) = 1$ for all $t\geq 0$ and $N_{G,v}(t,n) = 0$ whenever $t<0$ or $n<0$.
\end{lemma}
\begin{proof}
	Denote by $\mathcal{P}_{G,v}(t,n)$ the set of all length-$n$ paths in $G$ that start from vertex $v$ and have cost at most $t$. By the deterministic property of the graph, $N_{G,v}(t,n) = |\mathcal{P}_{G,v}(t,n)|$. Partition the paths $\mathcal{P}_{G,v}(t,n)$ according to the first traversed edge $e \in \E$ into the distinct parts $\mathcal{P}_{G,v,e}(t,n) = \{ \ve{p} \in \mathcal{P}_{G,v}(t,n): \ve{p} = (e,e_2,\dots) \}$ for any  
	$e \in \E$ that emanates from $v$, i.e., $\init(e)=v $. To start with, $\mathcal{P}_{G,v}(t,n) = \bigcup_{e\in \E } \mathcal{P}_{G,v,e}(t,n)$ and the parts are distinct by definition. Now, any path $\ve{p} \in \mathcal{P}_{G,v,e}(t,n)$ starts by traversing $e$, which costs $\tau(e)$ and results in the vertex $\final(e)$. Therefore, each path $\ve{p} \in \mathcal{P}_{G,v,e}(t,n)$ can be assembled by prepending $e$ to some path of cost at most $T-\tau(e)$ and length $n-1$ that starts from $\final(e)$, i.e., $\mathcal{P}_{G,v,e}(t,n) = \left\{ \ve{p} = (e,\ve{p}') : \ve{p}' \in \mathcal{P}_{G,\final(e)}(t{-}\tau(e),n{-}1) \right\}.$ 

	Thus $|\mathcal{P}_{G,v,e}(t,n)| = N_{G,\final(e)}(t-\tau(e),n-1)$, which proves the recursive statement of the lemma. The initial condition $N_{G,v}(t,0) = 1$ for all $t\geq 0$ comes from the fact that we include the length $0$ string in our computations.
\end{proof}
This recursion allows us to derive the exact generating function of the integer sequences $N_{G,v}(t,n)$. Furthermore, we can extract the asymptotic behavior of integer sequences by means of powerful methods in complex analysis \cite{pemantle_twenty_2008,melczer_invitation_2021}. Note that here we restrict our attention to the sequence $N_{G,v}(t,n)$, which directly implies the generating function for $N_{G,v}(t)$ because we have $N_{G,v}(t) = \sum_{n \geq 0}N_{G,v}(t,n)$. We proceed with the proof of Lemma~\ref{lemma:generating:function}.

\begin{proof}[Proof of Lemma~\ref{lemma:generating:function}]
	Starting from the recursive expression of $N_{G,v}(t,n)$, we first incorporate the beginning of the recursion and obtain
	$$ N_{G,v}(t,n) = \hspace{-1em}\sum_{e\in \E: \init(e)= v } \hspace{-1.2em}N_{G,\final(e)}(t-\tau(e),n-1) + U(t,n), $$
	where $U(t,n) = 1$ if $n=0$ and $t\geq 0$, and $U(t,n) =0$ otherwise. Multiplying by $x^ty^n$ on both sides and summing over $n$ and $t$ yields
	\begin{align*}
		F_{G,v}(x,y) 
		&=  \sum_{e\in \E: \init(e)= v } \hspace{-1.2em}x^{\tau(e)}y F_{G,\final(e)}(x,y)  + \frac{1}{1-x},
	\end{align*}
	where we used that $N_{G,v}(t,n) = 0$ for any $t<0$ or $n<0$ and the fact that  $\sum_{t\geq0}x^t = 1/(1-x)$.
	Combining the generating functions of all vertices into one vector $\ve{F}_G(x,y)$, we obtain
	$$ \ve{F}_G(x,y) = y\ve{P}_G(x)\ve{F}_G(x,y) + \frac{1}{1-x}\ve{1}^\mathrm{T},$$
	and rearranging the above equality yields the claimed recursion.
\end{proof}
\subsection{Analytic Combinatorics in Several Variables} \label{sec:acsv}
We briefly review the ingredients required to invoke ACSV results \cite{melczer_invitation_2021}. For reasons of clarity we present definitions for the bivariate case, where we wish to compute the asymptotic behavior of $N(\alpha_1 t, \alpha_2 t)$ as $t \to \infty$. Notice that in our setup $\alpha_1=1$ and $\alpha_2=\alpha$. Furthermore, we assume that the generating function has the form $F(x,y) = Q(x,y)/H(x,y)$ for two polynomials $Q(x,y)$ and $H(x,y)$. We start with the notion of singularities of a generating function.

\begin{defn}[{\cite[Def. 3.5]{melczer_invitation_2021}}] \label{def:singularity}
	A point $(x_0,y_0) \in \mathbb{C}^2$ is called a \emph{singularity} of $F(x,y)$, if $F(x,y)$ is unbounded in any neighborhood around $(x_0,y_0)$.
\end{defn}
Similar to the univariate case, a sufficient condition for a point to be a singularity is that $H(x_0,y_0) = 0$ and $Q(x_0,y_0) \neq 0$, and when $Q$ and $H$ are coprime then the singularities of $F$ are precisely the zeroes of $Q$. An important property of singularities is the following.
\begin{defn}[{\cite[Def. 3.9]{melczer_invitation_2021}}] \label{def:minimal:singularity}
	A point $(x_0,y_0) \in \mathbb{C}^2$ is called a \emph{minimal singularity} of $F(x,y) = Q(x,y)/H(x,y)$  if it is a singularity of $F(x,y)$ and there exists no other singularity $(x',y')\in \mathbb{C}^2$ with $|x'|<|x|$ and $|y'|<|y|$.
\end{defn}
A minimal singularity is called \emph{finitely minimal} \cite[Def. 5.6]{melczer_invitation_2021}, if there exist only a finite number of singularities with the same coordinate-wise modulus. In contrast to the case of univariate generating functions, not all minimal singularities contribute to the asymptotic behavior of the sequence under consideration. The following notion of critical points helps to determine those singularities which are important for the asymptotic expansion.
\begin{defn}[{\cite[Def. 5.4]{melczer_invitation_2021}}] \label{def:critical:smooth}
	When $H(x,y)$ is square-free (has no repeated irreducible factors) then a point $(x_0,y_0) \in \mathbb{C}^2$ is a \emph{smooth critical point} of $F(x,y) = Q(x,y)/H(x,y)$ with respect to the direction $(\alpha_1,\alpha_2)$  if at least one of the partial derivatives $H_x(x_0,y_0)$ and $H_y(x_0,y_0)$ is nonzero and
	\begin{align*}
		H(x_0,y_0) = \alpha_2 x H_x(x_0,y_0)- \alpha_1 y H_y(x_0,y_0)=0,
	\end{align*}
	and is a \emph{non-smooth critical point} if 
	\[ H(x_0,y_0)=H_x(x_0,y_0)=H_y(x_0,y_0)=0. \]
	When $H(x,y)$ is not square-free then critical points are defined by replacing $H$ with its square-free part (the product of its distinct irreducible factors) in these equations.
\end{defn}
We need one final definition before describing the asymptotic results of ACSV.
\begin{defn}[{\cite[Def. 5.7, Prop. 5.2]{melczer_invitation_2021}}] \label{def:nondegenerate}
Let $(x_0,y_0) \in \mathbb{C}^2$ be a smooth critical point with respect to the direction $(\alpha_1,\alpha_2)$. Assume without loss of generality that $H_y(x_0,y_0) \neq 0$, and let $g(x)$ be the analytic function characterizing the singularities $(x,g(x))$ in a neighborhood around $(x_0,y_0)$. The point $(x_0,y_0)$ is called a \emph{nondegenerate critical point} if $\mathcal{H}_{x_0,y_0}$,  the value of the Hessian matrix (i.e., second derivative) $\mathcal{H}$ of
	$$ \phi(\theta) = \ln \left( \frac{g(x_0 \mathrm{e}^{\imu\theta})}{g(x_0)} \right) + \frac{\imu\theta \alpha_1}{\alpha_2}  $$
at $\theta=0$, is nonzero.
\end{defn}
\begin{remark}[{\cite[Lemma 5.5]{melczer_invitation_2021}}]\label{rem:nondegenerate:explicit}
An explicit characterization of nondegeneracy in terms of $H(x,y)$ can be obtained as follows. 
	%Let $(x_0,y_0) \in \mathbb{C}^2$ be a smooth critical point. Assume without loss of generality that $H_y(x_0,y_0) \neq 0$. 
	 We say that $(x_0,y_0)$ is \emph{nondegenerate with respect to the direction $(\alpha_1,\alpha_2)$}
when the quantity
\begin{align*}
\mathcal{H}_{x_0,y_0} &= \frac{\alpha_1(\alpha_1+\alpha_2)}{\alpha_2^2} + \frac{x_0^2H_{xx}(x_0,y_0)}{y_0 H_y(x_0,y_0)} 
 \nonumber\\& 
- 2\frac{\alpha_1 x_0 H_{xy}(x_0,y_0)}{\alpha_2H_y(x_0,y_0)} + \frac{\alpha_1^2 y_0 H_{yy}(x_0,y_0)}{\alpha_2^2H_y(x_0,y_0)} \end{align*}

\noindent
exists and is nonzero, where subscripted variables refer to partial derivatives.
\end{remark}
	
First, we present a theorem for `smooth' asymptotics, which applies when asymptotics are determined by a smooth critical singularity. We will apply this result in the regime when $\alpha_{G}^{\mathsf{lo}}<\alpha<\alpha_{G}^{\mathsf{up}}$.
\begin{thm}[{\cite[Thm. 5.1]{melczer_invitation_2021}}]
	\label{prop:ACSVsmooth}
	Let $\alpha_1,\alpha_2>0$ and let $Q(x,y)$ and $H(x,y)$ be coprime polynomials such that the generating function $F(x,y)=Q(x,y)/H(x,y)$ admits a power series expansion \mbox{$F(x,y) = \sum_{t,n \geq 0}N(t,n)x^ty^n$}. Suppose that the system of polynomial equations
	\begin{equation} H(x,y) = \alpha_2xH_x(x,y) - \alpha_1yH_y(x,y) = 0 \label{eq:CPeq2}\end{equation}
	admits a finite number of solutions, exactly one of which $(x_0,y_0)\in\mathbb{C}^2$ is minimal. Suppose further that $(x_0,y_0)$ has nonzero coordinates, $H_y(x_0,y_0)\neq0$, and $(x_0,y_0)$ is nondegenerate with respect to the direction $(\alpha_1, \alpha_2)$. Then, as $t\rightarrow\infty$,
	\begin{align} N(t\alpha_1,t\alpha_2) =&~ x_0^{-t\alpha_1}y_0^{-t\alpha_2} t^{-1/2} \frac{1}{\sqrt{2\pi\alpha_2 \mathcal{H}_{x_0,y_0}}} \nonumber\\&\cdot 
	\left( \frac{-Q(x_0,y_0)}{y_0 H_y(x_0,y_0)} + O\left(\frac{1}{t}\right)\right) \label{eq:easySmoothASM}
	\end{align}
	when $t(\alpha_1,\alpha_2)\in\N^2$.
\end{thm}
Theorem \ref{prop:ACSVsmooth} has been extended to the case  when the critical point equations~\eqref{eq:CPeq2} admit a finite set of minimal singularities, which all have the same coordinate-wise modulus. Provided that all such points fulfill the conditions of Theorem \ref{prop:ACSVsmooth}, an asymptotic expansion of $N(t\alpha_1,t\alpha_2)$ is obtained by summing the right-hand side of~\eqref{eq:easySmoothASM} over all such singularities \cite[Cor. 5.2]{melczer_invitation_2021}. In order to compute the asymptotic expansion in the smooth case, we thus need to verify the following properties. First, we need to characterize the minimal points that satisfy~\eqref{eq:CPeq2} and show that $H_y$ does not vanish at these points. Second, the points need to be nondegenerate and the numerator should be nonzero to guarantee a dominant asymptotic term.

The other case of interest is the \emph{multiple-point} case where two smooth branches of the singular set collide. In this case, the asymptotic behavior is obtained using the following theorem.

\begin{thm}[{\cite[Prop.~9.1 and Thm.~9.1]{melczer_invitation_2021}}]
	\label{prop:ACSVnonsmooth}
	Let $\alpha_1,\alpha_2>0$ and let $Q(x,y)$ and $H(x,y)$ be coprime polynomials such that $F(x,y)=Q(x,y)/H(x,y)$ admits a power series expansion $F(x,y) = \sum_{t,n \geq 0}N(t,n)x^ty^n$. Suppose that $(x_0,y_0)$ is a strictly minimal point, and near $(x_0,y_0)$ the zero set of $H(x,y)$ is locally the union of the sets defined by the vanishing of polynomials $R(x,y)$ and $S(x,y)$ such that $R(x_0,y_0)=S(x_0,y_0)=0$ and the gradients of $R(x,y)$ and $S(x,y)$  are linearly independent at $(x_0,y_0)$ (in particular, both gradients must be nonzero so each of the zero sets are locally smooth near $(x_0,y_0)$). If there exist $\nu_1,\nu_2>0$ such that
	\begin{equation*} 
		(\alpha_1,\alpha_2) = \nu_1 \left(1,\frac{y_0 R_y(x_0,y_0)}{x_0 R_x(x_0,y_0)}\right) + \nu_2 \left(1,\frac{y_0 S_y(x_0,y_0)}{x_0 S_x(x_0,y_0)} \right)
		%\label{eq:transvCP}
	\end{equation*}
	and the matrix
	\[ \mathbf{H} = \begin{pmatrix} x_0R_x(x_0,y_0) & y_0R_y(x_0,y_0) \\ x_0S_x(x_0,y_0) & y_0S_y(x_0,y_0)  \end{pmatrix} \]
	is nonsingular then, as $t\rightarrow\infty$,
	\[
	N(t\alpha_1,t\alpha_2) = x_0^{-t\alpha_1}y_0^{-t\alpha_2} \frac{Q(x_0,y_0)}{|\det \mathbf{H}|} + O\left(\delta^t\right)
	\]
	for some $0 < \delta < x_0^{-\alpha_1}y_0^{-\alpha_2}$.
\end{thm}
We will apply Theorem~\ref{prop:ACSVnonsmooth} when $0<\alpha<\alpha_{G}^{\mathsf{lo}}$. As in the smooth case, if there exist a finite number of singularities with the same coordinate-wise modulus as $(x_0,y_0)$ that all satisfy the conditions of Theorem \ref{prop:ACSVnonsmooth}, then we get an asymptotic expansion by summing the asymptotic contributions of each. 
\subsection{Singularity and Critical Point Analysis} \label{sec:singularity:critical:point}
The main challenge in proving Theorem~\ref{thm:fixed-length} is showing that the prerequisites of Theorems~\ref{prop:ACSVsmooth} and~\ref{prop:ACSVnonsmooth} are fulfilled. We establish the necessary conditions through a careful study of the singularities of our generating functions
\[ \ve{F}_G(x,y) = \frac{1}{1-x} \cdot (\ve{I} - y \ve{P}_G(x))^{-1}\ve{1}^\mathrm{T}. \]
We can write $\ve{F}_G(x,y) = {\ve{Q}_G(x,y)}/{H_G(x,y)}$ for a polynomial vector \mbox{$\ve{Q}_G(x,y) = \adj(\ve{I}-y\ve{P}_G(x)) \ve{1}^\mathrm{T}$} and polynomial $H_G(x,y) = (1-x) \det(\ve{I}-y\ve{P}_G(x))$. In particular, all entries of $\ve{F}_G(x,y)$ share the same denominator, which allows us to analyze crucial properties such as minimality and criticality for singularities just once instead of for each entry. According to Definition~\ref{def:capacity:costly}, we always work with respect to the diagonal $(\alpha_1,\alpha_2) = (1,\alpha)$, and this direction is assumed when discussing notions like critical points and nondegeneracy.

The first step in our multivariate singularity analysis is to identify those singularities which are minimal, i.e., for which there exists no other singularity that has a smaller magnitude in all coordinates (see Definition~\ref{def:minimal:singularity}).
\begin{lemma} \label{lemma:minimal:singularities}
	Let $G$ be a strongly connected and cost-diverse graph with  period $d$ and cost period $c$. The points 
	\begin{align*}
		&\left\{ (x_0,1/\rho_G(x_0)) : 0<x_0<1\right\}  \\
		&\qquad\cup\; \left\{(1,y_0) : y_0 \in \mathbb{C}, |y_0| \leq 1/\rho_G(1)\right\}
	\end{align*}
	are minimal singularities of each coordinate of $\ve{F}_G(x,y)$. All other minimal singularities are
	$$\left(x_0 \mathrm{e}^{\imu 2\pi k/c}, \mathrm{e}^{-2\pi\imu ( kb/c + j/d)}/ \rho_G(x_0) \right)$$
	for some $0<x_0\leq1$, $k\in \{0,1,\dots,c-1\}$, and $j\in \{0,1,\dots,d-1\}$, where $b$ is the constant of the $c$-periodic coboundary condition.
\end{lemma}
\begin{proof}
	The singularities of the coordinates of $\ve{F}_G(x,y)$ are a subset of the solutions to the equation \mbox{$(1-x) \det(\ve{I}-y\ve{P}_G(x))=0$}, and any root of the denominator where the numerator does not vanish is a singularity. Using that $\det(\ve{I}-y\ve{P}_G(x)) = \prod_{j} (1-y\lambda_j(x))$, where $\lambda_j(x)$ are the eigenvalues of $\ve{P}_G(x)$, the singularities of $\ve{F}_G$ are thus a subset of the variety 
	\begin{align*}
	 \mathcal{X} =  &\{(x,1/\lambda_j(x)): x\in\mathbb{C}, 1\leq j\leq |\V| \} %\\&
	 \cup \{ (1,y) : y \in \mathbb{C} \}. 
	 \end{align*}
	We start by investigating the first set of singularities. Right away, we see that for all $x \in \mathbb{C}$ with $|x|>1$ the singularities $(x,1/\lambda_j(x))$ cannot be minimal, since there exists $y\in\mathbb{C}$ such that $(1,y)$ has a coordinate-wise smaller modulus than $(x,1/\lambda_j(x))$. We thus focus on those singularities with $0<|x|\leq1$. Due to the fact that the graph $G$ is strongly connected, it follows that $\ve{P}_G(x_0)$ is irreducible for all $x_0 \in \mathbb{R}^+$ and thus, by the Perron-Frobenius Theorem, has a single real eigenvalue $\rho_G(x_0)$ of maximum modulus. In the following we identify the Perron-Frobenius eigenvalue as the first eigenvalue $\rho_G(x_0) = \lambda_0(x_0)$. 
	
	We now show that for all $0<x_0\leq1$ the points $(x_0,1/\rho_G(x_0))$ are minimal singularities. To begin, the numerator of $\ve{F}_G$ at this point can be expressed as
	\begin{align*}
		\ve{Q}_G(x_0,1/\rho_G(x_0))& = \adj\left(\ve{I}-\frac{\ve{P}_G(x_0)}{\rho_G(x_0)}\right) \mathbf{1}^\mathrm{T} \\
		&\hspace{-3em}= \rho_G(x_0)^{1-|\V|} \adj(\rho_G(x_0)\ve{I}-\ve{P}_G(x_0))\mathbf{1}^\mathrm{T},
	\end{align*}
	so an application of Lemma~\ref{lemma:adj:outer:product} shows that the numerator is nonzero, as $\adj(\rho_G(x_0)\ve{I}-\ve{P}_G(x_0))$ is either all-positive or all-negative.
	In particular, these points are singularities of each coordinate and it remains to show minimality. We prove minimality using Proposition 5.4 of~\cite{melczer_invitation_2021}, which states that a singularity $(x_0,1/\rho_G(x_0))$ with positive coordinates is minimal if and only if $H_G(tx_0,t/\rho_G(x_0))$ is nonzero for all $0<t<1$. The term $(1-tx_0)$ does not vanish for $0<t<1$, so if $H_G(tx_0,t/\rho_G(x_0))=0$ then $t/\rho_G(x_0) = 1/\lambda_j(t x_0)$ for some $0 < t < 1$ and $j \geq 1$. However,
	$$ t/\rho_G(x_0)<1/\rho_G(x_0)\overset{(a)}{\leq} 1/\rho_G(tx_0) \leq |1/\lambda_j(tx_0)|, $$
	where inequality $(a)$ uses that each entry of $\ve{P}_G(x_0)$ is monotonically increasing in $x_0$ and thus $\rho_G(x_0)$ is also monotonically increasing in $x_0$. Hence $H_G(tx_0,t/\rho_G(x_0))$ does not vanish on $0<t<1$ and it follows that any point $(x_0,1/\rho_G(x_0))$ with $0<x_0<1$ is a minimal singularity. 
	
	We next prove that the only other minimal singularities in $\{(x,1/\lambda_j(x)): x\in\mathbb{C}, 1\leq j\leq |\V| \}$ are as given in the statement of the lemma. To start with, by Theorem~\ref{thm:irreducible:spectral:circle}, for each $0<x_0\leq1$ there are precisely $d$ simple eigenvalues $\lambda_0(x_0),\dots,\lambda_{d-1}(x_0)$ with the same modulus as the spectral radius and they are given by
	$$\lambda_j(x_0) = \rho_G(x_0) \mathrm{e}^{2\pi\imu (j-1)/d }.$$
	Due to the similarity of $\ve{P}_G(x\mathrm{e}^{\imu\phi_k})$ and $\mathrm{e}^{\imu\phi_kb}\ve{P}_G(x)$ for all $\phi_k = 2\pi k/c$ and $k \in \{0,1,\dots,c-1\}$, which was derived in Lemma~\ref{lemma:cost-enumerator:complex:circle}, the eigenvalues of $\ve{P}_G(x_0\mathrm{e}^{\imu\phi_k})$ are given by
	$ \lambda_j(x_0 \mathrm{e}^{\imu\phi_k}) = \mathrm{e}^{\imu \phi_k b} \lambda_j(x_0).  $
	Therefore, for each $j$ and $k$ we obtain one candidate for a minimal singularity,
	$$\left(x_0 \mathrm{e}^{\imu\phi_k},\mathrm{e}^{-\imu (\phi_kb + 2\pi j/d) }/\rho_G(x_0) \right).$$
	For all other $\phi$ that are not integer multiples of $2\pi/c$, the singularities $(x_0\mathrm{e}^{\imu\phi},1/\lambda_j(x_0\mathrm{e}^{\imu\phi}))$ are not minimal, as in this case the inequality 
	$\rho_G(x_0\mathrm{e}^{\imu\phi}) < \rho_G(x_0)$  was proven in Lemma~\ref{lemma:spectral:radius:complex:circle}. Furthermore, all other eigenvalues $\lambda_j(x)$ with $j\geq d$ have $|\lambda_j(x)|<\rho_G(x)$, which implies that they cannot be minimal.
	
	Finally, we study the singularities in $\{ (1,y) : y \in \mathbb{C} \}$. All points $(1,y_0)$ for $y_0 \in \mathbb{C}$ and $|y_0| \leq 1/\rho_G(1)$ are singularities, since the matrix $\ve{I}-y_0\ve{P}_G(1)$ is invertible. Furthermore, these singularities are minimal due to the fact that $(1,1/\rho_G(1))$ is minimal as proven above. Conversely, for $|y_0| > 1/\rho_G(1)$ the points $(1,y_0)$ are not minimal due to the existence of the singularities $(x_0,1/\rho_G(x_0))$.
\end{proof}
It is worth noting that, while we have proven that the points $(x_0,1/\rho_G(x_0))$ are indeed singularities, the same is not necessarily true for the other minimal points. This is because, for these points, the numerator is not guaranteed to be nonnegative. Next is a statement on the smoothness and criticality of the singularities.

\begin{lemma} \label{lemma:singularities:critical:smooth}
	Let $G$ be a strongly connected and cost-diverse graph with period $d$ and cost period~$c$. For all $x_0 \in \mathbb{R}^+$ with $x_0\neq 1$ and all $k\in \{0,1,\dots,c-1\}$ and $j \in \{0,1,\dots,d-1\}$, the points
	$$\left(x_0 \mathrm{e}^{ 2\pi\imu k/c}, \mathrm{e}^{-2\pi\imu ( kb/c + j/d)}/ \rho_G(x_0) \right)$$
	are smooth points of $\ve{F}_G(x,y)$, and are critical if and only if $ \alpha x_0 \rho_G'(x_0) = \rho_G(x_0)$. Any point $(1,y_0)$ with $y_0\in \mathbb{C}$ and $|y_0|<1/\rho_G(1)$ is not a root of $\det(\ve{I}-y\ve{P}_G(x))$ and thus is a smooth point that is never critical.
\end{lemma}
\begin{proof}
	Abbreviate for convenience $\phi_k \triangleq 2\pi  k/c$ and $\theta_j\triangleq 2\pi  j/d $. We start by verifying that for all $x_0 \in \mathbb{R}^+$ with $x_0\neq 1$ and $k\in \{0,1,\dots,c-1\}$, $j \in \{0,1,\dots,d-1\}$, the points $(x_0 \mathrm{e}^{\imu \phi_k}, \mathrm{e}^{-\imu ( \phi_kb + \theta_j)}/ \rho_G(x_0) )$ are smooth.
	By Jacobi's Formula, we have
	\begin{align*}
		\frac{\partial H_G(x,y)}{\partial y} &= -(1-x)\tr\left(\adj(\ve{I}-y\ve{P}_G(x)) \ve{P}_G(x)\right).
	\end{align*}
	For the rest of this proof we write $ \lambda_j(x_0\mathrm{e}^{\imu\phi_k})$ for the $d$ eigenvalues of $\ve{P}_G(x_0\mathrm{e}^{\imu\phi_k})$ of maximum modulus, which satisfy $\lambda_j(x_0 \mathrm{e}^{\imu\phi_k}) = \mathrm{e}^{\imu (\phi_k b + \theta_j)} \lambda_0(x_0)$, where $\lambda_0(x_0)=\rho_G(x_0)$ is the Perron root of $\ve{P}_G(x_0)$, according to Theorem \ref{thm:irreducible:spectral:circle} and Lemma~\ref{lemma:cost-enumerator:complex:circle}. The corresponding normalized eigenvectors are $\ve{u}_j(x_0 \mathrm{e}^{\imu\phi_k})$ and $\ve{v}_j(x_0\mathrm{e}^{\imu\phi_k})$, and plugging in the points of interest we obtain
	\begin{align*}
		&\left.\frac{\partial H_G(x,y)}{\partial y}\right|_{\begin{subarray}{l}
				x=x_0 \mathrm{e}^{\imu \phi_k},\\y= 1/\lambda_j(x_0\mathrm{e}^{\imu \phi_k})
		\end{subarray}} \\
		&=-(1-x_0 \mathrm{e}^{\imu \phi_k})\\& \qquad
	\cdot\tr\left(\adj(\ve{I}-\ve{P}_G(x_0 \mathrm{e}^{\imu \phi_k})/\lambda_j(x_0\mathrm{e}^{\imu \phi_k})) \ve{P}_G(x_0 \mathrm{e}^{\imu \phi_k})\right).
	\end{align*}
	Here we can use Lemma \ref{lemma:cost-enumerator:complex:circle} to simplify the cost-enumerator matrix and Lemma \ref{lemma:adj:outer:product} to find an explicit representation of the adjoint matrix, simplifying the above expression to
	\begin{align*}
		&-c_j(x_0)(1-x_0 \mathrm{e}^{\imu \phi_k}) (\lambda_j(x_0))^{1-|\V|} \\&\hspace{1em}
		\cdot \tr (\mathrm{e}^{\imu\phi_kb}\ve{v}_j(x_0) \ve{P}_G(x_0) \ve{u}_j^\mathrm{T}(x_0)) \\
		&= -c_j(x_0)(1-x_0 \mathrm{e}^{\imu \phi_k}) (\lambda_j(x_0))^{1-|\V|} \lambda_j(x_0 \mathrm{e}^{\imu\phi_k}),
	\end{align*}
	where $c_j(x_0) \in \mathbb{R}\setminus \{0\}$ is a nonzero constant. This expression is nonzero for all $x_0\in\mathbb{R}^+$ with $x_0\neq 1$ and $k\in \{0,1,\dots,c-1\}$ and $j \in \{0,1,\dots,d-1\}$, so the points are smooth.
	
	We now examine when these minimal points are solutions of the critical point equations
	\begin{align*}
		\alpha x \frac{\partial H_G(x,y)}{\partial x} &= y \frac{\partial H_G(x,y)}{\partial y}.
	\end{align*}
	The partial derivative of the denominator with respect to $x$ is given by
	\begin{align*}
		\frac{\partial H_G(x,y)}{\partial x} =& - \det(\ve{I}-y\ve{P}_G(x)) - (1-x)y \\&~\cdot\tr\left(\adj(\ve{I}-y\ve{P}_G(x))\frac{\partial \ve{P}_G(x)}{\partial x}\right).
	\end{align*}
	Evaluating this partial derivative at the points $(x_0 \mathrm{e}^{\imu \phi_k}, 1/\lambda_j(x_0\mathrm{e}^{\imu \phi_k}) )$, we obtain
	\begin{align*}
	&	\left.\frac{\partial H_G(x,y)}{\partial x}\right|_{\begin{subarray}{l}
				x=x_0 \mathrm{e}^{\imu \phi_k},\\y= 1/\lambda_j(x_0\mathrm{e}^{\imu \phi_k})
		\end{subarray}} \\
		&=-(1-x_0 \mathrm{e}^{\imu \phi_k})
		\\&\quad
		\cdot\tr\left(\adj(\ve{I}-\ve{P}_G(x_0 \mathrm{e}^{\imu \phi_k})/\lambda_j(x_0\mathrm{e}^{\imu \phi_k})) \ve{P}_G'(x_0 \mathrm{e}^{\imu \phi_k})\right),
	\end{align*} 
	where $\ve{P}_G'(x)$ is the partial derivative of the cost-enumerator matrix with respect to $x$. Here we use that $\det(\ve{I}-y\ve{P}_G(x))$ evaluated at these points is $0$, as $\lambda_j(x_0\mathrm{e}^{\imu \phi_k})$ is an eigenvalue of $\ve{P}_G(x_0\mathrm{e}^{\imu \phi_k})$. Similar to the case of the derivative with respect to $y$, we simplify this expression to
	\begin{align*}
		&-c_j(x_0)(1-x_0 \mathrm{e}^{\imu \phi_k}) (\lambda_j(x_0))^{1-|\V|} \mathrm{e}^{\imu\phi_k(b-1)}
		\\&\quad\quad
		\cdot \tr (\ve{v}_j(x_0) \ve{P}_G'(x_0) \ve{u}_j^\mathrm{T}(x_0)) \\
		\overset{(a)}{=}& -c_j(x_0)(1-x_0 \mathrm{e}^{\imu \phi_k}) (\lambda_j(x_0))^{1-|\V|} 
		%\\&\quad\quad
		\cdot \lambda_0'(x_0)\mathrm{e}^{\imu (\phi_k(b-1) +\theta_j)}
	\end{align*}
	where, in the first step, we used that $\ve{P}_G'(x_0\mathrm{e}^{\imu\phi_k}) = \mathrm{e}^{\imu \phi_k(b-1)}\ve{D}_k^{-1}\ve{P}_G'(x_0)\ve{D}_k$ according to Lemma~\ref{lemma:cost-enumerator:complex:circle}, and equality $(a)$ follows from an application of Lemma~\ref{lemma:diff:perron}.
	Substituting our expressions for the partial derivatives into the critical point equations shows that the critical point equations simplify to
	$ \alpha x_0 \lambda_0'(x_0) = \lambda_0(x_0) $. Since $\lambda_0(x_0) = \rho_G(x_0)$, the first part of the lemma follows.
	
	The singularities $(1,y_0)$ with $|y_0|<1/\rho_G(1)$ are not roots of $\det(\ve{I}-y\ve{P}_G(x))$ as $\rho_G$ is an eigenvalue of $\ve{P}_G$ of largest modulus. Thus, near these points the zero set of the denominator is locally the zero set of the factor $1-x$ and is therefore smooth (algebraically, the partial derivative with respect to $x$ is nonzero at these points). These points can never be critical because the partial derivative of $H_G(x,y)$ with respect to $y$ vanishes at any such point.
\end{proof}
Notice that the derivative $\rho_G'(x)$ in the statement of Lemma~\ref{lemma:singularities:critical:smooth} should crucially be understood with respect to real-valued $x$. The complex derivative does not necessarily exist, since the spectral radius is the largest magnitude of all eigenvalues, \mbox{$\rho_G(x) = |\lambda_0(x)|$}, and the magnitude function is not complex differentiable on the whole complex plane.
\begin{lemma} \label{lemma:alpha:boundary:unique:solution}
	Let $G$ be a strongly connected and cost-diverse graph. Then the critical point equation \mbox{$\alpha x \rho_G'(x) = \rho_G(x)$} has a positive real solution $x_0$ if and only if 
	$$ \lim_{x\to \infty} \frac{\rho_G(x)}{x\rho_G'(x)}<\alpha < \lim_{x\to 0^+} \frac{\rho_G(x)}{x\rho_G'(x)}. $$
	This solution, if it exists, is unique among all positive real $x$. If $\alpha>{\rho_G(1)}/{\rho_G'(1)}$ then $x_0<1$, and if $\alpha<{\rho_G(1)}/{\rho_G'(1)}$ then $x_0>1$. 
\end{lemma}
\begin{proof}
	Since $\rho_G(x)>0$ for $x \in \mathbb{R}^+$, we can rewrite the equation we are trying to solve as $f(x) = 1$, where $f(x) \triangleq \alpha x \rho_G'(x) / \rho_G(x)$. 
	To start we investigate the limit of $f(x)$ as $x\to0^+$. Note that $f(x) >0$ for all $x \in \mathbb{R}^+$. Furthermore, the strict log-log-convexity of $\rho_G(x)$ proven in Lemma \ref{lemma:log:log:convex} implies that $f'(x) > 0$: strict log-log-convexity of $\rho_G(x)$ means that $\log \rho_G(\mathrm{e}^s)$ is strictly convex in~$s$, and substituting $x=\mathrm{e}^s$ gives
	\begin{align*}
		\frac{\partial}{\partial x} f(x) = \mathrm{e}^{-s} \frac{\partial}{\partial s} f(\mathrm{e}^s) &=  \alpha\mathrm{e}^{-s}\frac{\partial}{\partial s} \frac{\mathrm{e}^s \rho_G'(\mathrm{e}^s)}{\rho_G(\mathrm{e}^s)}\\
		&= \alpha\mathrm{e}^{-s}\frac{\partial^2}{\partial s^2} \log \rho_G(\mathrm{e}^s) >  0.
	\end{align*}
	Since $f'(x)>0$ and $f(x)>0$ we see that $f(x)$ is a bounded and decreasing function as $x\to0$ from above, and the monotone convergence theorem implies $\lim_{x\to0^+}f(x)$ exists. Consequently, if $\alpha < \lim_{x\to 0^+} \rho_G(x)/(x\rho_G'(x))$ then $\lim_{x\to0^+} f(x) <1$, as both limits exist. Notice that we allow the upper bound on $\alpha$ to diverge to $\infty$, in which case we can take $\alpha$ as large as desired. This can happen, for example, when there exists a cycle of weight $0$ in $G$. Similarly, the limit $\lim_{x\to \infty}1/f(x)$ exists, as $1/f(x)$ is decreasing and positive. Hence, if $\alpha > \lim_{x\to \infty} \rho_G(x)/(x\rho_G'(x))$ then $\lim_{x\to\infty} 1/f(x) <1$.
	
	To summarize, under our conditions on $\alpha$ we have $\lim_{x\to\infty} f(x)<1$ and $\lim_{x\to0^+} f(x) >1$. By the intermediate value theorem, there is at least one solution to $f(x) = 1$ in $x\in\mathbb{R}^+$. This solution is unique, due to the strict monotonicity of $f$ coming from $f'(x)>0$. We further see that if $\alpha$ is not within these boundaries, there will be no solution in $x \in \mathbb{R}^+$ due to this monotonicity.
	
	If $\alpha \rho_G'(1)>\rho_G(1)$ then $f(1)>1$, and the solution to $f(x)=1$ must occur at $x_0<1$. Similarly, if $\alpha \rho_G'(1)<\rho_G(1)$ then $f(1)<1$ and it follows that $x_0>1$. For a visualization, see Fig.~\ref{fig:critical:point:solution}.
\end{proof}
\input{fig/critical-point-solution}
Another requirement of Theorem~\ref{prop:ACSVsmooth} is the nondegeneracy of the singularities. We prove this in the following.
\begin{lemma} \label{lemma:singularity:nondegenerate}
	Let $G$ be a strongly connected and cost-diverse graph with period $d$ and cost period~$c$. For all $x_0\in \mathbb{R}^+$ and $k\in \{0,1,\dots,c-1\}$ and $j \in \{0,1,\dots,d-1\}$, the points
	$$\left(x_0 \mathrm{e}^{ 2\pi\imu k/c}, \mathrm{e}^{-2\pi\imu ( kb/c + j/d)}/ \rho_G(x_0) \right)$$
	are nondegenerate.
\end{lemma}
\begin{proof}
	Write $\xi_k \triangleq x_0\mathrm{e}^{2\pi\imu k/c}$. According to Theorem~\ref{thm:irreducible:spectral:circle} and Lemma~\ref{lemma:cost-enumerator:complex:circle}, there are $d$ eigenvalues of maximum modulus $\lambda_j(\xi_k)$ of $\ve{P}_G(\xi_k)$ that satisfy $\lambda_j(\xi_k) = \mathrm{e}^{\imu (\phi_k b + \theta_j)} \rho_G(x_0)$, where $\rho_G(x_0)$ is the Perron root of $\ve{P}_G(x_0)$. 
	From the proof of Lemma~\ref{lemma:minimal:singularities} it follows that the 
analytic function $g(x)$ in Definition~\ref{def:nondegenerate} is given by $g(x) = 1/\lambda_j(x)$. The quantity $\mathcal{H}_{x_0,y_0}$ determining nondegeneracy in the smooth case is therefore the second derivative of
	$$ \phi(\theta) = \log \left( \frac{\lambda_j(\xi_k)}{\lambda_j(\xi_k\mathrm{e}^{\imu\theta})} \right) + \frac{\imu \theta}{\alpha} $$
	at $\theta=0$. 
	Differentiating twice with respect to $\theta$ gives
	\begin{align*}
		\left.\frac{\partial^2}{\partial \theta^2} \phi(\theta)\right|_{\theta=0} &\overset{(a)}{=} - \left.\frac{\partial^2}{\partial \theta^2} \log \lambda_j(x_0\mathrm{e}^{\imu\theta})\right|_{\theta=0} \\
		&= \left.\frac{\partial^2}{\partial s^2} \log\lambda_j(\mathrm{e}^s)\right|_{s=\log x_0}\\
		&\overset{(b)}{=} \left.\frac{\partial^2}{\partial s^2} \log\rho_G(\mathrm{e}^s)\right|_{s=\log x_0} \\
		&\overset{(c)}{>} 0.
	\end{align*}
	In $(a)$ we used Lemma~\ref{lemma:cost-enumerator:complex:circle} to conclude that $\lambda_j(\xi_k\mathrm{e}^{\imu\phi}) = \mathrm{e}^{2\pi\imu kb/c} \lambda_j(x_0\mathrm{e}^{\imu\phi})$.
	Note that the differentiation to the left and right hand side of $(b)$ should be understood with respect to complex-valued $s$ and real-valued $s$, respectively, as $\rho_G(\mathrm{e}^s)$ is not complex differentiable in $s$ in general. In $(b)$ we used that for analytic functions, by the definition of complex differentiation, the derivative along the real line equals the complex derivative. Inequality $(c)$ follows from the strict log-log-convexity of $\rho_G(x)$ for $x\in \mathbb{R}^+$, as was proven in Lemma~\ref{lemma:log:log:convex}.
\end{proof}
\subsection{Proof of Theorem~\ref{thm:fixed-length}}
	The final ingredient in the proof of Theorem \ref{thm:fixed-length} is to identify the critical singularities that contribute to the asymptotic expansion, depending on the value of $\alpha$. 
We call a smooth critical singularity  \emph{contributing} if it satisfies the hypotheses of Theorem~\ref{prop:ACSVsmooth} and a non-smooth critical singularity \emph{contributing}  if it satisfies the hypotheses of Theorem~\ref{prop:ACSVnonsmooth}.

	\begin{lemma} \label{lemma:contributing}
		Let $G$ be a strongly connected and cost-diverse graph with  period $d$ and cost period $c$. Let $\lambda_j(x_0)=\mathrm{e}^{2\pi\imu j/d} \rho_G(x_0)$ with $j\in\{0,1,\dots,d-1\}$ denote the $d$ eigenvalues of maximum modulus of $\ve{P}_G(x_0)$.
		\begin{itemize}
			\item If $0 < \alpha < \alpha_G^{\mathsf{lo}}$ then $(1,1/\lambda_j(1))$ for $j\in\{0,1,\dots,d-1\}$ are contributing points. 
			\item If $\alpha_G^{\mathsf{lo}} < \alpha < \alpha_G^{\mathsf{up}}$ then $(x_0\mathrm{e}^{2\pi\imu k/c},1/\lambda_j(x_0\mathrm{e}^{2\pi\imu k/c}))$ for $j\in\{0,1,\dots,d-1\}$ and $k\in \{0,1,\dots,c-1\}$ with $\alpha x_0 \rho_G'(x_0) = \rho_G(x_0)$ are smooth contributing points.
		\end{itemize}
		In both cases, there are no contributing points other than those mentioned.
	\end{lemma}
	\begin{proof}
		We first discuss the multiple-point, non-smooth case $0<\alpha<\alpha_{G}^{\mathsf{lo}}$. We start by proving that $(x_0,y_0) = (1,1/\rho_G(1))$ satisfies the conditions of Theorem \ref{prop:ACSVnonsmooth}. The two surfaces defined by the vanishing of $R(x,y) = 1-x$ and $S(x,y) = \det(\ve{I}-y\ve{P}_G(x))$ intersect at this point. Direct computation shows $R_x(x,y) = -1$ and $R_y(x,y) = 0$, while Jacobi's formula implies $S_x(x,y) = -y\tr(\adj(\ve{I}-y\ve{P}_G(x))\ve{P}_G'(x))$ and $S_y(x,y) = -\tr(\adj(\ve{I}-y\ve{P}_G(x))\ve{P}_G(x))$. Hence, $(1,1/\rho_G(1))$ is a contributing point if there exist $\nu_1,\nu_2 >0$ such that
		\begin{align*}
			&\nu_1 \left(1,\frac{y_0 R_y(x_0,y_0)}{x_0 R_x(x_0,y_0)}\right) + \nu_2 \left(1,\frac{y_0 S_y(x_0,y_0)}{x_0 S_x(x_0,y_0)} \right) \\
		=&\left(\nu_1+\nu_2,\nu_2 \frac{\rho_G(1)}{\rho_G'(1)}\right) = (1,\alpha).
		\end{align*}
		We can set $\nu_1 = 1-\nu_2$ and $\nu_2 = \alpha \rho_G'(1)/\rho_G(1)$, which are both positive due to $\alpha <\rho_G(1)/\rho_G'(1)$, and the required conditions hold. Using the same arguments, the singularities $(1,1/\lambda_j(x_0))$ also contribute to the asymptotics. The remaining singularities $(x,y) \in \mathbb{C}^2$ with the same coordinate-wise modulus $(|x|,|y|) = (1,1/\rho_G(1))$ are smooth, however, by \cite[Cor. 5.6]{melczer_invitation_2021}, none of them are critical as $(1,1/\rho_G(1))$ is not critical by Lemma \ref{lemma:alpha:boundary:unique:solution}.
		
		We now move to the smooth case $\alpha_{G}^{\mathsf{lo}}<\alpha<\alpha_{G}^{\mathsf{up}}$. Lemmas~\ref{lemma:minimal:singularities}--\ref{lemma:singularity:nondegenerate} show that the point $(x_0,1/\rho_G(x_0))$ where $0<x_0<1$ and $\alpha x_0 \rho_G'(x_0) = \rho_G(x_0)$ is unique and a smooth, finitely minimal, critical and nondegenerate singularity. Furthermore, all other singularities with the same coordinate-wise modulus, which are $(x_0\mathrm{e}^{2\pi\imu k/c},1/\lambda_j(x_0\mathrm{e}^{2\pi\imu k/c}))$ for some $k,j \in \mathbb{Z}$, fulfill these properties as well.
	\end{proof}

	We are finally ready to prove Theorem~\ref{thm:fixed-length} by combining Lemmas~\ref{lemma:minimal:singularities}--\ref{lemma:contributing} with Theorems~\ref{prop:ACSVsmooth} and \ref{prop:ACSVnonsmooth}.
	\begin{proof}[Proof of Theorem \ref{thm:fixed-length}]
		We differentiate between the two cases $0<\alpha<\alpha_{G}^{\mathsf{lo}}$ and $\alpha_{G}^{\mathsf{lo}}<\alpha<\alpha_{G}^{\mathsf{up}}$. In the first case, the non-smooth singularity $(1,1/\rho_G(1))$ and those with the same coordinate-wise moduli are the singularities that determine the asymptotic behavior. In the second case, the singularities $(x_0,1/\rho_G(x_0))$ with $0<x_0<1$ and $\alpha x_0\rho_G'(x_0) = \rho_G(x_0)$, and those with the same coordinate-wise moduli, are the ones contributing.
		
		We start with the multiple-point, non-smooth case $0<\alpha<\alpha_{G}^{\mathsf{lo}}$, aiming to apply Theorem~\ref{prop:ACSVnonsmooth} with the extension \cite[Cor. 9.1]{melczer_invitation_2021}. For any $x_0 \in\mathbb{R}^+$ let $\lambda_j(x_0)=\mathrm{e}^{2\pi\imu j/d} \rho_G(x_0)$ with $j\in\{0,1,\dots,d-1\}$ denote the $d$ eigenvalues of maximum modulus of $\ve{P}_G(x_0)$ and let $\ve{u}_j(x_0)$ and $\ve{v}_j(x_0)$ be the corresponding right and left eigenvectors. By Lemma~\ref{lemma:contributing}, Theorem~\ref{prop:ACSVnonsmooth} is applicable for the contributing singularities $(1,1/\lambda_j(1))$ and it remains to compute the required terms. The numerator of the generating function is given by
		\begin{align*}
			\ve{Q}_{G,v}(1,1/\lambda_j(1)) &= \lambda_j(1)^{1-|\V|} \adj(\lambda_j(1) \ve{I} - \ve{P}_G(1)) \mathbf{1}^\mathrm{T}  \\
			&\overset{(a)}{=}c_j(1) \lambda_j(1)^{1-|\V|}  \ve{u}_j^\mathrm{T}(1)\ve{v}_j(1) \mathbf{1}^\mathrm{T},
		\end{align*}
		where $(a)$ follows from an application of Lemma \ref{lemma:adj:outer:product}. Similarly, we obtain for the numerator
		\begin{align*}
			\det \mathbf{H} &= - \frac{1}{\lambda_j(1)} S_y(1,1/\lambda_j(1)) \\
			&=   \frac{1}{\lambda_j(1)} \tr (\adj(\ve{I}-\ve{P}_G(1)/\lambda_j(1))\ve{P}_G(1)) \\
			&= c_j(1) \lambda_j(1)^{1-|\V|}.
		\end{align*}
		Plugging these results into the expressions of Theorem \ref{prop:ACSVnonsmooth} and summing over all contributing points $(1,1/\lambda_j(1))$ according to \cite[Cor. 9.1]{melczer_invitation_2021} proves the first statement of Theorem \ref{thm:fixed-length}.

		In the smooth case $\alpha_{G}^{\mathsf{lo}}<\alpha<\alpha_{G}^{\mathsf{up}}$, the point $(x_0,1/\rho_G(x_0))$ where $0<x_0<1$ and $\alpha x_0 \rho_G'(x_0) = \rho_G(x_0)$ is unique and a smooth, finitely minimal, critical and nondegenerate singularity and thus contributing by Lemma~\ref{lemma:contributing}. The same applies to the other singularities with the same coordinate-wise modulus, which are $(x_0\mathrm{e}^{2\pi\imu k/c}, \mathrm{e}^{-2\pi \imu(kb/c+j/d)}/\rho_G(x_0))$ for some $k,j \in \mathbb{Z}$. This allows us to invoke the extension \cite[Cor. 5.2]{melczer_invitation_2021} of Theorem \ref{prop:ACSVsmooth}. Notice that there may be values of $j$ and $k$ where the numerator vanishes, however we have shown in Lemma \ref{lemma:minimal:singularities} that this does not occur when $k=j=0$. Thus, it is possible that the leading asymptotic terms from some of these points vanishes, but the sum of all terms always captures the dominant asymptotic behavior of the sequence under consideration. The quantity $\mathcal{H}$ appearing in the asymptotic expansion was derived in Lemma \ref{lemma:singularity:nondegenerate}. Abbreviating $\phi_k = 2\pi k/c$ in the following, we find that
		\begin{align*}
			&\ve{Q}_{G,v}(x_0\mathrm{e}^{\imu\phi_k},1/\lambda_j(x_0\mathrm{e}^{\imu\phi_k})) \\&~=c_j(x_0) \lambda_j(x_0)^{1-|\V|} \ve{D}_k^{-1} \ve{u}_j^\mathrm{T}(x_0)\ve{v}_j(x_0) \ve{D}_k \mathbf{1}^\mathrm{T},
		\end{align*}
		where we used that for any two square matrices $\ve{D}$ and $\ve{P}$, the adjoint of the conjugation of $\ve{P}$ by  $\ve{D}$  is given by $\adj(\ve{D}^{-1}\ve{P}\ve{D}) = \ve{D}^{-1}\adj(\ve{P})\ve{D}$.
\end{proof}

\subsection{Proof of the Other Theorems}

We continue with proving the remaining theorems.

\begin{proof}[Proof of Theorem~\ref{thm:fixed-length:capacity}]
	Theorem~\ref{thm:fixed-length:capacity} directly follows from Theorem~\ref{thm:fixed-length}. By the definition of the capacity, we take the logarithm of the asymptotic expansion $N_{G,v}(t,\alpha t)$ and divide by $t$. Computing the limit $t\to\infty$, all terms except for the exponential in $t$ vanish.
\end{proof}

Theorems~\ref{thm:variable-length:capacity} and \ref{thm:variable-length:exact} can be proven using standard univariate singularity analysis \cite{flajolet_analytic_2009}. We start with proving  Theorem~\ref{thm:variable-length:exact}, which depicts the more general statement of the exact representation of the follower set size.

\begin{proof}[Proof of Theorem~\ref{thm:variable-length:exact}]
	By Lemma~\ref{lemma:generating:function}, the generating functions of $N_{G,v}(t)$ are given by the fractions \mbox{$F_{G,v}(x) = Q_{G,v}(x)/H_G(x)$}, with the polynomials $Q_{G,v}(x) = [\adj(\ve{I}-\ve{P}(x))\mathbf{1}^\mathrm{T}]_v$ and $H_G(x) = (1-x)\det(\ve{I}-\ve{P}_G(x))$, and thus the singularities are a subset of the solutions to $(1-x)\det(\ve{I}-\ve{P}_G(x)) = 0$. Invoking \cite[Thm. IV.9]{flajolet_analytic_2009} then proves Theorem~\ref{thm:variable-length:exact}. Note that in principle not all solutions have to be singularities, as the numerator and denominator are not guaranteed to be coprime. This case is covered by setting $\Pi_{G,v,i}(t) = 0$ for all roots which share common factors with the numerator in the partial fraction decomposition.
\end{proof}

\begin{proof}[Proof of Theorem~\ref{thm:variable-length:capacity}]
	The theorem follows directly  from Theorem~\ref{thm:variable-length:exact} by the computation of $C_G = \lim_{t\to\infty}{\log N_{G,v}(t)}/{t}$ and the fact that  $H_G(x) = (1-x) \prod_{j}(1-\lambda_j(x))$, where $\lambda_j(x)$ are the eigenvalues of $\ve{P}_G(x)$. Since $\ve{P}_G(x)$ is an irreducible matrix, there is an eigenvalue  which is equal to the spectral radius and thus the singularity of smallest magnitude of $F_{G,v}(x)$ is that for which $\rho_G(x) = 1$. The numerator at this singularity is nonzero due to Lemma~\ref{lemma:adj:outer:product}.
\end{proof}

\subsection{Details of Remarks~\ref{rem:alpha:comb:interp} and~\ref{rem:maximality:concavity}}
\label{sec:details:remarks}

We now return to Remarks~\ref{rem:alpha:comb:interp} and~\ref{rem:maximality:concavity}.

\begin{proposition}\label{prop:alpha:comb:interp}
The inverse of $\alpha_{G}^{\mathsf{lo}}$ is the average cost per edge, asymptotically in $n$, over all paths of length $n$ in $G$. Equivalently, it is the average cost per edge associated with the unique stationary Markov chain of maximum entropy on $G$. The inverse of $\alpha_{G}^{\mathsf{up}}$ is the minimum average cost per edge among the cycles in $G$. 
\end{proposition}

\begin{proof}\footnote{The authors thank Andrew Tan for helpful comments regarding this proof.}
 Recall the definitions $\alpha_{G}^{\mathsf{lo}} \triangleq {\rho_G(1)}/{\rho_G'(1)}$ and $\alpha_{G}^{\mathsf{up}} \triangleq \lim\limits_{x\to 0^+} {\rho_G(x)}/{(x\rho_G'(x))}$.  Define the sequence of functions 
$$f_m(x) \triangleq \frac{1}{m} \log \left(\sum\limits_{u,v}(\ve{P}_G(x)^m)_{u,v}\right),$$  
where $(\ve{P}_G(x)^m)_{u,v}$ is the generating function for the costs of  paths $\ve{p}$ of length $m$ from state $u$ to state $v$, i.e., 
\begin{equation*}
  (\ve{P}_G(x)^m)_{u,v} = \sum_{\begin{array}{c}   \ve{p}:u \rightarrow v \\ {\rm length} \; \ve{p}=m \end{array}} x^{\tau(\ve{p})}.
\end{equation*}
 
For $x > 0$, the cost-enumerator matrix $\ve{P}_G(x)$ is irreducible, so by~\cite[Lemma 3.5]{marcus_introduction_2001}
\begin{equation}
    \label{eq:growth_rate}
    \lim_{m\rightarrow\infty} f_m(x)= \log\rho_G(x).
\end{equation}

The corresponding sequence of derivatives $f_m'(x)$ is given by
 \begin{equation*}
 f_m'(x) = \frac{1}{m} \frac{\frac{d}{dx} \sum\limits_{u,v}(\ve{P}_G(x)^m)_{u,v}}{\sum\limits_{u,v}(\ve{P}_G(x)^m)_{u,v}} =   \frac{1}{m} \frac{\sum_{\ve{p}} \tau(\ve{p}) x^{\tau(\ve{p})-1}}{\sum_{\ve{p}} x^{\tau(\ve{p})}}.
 \end{equation*}
It can then be shown that $f_m'(x) \xrightarrow{m \rightarrow \infty} \frac{d}{dx} \log \rho_G(x)$  uniformly on a closed interval in $(0,\infty)$, i.e., the sequence of derivatives converges uniformly to the derivative of the limit. Thus, the derivative and limit can be interchanged, so that
\begin{equation}
    \label{eq:growth_rate_derivative}
    \lim_{m\rightarrow\infty} f_m'(x) = \left( \lim_{m \rightarrow \infty} f_m(x) \right)' = \left( \log \rho_G (x) \right)' = \frac{\rho'_G(x)}{\rho_G(x)}.
\end{equation}
The proof of uniform convergence of the sequence of derivatives makes use of the following fact; see~\cite[proof of Lemma 3.17]{marcus_introduction_2001}.

\medskip
\noindent
{\bf Fact:}
Let $\ve{P}$ be a primitive matrix, i.e., $\ve{P}$ is irreducible and has period $1$. Then, there exist real vectors $\ve{u},\ve{v}$ such that $\ve{P}\ve{u} = \rho \ve{u}, \ve{v}^T \ve{P} = \rho \ve{v}^T$, and $\ve{v}^T \ve{u} = 1$. Moreover, $\lim\limits_{m\rightarrow \infty} \frac{\ve{P}^m}{\rho^m} = \ve{u}\ve{v}^T$ and $\ve{P}^m = \rho^m \ve{u}\ve{v}^T + \ve{E}^{(m)}$ with $|\ve{E}^{(m)}_{s,t}| = O(m^{h-1} \mu^m)$, where $\rho$ is the spectral radius of $\ve{P}$, $\mu$ is the largest absolute value of the eigenvalues of $\ve{P}$ other than $\rho$, and $h$ is the highest algebraic multiplicity of the eigenvalues of $\ve{P}$ whose absolute value is $\mu$. \qed \\

The derivative $f_m'(x)$ evaluated at $x=1$ is
\begin{equation*} \label{eq:derivative}
\begin{aligned}
f_m'(1)&=  
     \frac{1}{m} \frac{\sum_{\ve{p}} \tau(\ve{p}) x^{\tau(\ve{p})-1}}{\sum_{\ve{p}} x^{\tau(\ve{p})}}\Biggr\vert_{x=1}\\
   &=\frac{\sum_{\ve{p}} \tau(\ve{p})}{N(m)} \triangleq T_{\text{ave},m} , 
    \end{aligned}
\end{equation*}
where the sums in the numerator and denominator of the expressions on the right run over all paths $\ve{p}$ of length $m$ in $G$ and
$N(m)$ denotes the total number of such paths.
Since $\tau(\ve{p})/m$ is the average edge cost of the path $\ve{p}$, the derivative evaluated at $x=1$ equals $T_{\text{ave},m}  $, the average edge cost over all paths of length $m$ in $G$. Setting $x=1$ in~(\ref{eq:growth_rate_derivative}), we obtain
\begin{equation*}
   T_{\text{ave}}   \triangleq \lim_{m\rightarrow \infty}\frac{\sum_{\ve{p}} \tau(\ve{p})/m}{N(m)} = \frac{\rho'_G(1)}{\rho_G(1)} = (\alpha_{G}^{\mathsf{lo}})^{-1},
\end{equation*}
where we can interpret $T_{\text{ave}}$ as the asymptotic average cost per edge over all paths in $G$.   Noting that the Markov chain of maximum entropy $H_{max}$ on $G$ assigns probability approximately $2^{-mH_{max}}$ to each path in $G$, we see that this is also the average cost per edge with respect to this Markov chain. 

Turning to $\alpha_{G}^{\mathsf{up}}$, we note from~\cite[Theorem 3.17]{marcus_introduction_2001} that an expression similar to (\ref{eq:growth_rate}) applies to the cycles at each state $u$ in $G$, namely 
\begin{equation}
    \label{eq:cycle_growth_rate}
   \lim_{m\rightarrow\infty} \frac{1}{m} \log  (\ve{P}_G(x)^m)_{u,u}  = \log\rho_G(x).   
\end{equation}
The matrix element $(\ve{P}_G(x)^m)_{u,u}$ is the generating function for the cost of cycles of length $m$ at vertex $u$. So the logarithm in the argument of the limit can be written as 
\begin{equation*}
\label{eq:log_sum}
  \log (\ve{P}_G(x)^m)_{u,u} = \log\left(\sum_{\text{length-}m\; \text{cycles}\;\ve{p}\;\text{at}\;u} x^{\tau(\ve{p})}\right).
\end{equation*}
Let $T_{\min,m}  (u)$ denote the minimum average cost per edge over cycles of length $m$ at $u$. We note that  the limit $T_{\min}   = \liminf_{m\rightarrow \infty}T_{\min,m}  (u)$ is independent of $u$ and equals the minimum average cost per edge in a simple (non-intersecting) cycle in $G$. 

Rewriting $x=e^{\log(x)}$  and applying the log-sum-exp inequality to the resulting expression, we obtain
\begin{equation*}
    \label{eq:log_sum_exp}
    \begin{aligned}
    &\max_{\text{cycles}\;\ve{p}}\{\log(x)\tau(\ve{p})\}\leq \log (\ve{P}_G(x)^m)_{u,u} \\
    &\leq  
    \max_{\text{cycles}\;\ve{p}}\{\log(x)\tau(\ve{p})\} + \log (\ve{P}_G(1)^m)_{u,u},
    \end{aligned}
\end{equation*}
where the term on the right is the logarithm of the number of cycles of length $m$ at $u$. 

We are interested in the limit as $x\rightarrow 0^+$ so, restricting to the range $0 <x<1$ where $\log(x)<0$, the maximum evaluates to $\log(x)T_{\min,m}  $, yielding 
\begin{equation*}
    \label{eq:log_sum_exp_eval}
    \begin{aligned}
    \log(x)mT_{\min,m}   &\leq \log (\ve{P}_G(x)^m)_{u,u} \\
    &\leq  
   \log(x)mT_{\min,m}   + \log (\ve{P}_G(1)^m)_{u,u}.
   \end{aligned}
\end{equation*}
Normalizing by $m$ and taking the limit, we invoke 
(\ref{eq:cycle_growth_rate}) to conclude that 
\begin{equation*}
  \log(x)T_{\min}   \leq  \log \rho_G(x) \leq 
   \log(x)T_{\min}   +\log \rho_G(1).
 \end{equation*}
We divide by $\log(x)$ and take the limit as $x\rightarrow 0^+$. Recalling  that $\rho_G(x)\rightarrow 0$ as $x \rightarrow 0^+$, we apply  l'H\^{o}pital's rule to conclude 
 \begin{equation*}
     T_{\min}   = \lim_{x\rightarrow0^+} \frac{x\rho_G'(x)}{\rho_G(x)} = (\alpha_{G}^{\mathsf{up}})^{-1}.
 \end{equation*}
\end{proof}

\begin{proposition} \label{prop:maximality:concavity}
	Let $G$ be a strongly connected, deterministic, cost-diverse graph. Then $C_G(\alpha)$ is a concave function in $\alpha$ and its maximum is equal to $C_G(\alpha^*) = C_G$, where $\alpha^*=2^{C_G}/\rho_G'(2^{-C_G})$. 
	\end{proposition}
\begin{proof}
	To start with, $C_G(\alpha)$ is linear in the interval $0\leq \alpha < \alpha^{\mathsf{lo}}_G$. In the interval $\alpha^{\mathsf{lo}}_G\leq \alpha < \alpha^{\mathsf{up}}_G$, $C_G(\alpha) = -\log x_0(\alpha) +\alpha \log \rho_G(x_0(\alpha)) $, where $x_0(\alpha)$ is the unique positive solution to $f(x) = \alpha^{-1}$ with $f(x) \triangleq x \rho_G'(x)/\rho_G(x)$. Notice that $\rho_G(x)>0$ for all $x>0$ and thus, by Lemma~\ref{lemma:spectral:radius:analytic}, $f(x)$ is analytic for all $x>0$. Furthermore, as in the proof of Lemma~\ref{lemma:alpha:boundary:unique:solution}, we can show that $\frac{\partial}{\partial x}f(x)>0$, which means that $x_0(\alpha)$ is analytic in $\alpha$ and also strictly monotonically decreasing in $\alpha$. Therefore, for $\alpha^{\mathsf{lo}}_G\leq \alpha < \alpha^{\mathsf{up}}_G$,
	\begin{align*}
		\frac{\partial}{\partial \alpha} C_G(\alpha) &= -\frac{x_0'(\alpha)}{x_0(\alpha)} +\log \rho_G(x_0(\alpha)) 
		%\\&\hspace{0.25in}
		+ \alpha \frac{x_0'(\alpha)\rho_G'(x_0(\alpha))}{\rho_G(x_0(\alpha))} \\
		 &= -\frac{x_0'(\alpha)}{x_0(\alpha)} +\log \rho_G(x_0(\alpha)) 
		 %\\&\hspace{0.25in}
		 + \alpha f(x_0(\alpha)) \frac{x_0'(\alpha)}{x_0(\alpha)} \\
		 &\overset{(a)}{=}	 \log \rho_G(x_0(\alpha)),
	\end{align*}
	where we used in $(a)$ that $f(x_0(\alpha)) = \alpha^{-1}$ by definition. Since $x_0(\alpha)$ is strictly monotonically decreasing in $\alpha$ and also $\rho_G(x)$ and the logarithm are strictly monotone functions (see Lemma~\ref{lemma:diff:perron}), $C_G(\alpha)$ is strictly concave in the considered interval. Then $C_G(\alpha) \to \alpha^{\mathsf{lo}}_G \log \rho_G(1)$ by the definition of $\alpha^{\mathsf{lo}}_G$, as $\alpha$ approaches $\alpha^{\mathsf{lo}}_G$ from both the left and right, proving continuity of $C_G(\alpha)$. Thus $C_G(\alpha)$ is a concave function on the full interval $0 \leq \alpha \leq\alpha^{\mathsf{up}}_G$. 
	
	From the above derivation of the derivative of $C_G(\alpha)$, we further see that $\alpha^*$ with $\rho_G(x_0(\alpha^*)) = 1$ is a unique stationary point of $C_G(\alpha)$ with capacity $C_G(\alpha^*) = - \log x_0(\alpha^*) = C_G$. It follows that $\alpha^*$ is the unique solution for $\alpha$ to the system of two equations $f(x)=\alpha^{-1}$ and $\rho_G(x) = 1$ with $x>0$. The above exposition proves that the solution to these equations are $\alpha^*$ and $x_0(\alpha^*)$, where $\alpha^*$ is as given in the statement.
	\end{proof}

%% file: fig/critical-point-solution.tex
\begin{figure}
	\centering
	\begin{tikzpicture}
		\begin{axis}[%
			width=0.85\linewidth,
			height=8cm,
			xticklabel style={
				/pgf/number format/fixed,
				/pgf/number format/precision=5
			},
			scaled x ticks=false,
			grid=both,
			grid style={line width=.1pt, draw=gray!10},
			major grid style={line width=.2pt,draw=gray!50},
			minor tick num=4,
			xtick distance = 0.5,
			ytick distance = 0.5,
			xmin=0,
			xmax=3,
			ymin=0,
			ymax=2,
			xlabel = {$x$},
			ylabel = {${\alpha x \rho_G'(x)}/{\rho_G(x)}$},
			legend pos = south east,
			legend cell align={left},
			extra x ticks = {0.5,2},
			extra x tick labels = {$x_0$,$x_0'$},
			extra x tick style = {ticklabel pos=top},
			]
			
			\addplot [domain=0.01:3,samples=150]({x},{0.75*(2*x^2+x)/(x^2+x)});
			\addplot [dashed,domain=0.01:3,samples=150]({x},{0.6*(2*x^2+x)/(x^2+x)});
			\addplot [thick,densely dotted,domain=0.01:3,samples=150]({x},{0.35*(2*x^2+x)/(x^2+x)});
			\addplot [thick,densely dotted,domain=0.01:3,samples=150]({x},{1.1*(2*x^2+x)/(x^2+x)});
			\addlegendentry{$0<x_0<1$};
			\addlegendentry{$1<x_0$};
			\addlegendentry{No solution};
			\addplot [mark=*, only marks] coordinates {(2,1) (0.5, 1)};
			\addplot [opacity=0.25,domain=0:3,samples=2]({x},{1});
			\addplot [opacity=0.25] coordinates {(1,0) (1,2)};
			%\node[] at (axis cs:0.5, 1.1) {$x_0$};
			%\node[] at (axis cs:2, 1.1) {$x_0$};
		\end{axis}
	\end{tikzpicture}%
	\caption{Visualization of the solutions of the critical point equation for $\rho_G(x) = x+x^2$. The illustrated cases correspond to different values of $\alpha$.}
	\label{fig:critical:point:solution}
\end{figure}

%% file: conclusion.tex
%!TEX root = bare_jrnl.tex
\section{Conclusion}
In this paper we have analyzed costly constrained channels, i.e., directed graphs with labeled and weighted edges. We have derived the precise asymptotic behavior of the size of the number of limited-cost paths for arbitrary strongly connected and cost-diverse graphs. That is, we have explicitly derived an easily computable function, whose ratio with respect to the true number of followers approaches one for large costs. Our theorems imply explicit expressions for the fixed-length and variable-length capacity, i.e., the exponential growth rate of the number of paths. Interestingly, through the direct derivation of the capacity, we recover a known result on the equivalence of the combinatorial and probabilistic capacity of costly constrained channels subject an average cost constraint.  While previous works have shown this equivalence using a typical sequence argument and converse inequalities, this proof yields an expression for the combinatorial capacity in terms of singularities of a generating function that matches the known formula for the probabilistic capacity obtained by Markov chain analysis.

Establishing an explicit and comprehensive framework to compute both the fixed-length and variable capacity for arbitrary strongly connected graphs, our results  not only open the way for future research but can also directly be employed in suitable applications. For our derivations, we have extended the well-known notions of periodicity to weighted graphs. We show that the notion of cost-diversity is the precise property that distinguishes between degenerate and smooth behavior of the fixed-length capacity. In our exposition we use results from analytic combinatorics in several variables, which establishes novel and intriguing connections between noiseless information theory and complex analysis. In order to prove these connections we have built a comprehensive theory that extends results from the Perron-Frobenius theory of irreducible matrices. These results were then related to properties of the singularities of the generating functions of the follower set size, which built the bridge to the theory of analytic combinatorics in several variables.

We illustrated our capacity results by analyzing the discrete noiseless channels describing the synthesis of $q$-ary sequences using the $q$-ary alternating sequences. The  case $q=4$ is particularly relevant to the synthesis of DNA strands.  Our framework can be extended to the analysis of maximum achievable synthesis rates for general synthesis sequences and synthesis of constrained sequences, as well as to the enumeration of subsequences of a given supersequence. These extensions are discussed in~\cite{lenz_thesis_2022}.

%% file: appendix.tex
%!TEX root = bare_jrnl.tex
\section{Auxiliary Results} \label{appendix}

\begin{lemma} \label{lemma:gcd:cycle:lengths}
	Let $G = (\V,\E,\sigma,\tau)$ be a strongly connected graph with  period $d$, as defined in Definition~\ref{def:period}. Then the greatest common divisor of all cycle lengths is $d$.
\end{lemma}
\begin{proof}
	To start with, we note that the length $m$ of each cycle must be divisible by $d$, since otherwise the twice repetition of this cycle would not have a length congruent to that of the single cycle modulo $d$. Analogously to the proof of Lemma~\ref{lemma:cost:period:existence:paths}, we can prove the existence of two cycles at the same state whose lengths differ by precisely $d$. This implies that the the greatest common divisor of the cycle lengths is $d$, which proves the statement.
\end{proof}
\begin{lemma} \label{lemma:var:crt}
	Let $c\in \N$ and $(m_1,\tau_1),(m_2,\tau_2), \dots$ be pairs of integers $ (m_i,\tau_i)\in\N^2 $. Denote by $d$ the greatest common divisor of all $m_i$. If these pairs satisfy
	$$ m_i\tau_j \equiv m_j\tau_i \pmod{(c d)} $$
	for all $i$ and $j$ then there exists $b \in \mathbb{Z}$ such that, for all $i$,
	$$ d\tau_i \equiv m_i b  \pmod{(cd)}. $$
	Furthermore, any $b' \in \mathbb{Z}$ with $ b' \equiv b \pmod c$ has the same property.
\end{lemma}
\begin{proof}
	We prove the statement by a direct construction. Assume without loss of generality that $\gcd(m_1,\dots,m_n) = d$ for some $n\in\N$. This is possible, since there exist finitely many $m_i$ such that their greatest common divisor is equal to $d$. By Bézout's identity, there exist $z_1,\dots,z_n \in \mathbb{Z}$ with $z_1m_1+\dots+z_nm_n = d$. Choosing $ b = z_1\tau_1+\dots z_n\tau_n$, we obtain for any $1\leq i \leq n$,
	\begin{align*}
		m_ib &= z_im_i\tau_i + \sum_{j\neq i} z_jm_i \tau_j \\
		&= \tau_i\left(d- \sum_{j\neq i} z_jm_j\right)+ \sum_{j\neq i} z_jm_i \tau_j \\
		&= \tau_id + \sum_{j\neq i} z_j(m_i\tau_j-m_j\tau_i).
	\end{align*}
	By assumption $m_i\tau_j-m_j\tau_i \equiv 0 \pmod{(cd)}$ and thus $m_ib \equiv \tau_id \pmod{(cd)}$. On the other hand, for any $i>n$, we set $z_i = 0$ and obtain via a similar argument
	$$ m_ib = z_im_i\tau_i + \sum_{j=1}^n z_jm_i \tau_j = \tau_i d + \sum_{j=1}^n z_j(m_i \tau_j-m_j\tau_i), $$
	which implies that $m_ib \equiv \tau_id \pmod{(cd)}$. This concludes the proof.

\end{proof}